\documentclass[reqno]{amsart}
\usepackage[english]{babel}
\usepackage[utf8]{inputenc}
\usepackage[T1]{fontenc}
\usepackage{amsfonts,amssymb,amsmath, amsthm}
\usepackage{mathrsfs}
 \usepackage[normalem]{ulem}

\usepackage{graphicx}
\usepackage{subcaption}
\usepackage{tikz}
\usepackage{cancel}
\usepackage{comment}
\usepackage{bbm}
\usepackage{color} 
\definecolor{bleu_sombre}{rgb}{0,0,0.6} 
\definecolor{bs}{rgb}{0,0,0.6}  
\definecolor{rouge_sombre}{rgb}{0.8,0,0}
\definecolor{rs}{rgb}{0.8,0,0}
\definecolor{vert_sombre}{rgb}{0,0.6,0}
\definecolor{vs}{rgb}{0,0.6,0}
\usepackage[plainpages=false,colorlinks,linkcolor=bleu_sombre,citecolor=rouge_sombre,urlcolor=vert_sombre,breaklinks]{hyperref}

\usepackage{enumerate}

\theoremstyle{plain} 
\newtheorem{theorem}{Theorem}[section]
\newtheorem{lemma}[theorem]{Lemma}
\newtheorem{coro}[theorem]{Corollary}
\newtheorem{proposition}[theorem]{Proposition}

\newtheorem*{Assumption}{Condition}

\theoremstyle{definition}
\newtheorem{remark}[theorem]{Remark}

\newtheorem{definition}[theorem]{Definition}

\newcommand{\iu}{{\rm i}}
\newcommand{\Psitp}{ \Psi}
\newcommand{\DETAILS}[1]{}

\newcommand{\supp}{\operatorname{supp}}

\newcommand{\Ran}{\operatorname{Ran}}

\renewcommand{\Re}{\mathrm{Re}}

\newcommand{\C}{\mathbb{C}}
\newcommand{\R}{\mathbb{R}}

\newcommand{\N}{\mathbb{N}}

\newcommand{\cB}{{\mathcal{B}}}

\newcommand{\cH}{{\mathcal{H}}}

\newcommand{\cM}{{\mathcal{M}}}

\newcommand{\Q}{{Q}}

\renewcommand{\d}{\mathrm{d}}

\newcommand{\SO}{\mathrm{SO}}
\newcommand{\cE}{\mathcal{E}}
\newcommand{\cI}{\mathcal{I}}

\newcommand{\cW}{\mathcal{W}}

\newcommand{\cZ}{\mathcal{Z}}

\newcommand{\ii}{\infty}
\newcommand{\eps}{{\varepsilon}}

\newcommand{\U}{{U}}
\newcommand{\V}{{V}}

\newcommand{\h}{\theta}

\newcommand{\lk}{L_{\mathrm{cut}}}

\makeatletter
\newcommand{\customlabel}[2]{%
   \protected@write \@auxout {}{\string \newlabel {#1}{{#2}{\thepage}{#2}{#1}{}} }%
   \hypertarget{#1}{#2} %
}
\makeatother

\renewcommand{\leq}{\leqslant}	
\renewcommand{\geq}{\geqslant}

\newcommand{\cF}{\mathcal{F}}
\newcommand{\cG}{{\mathcal{G}}}
\newcommand\pscal[1]{{\ensuremath{\left\langle #1 \right\rangle}}}

\def\R{{\mathbb R}}
\def\N{{\mathbb N}}

\def\Hilb{{\mathcal H}}
\def\Dom{{\mathcal D}}

\def\({\left(}
\def\){\right)}
\def\<{\left\langle}
\def\>{\right\rangle}

\def\eps{\varepsilon}

\numberwithin{equation}{section}

\newcommand{\dd}{\mathrm{d}}

\newcommand{\ut}{u_{\tau}}
\newcommand{\mt}{m_{\tau}}

\newcommand{\elt}{l_{\tau}}
\newcommand{\nt}{n_{\tau}}
\newcommand{\uti}{u_{\tau}^{-1}}
\newcommand{\dist}{\mathrm{dist}}
\newcommand{\tb}{\tau_0}
\newcommand{\taut}{{\tau_t}}

\newcommand{\Psitb}{\Psi_{\tb}}
\newcommand{\be}{\begin{equation}}
\newcommand{\ee}{\end{equation}}
\newcommand{\bea}{\begin{eqnarray}}
\newcommand{\eea}{\end{eqnarray}}
\newcommand{\bee}{\begin{eqnarray*}}
\newcommand{\eee}{\end{eqnarray*}}

\def\eps{\varepsilon}

\usepackage{geometry}
\title{On Boundedness of  Isomerization paths for Non- and Semirelativistic Molecules}
\author{Ioannis Anapolitanos}
\author{Marco Olivieri}
\author{Sylvain Zalczer}

\begin{document}








\maketitle
\begin{abstract}
	This article focuses on isomerizations of molecules, \textit{i.e.} chemical reactions during which a molecule is transformed into another one with atoms in a different spatial configuration. We consider the special case in which the system breaks into two submolecules whose internal geometry is solid during the whole procedure. We prove, under some conditions, that the distance between the two submolecules stays bounded during the  reaction.  To this end, we provide an asymptotic expansion of the interaction energy between two molecules, including multipolar interactions and the van der Waals attraction. In addition to this static result, we proceed to a quasistatic analysis to investigate the variation of the energy when the nuclei move.  This paper extends \cite{al} in two directions. The first one is that we relax  the assumption that the ground state eigenspaces of the submolecules have  to fulfill. The second one is that we allow semirelativistic kinetic energy as well.
\end{abstract}
\tableofcontents

\bigskip

An isomerization is a chemical reaction with the property that the reactant is a single molecule with the same atoms as the product, but in a different spatial configuration. One simple example is the isomerization HCN $\rightarrow$ CNH. Isomerizations play, especially in organic and organometallic chemistry,  an important role,  see \emph{e.g.} \cite{isom-organo}. The question of how much energy is needed for an isomerization to take place is very fundamental.  The study of it occupies a large amount of the numerical computations in quantum chemistry, see for example~\cite{GSK-07} and references therein.

 As explained in~\cite[Chapter~17]{jensen}, real world chemical systems, involving a large number of particles, are very complicated and some approximations have to be made. We use here the Born-Oppenheimer approximation: we assume that electronic and nuclear motions occur at different time scales, so that one can first solve the Schrödinger equation for electrons with fixed nuclei, and then study the dynamics of the nuclei with a potential corresponding to the electronic energy. In this second step, we treat nuclei as classical particles. The Born-Oppenheimer approximation relies on the fact that nuclei are much heavier than electrons and it is customary 
 to use it for slow chemical reactions.  The system is thus considered in a \emph{quasistatic} way: the electrons are assumed to be in equilibrium at any moment of the motion of the nuclei.
   The physical setting that we consider is for zero temperature,  but it should be reasonable for small temperatures.  A discussion about the Born-Oppenheimer approximation and its limits, from a chemistry point of view, can be found in~\cite[Section~3.1]{jensen}. For a discussion from a mathematical point of view, we refer  to \cite{teufel}.

   Another usual approximation is to neglect the relativistic effects. In the non-relativistic setting, the kinetic energy of a particle with mass $m$ is given by the differential operator  $-(\hbar^2/2m)\Delta$, where $\Delta$ denotes the Laplace operator. 
 Nevertheless, the relativistic effects can be important, becoming stronger and stronger when the atomic numbers of involved atoms increase. Thus, if relativistic correction terms  are needed only for very precise calculations for light atoms such as hydrogen, for atoms as heavy as silver or copper, relativistic effects cannot be neglected. For elements of the sixth row of the periodic table such as gold, lead or mercury, relativistic effects alter the chemical and physical behavior qualitatively, see \cite{Aut}. Examples of relativistic effects are on molecular geometry, especially the length of covalent bonds, but they also give better values for dissociation energies of bonds and ionization potentials and explain common-life phenomena such as the yellow color of gold or the low melting point of mercury. Surveys about relativistic effects in structural chemistry can be found in \cite{Pyk} and \cite{Aut}. The specific case of isomerizations has been studied for molecules containing heavy atoms such as platinum~\cite{isom-platin}, gold~\cite{isom-gold} or silver~\cite{isom-silver} and more generally for organometallic molecules~\cite{isom-organo}.

In order to be able to easily use variational formulations for relativistic systems, we will use the semirelativistic Hamiltonian, where the kinetic energy is obtained by applying to the classical relativistic kinetic energy $E=\sqrt{c^2 
	p^2+m^2c^4}-mc^2$ the standard quantization $p\rightarrow-\iu\hbar\nabla$, which gives the  pseudo-differential operator
\[T=\sqrt{-c^2\hbar^2\Delta+m^2c^4}-mc^2.\]

          As it is well-known, the operator $T$ is non local, which makes it more difficult to deal with than the standard Laplacian. Moreover, it is not possible to include external electromagnetic fields in a relativistically invariant way since, in the associated evolution equation, space and time derivatives are not symmetric. That is why Dirac introduced the operator bearing his name to describe relativistic quantum particles. Chemists usually use it  in their computational study of relativistic effects, see for example~\cite{Pyk} and~\cite{Aut}. Nevertheless, it is well known that the Dirac operator is not bounded from below, making it impossible to use a simple variational formulation, see for example~\cite{thaller}. That is why we consider here the semirelativistic kinetic energy described by the operator $T$.
Hamiltonians associated with this kinetic energy have been studied as early as in the 70s, for example by Herbst in the case of an hydrogenoid ion~\cite{Herbst}. Lieb and Yau in \cite{LY} and Fefferman and de la Llave \cite{FeLla} discussed its stability in many-particle systems.

  The general mathematical framework to study isomerizations is similar in the non-relativistic and in the semirelativistic settings. Two isomers of a molecule correspond to two stable configurations of the entire system. In these configurations, the ground state energy of the system has local minima with respect to the positions of the nuclei. In the Born-Oppenheimer approximation, the isomerization reaction is described by a path of the nuclei connecting these two local minima. The difference between the maximum and initial ground state energies along one path corresponds to the amount of energy needed to bring the system from one state to the other. The \emph{activation energy} is the lowest energy needed for the reaction to happen, and it is obtained by minimizing the maximal energy over all possible paths linking the two stable configurations.  Figure~\ref{fig:optpath} provides a graphical representation  of a hypothetical optimal path.
  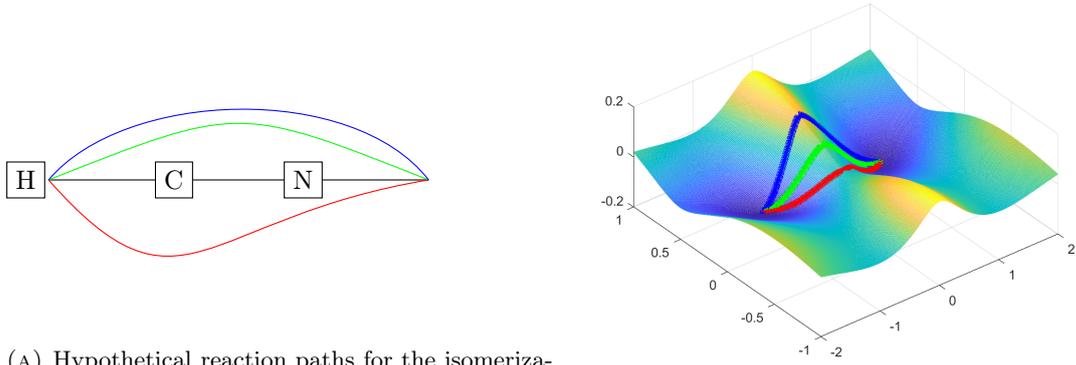
\begin{figure}[t]
  \begin{subfigure}{0.45\textwidth}
    \begin{tikzpicture}
     \node[draw] at (-0.3,0) {H};
    \draw (0, 0) -- (1.4, 0);
     \node[draw] at (1.65,0) {C};
    \draw (1.9, 0) -- (3.1, 0);
    \node[draw] at (3.35,0) {N};
        \draw (3.6, 0) -- (5, 0);
    \draw[green] (0, 0) .. controls (2.5,1) .. (5, 0);\draw[red] (0, 0) .. controls (1.75, -2) and (2, -0.5) .. (5, 0);
    \draw[blue] (0, 0) .. controls (1, 1.2) and (4, 1.3) .. (5, 0);
    \end{tikzpicture}
    \subcaption{Hypothetical reaction paths for the isomerization  HCN $\rightarrow$ CNH. The hydrogen atom can take different paths to go from one side to the other.}
  \end{subfigure}
     \begin{subfigure}{0.45\textwidth}
      \includegraphics[width=\textwidth]{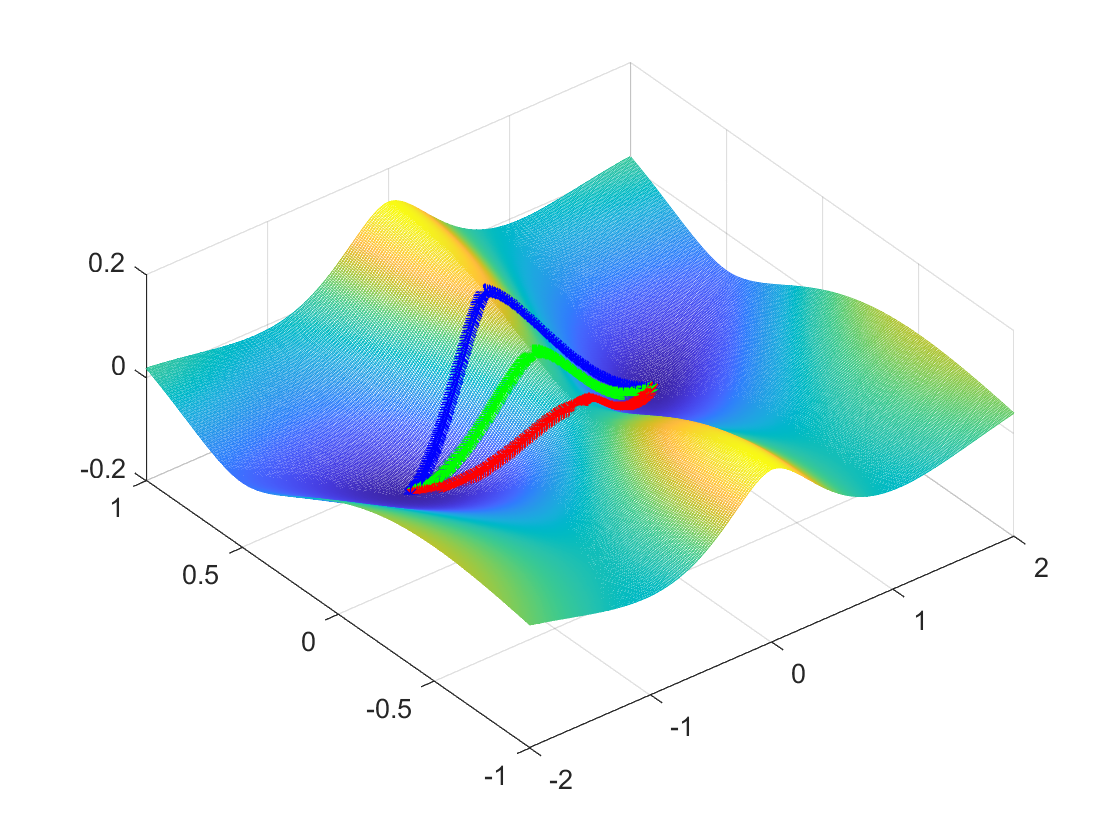}
	
	\subcaption{Here is a simplified hypothetical representation of the energy as a function of the relative position of the atoms. We see that the red path is optimal. We want to prove the existence of such an optimal path.}
	\end{subfigure}
	\caption{Existence of an optimal path.\label{fig:optpath}}\end{figure}
 While chemists work hard to compute the activation energy, from a mathematics point of view even its existence, namely the existence of an optimal path, is an open problem. 

From the point of view of critical point theory, an isomerization corresponds to a \emph{mountain pass problem}~\cite{AmbRab-73,Struwe,Jabri-03}, as Figure~\ref{fig:optpath} illustrates. The difficulties in this situation are, however, very different from the ones present in most cases of the mathematical literature where mountain pass problems are considered. Usually, the mountain pass theorem is used to prove existence of a critical point of a functional, implying existence of a solution to a PDE. In this situation, the configuration space is infinite dimensional. In our case, the mountain pass problem has, in its simpler formulation, a finite dimensional configuration space. However, as we will explain, proving that  an optimizing sequence  of paths stays on a bounded set is in our situation often difficult, and it is the main issue that we deal with in this paper. This boundedness would physically mean that isomerizations take place without breaking the whole system infinitely apart. This was conjectured and proven in some specific cases in  ~\cite{Lewin-04b,Lewin-PhD,Lewin-06}.  Another particularity of the mountain pass problem that we consider is that it would be  important to have an optimal path, as we discussed above.
The mountain pass theorem, however, does not imply the existence of such a path, even if the configuration space is finite dimensional.

One of the forces that play an important role in our analysis is the van der Waals force. It is a purely quantum and universal force that occurs when two molecules are placed at a distance $L$ from each other, where $L$ is not too small.
It is originated by an attractive potential of the order $1/L^6$ and it is thus
 an attractive force of the order of $1/L^7$. It is induced as a result of quantum correlations between the two molecules.
  A first proof of existence of an asymptotic expansion of the interaction energy of two neutral atoms in powers of the inverse distance between their nuclei  was given in \cite{MorSim-80}, assuming non-degeneracy of the ground state eigenspaces of the atoms. 
  An upper bound of the ground state energy of the system  was proved in ~\cite{LieThi-86}
  using an intricate test function. It implies  existence of an attractive potential that is, in the order of $L$, at least as strong as van der Waals potential for some orientations.
  The exact expression  in the leading order
  has  been more recently derived for individual non-relativistic atoms  in~\cite{AnaSig-17}, assuming that the ground state energies of the atoms are irreducible. Irreducible roughly means that the ground states of the atoms are unique up to spin permutations, see \emph{e.g.} Definition 1 and the related discussion in \cite{al}.
     In \cite{Anapolitanos-16}, the entire remainder was analyzed in the spinless case in order to extract information  on its dependence on the number of atoms. The long-range asymptotics was studied coupled with the number of atoms going to infinity. Progress in this direction would be important for the study of gases.
 For the case of two atoms, the results of \cite{AnaSig-17} were improved in \cite{AnaLewRot-19} in the following way: the monotonicity and derivatives of the interaction energy were studied
 and the irreducibility assumption was dropped for one of the two atoms.  In \cite{anbhundert},  the  interaction energy of a molecule with a half infinite plate was analyzed.
  In \cite{bhv}, the van der Waals force was derived for  atoms with semirelativistic kinetic energy, and the results also included an expansion with higher order terms  including \textit{e.g.} the Axilrod-Teller-Muto
three-body potential for  atoms, see also \cite{thesisMH}. In \cite{CaSc} the leading coefficient of the van der Waals interaction of two hydrogen atoms was studied and numerically approximated with rigorous proofs on the error estimates. In \cite{CaCoSc} the analysis was extended for coefficients of higher order terms. In \cite{CaCoSc}, it was moreover shown that  the scattering states of the hydrogen atom, namely the states associated with the continuous spectrum of the Hamiltonian, have a major contribution to the leading coefficient of the van der Waals expansion. Note that the results described above are \emph{static}, \emph{i.e.} they do not consider any motion of the nuclei.

In \cite{al}, the upper bound of \cite{LieThi-86} was improved and the results of \cite{AnaSig-17} were extended to the case of molecules, still in a non-relativistic model.
 The existence of the van der Waals force implies that \emph{all} neutral molecules can bind in the Born-Oppenheimer approximation, that is, the energy always has a global minimum with respect to the nuclear positions. This is because the energy of a molecule which splits into pieces is necessarily higher than the lowest possible energy, due to the van der Waals force. This was already observed in \cite{LieThi-86}.

However, studying isomerizations has additional difficulties, even in the simplest case that the system splits into two  solid submolecules, which is the case that we consider.
 This is because if two molecules are far apart, there usually exist some orientations for which they repel, for example if they have dipole moments oriented in opposite directions.  Here, the  information that the static results described above provide is not sufficient to conclude and we need to conduct a quasistatic analysis, in which nuclei are allowed to move. This was  already one of the main difficulties in the study of isomerizations in \cite{Lewin-04b,Lewin-06,al}. 

In~\cite{Lewin-04b},  Lewin proved  boundedness of all reaction paths, for a system containing two rigid molecules, each having a non-degenerate ground state with a non-vanishing dipole moment. A key point of the analysis is that the critical points of the dipole-dipole interaction with positive energy all have a Morse index greater than or equal to 2, \textit{i.e.} at least two directions of decrease.  In~\cite{Lewin-06}, the completely different situation of a molecule with only one moving atom (like for the reaction HCN $\rightarrow$ CNH) was treated. If the single atom escapes to infinity, then only the van der Waals force pertains since, by symmetry, the atomic ground state has no multipole moment in average. 
All other cases were left open in~\cite{Lewin-04b,Lewin-06}. 

In \cite{al}, the results  of ~\cite{Lewin-04b,Lewin-06} were extended in many directions, and important progress was made for the case of a molecule composed of two rigid submolecules. The most important results were in the case when the two molecules have low order multipole moments, that is, such that the leading energy of the system is of the order $1/L^p$ with $p<6$. In this case, the van der Waals force does not dominate in the long range asymptotics and the molecules may repel each other depending on their orientations. That was the main difficulty in \cite{al}.
 On the other hand, the authors proved that if sufficiently many of the multipole moments of the two molecules vanish, then the van der Waals force dominates in the long range asymptotics, as physically expected. In that case, the boundedness of reaction paths followed quite immediately from the dominance of the van der Waals force.

Here, we extend the results of \cite{al} in two directions. The first one is that we relax the assumption of  irreducibility of the ground state eigenspaces of the molecules. The new assumption is more natural and it states that the multipole moments of the molecules are independent of the choice of a normalized ground state. The second one is that we  consider molecules with semirelativistic kinetic energy  as well. As we explained above, isomerizations can involve molecules with heavy atoms, for which the relativistic effects of the electrons become important. Our approach  requires to extend results of \cite{al} and \cite{bhv} to the case of semirelativistic molecules on the van der Waals force, and relax the irreducibility assumption made in \cite{al} and \cite{bhv}. This direction is of interest on its own. Further difficulties in extending the results of \cite{al}
 arise on computing derivatives with respect to rotations of a variant of the ground state energy of the system, not only due to the nonlocality of the semirelativistic Laplacian, but also because, in the semirelativistic case, ground states have less regularity than in the non-relativistic case. In particular, in Section \ref{sec:firstderivative}, the approach of \cite{al} to take into account the quasistatic behavior does not seem to work in the semirelativistic setting,  while the results of~\cite{bhv} are purely static. We use ideas of \cite{Hunziker-86}, which Hunziker used to prove smoothness of the ground state energy for bosonic non-relativistic systems.  We  substantially rework  and refine these ideas so that we can quantify the derivatives and so that the analysis is applicable to the semirelativistic case as well.  In the non-relativistic case, our strategy  is simpler than the one of \cite{al} and quantifies higher order derivatives as well.

  On the other hand, the careful study of the multipole-multipole interactions in Section~4 of~\cite{al} does not depend on the kind of kinetic energy and is thus immediately applicable to our setting.

In the next section, we define the system properly and describe the main open questions. Then we state our new results. The rest of the paper is devoted to the proof of our theorems.

\bigskip

\section{Model and main results}\label{sec:modelresults}

\subsection{Schrödinger Hamiltonian for molecules}

We consider a molecule composed of 
\begin{itemize}
\item $N$ electrons of charge $e$ with position variables $x:=(x_1,...,x_N) \in \R^{3N}$;
\item $M$ nuclei of atomic numbers $\cZ:=(Z_1 ,..., Z_M ) \in \N^M$, and positions $Y:=(y_1,...,y_M) \in \R^{3M}$ with $y_j\neq y_k$ for $j\neq k$. 
\end{itemize}

The total nuclear charge will be denoted by
\begin{equation}\label{nuclearcharge}
|\cZ|:=\sum_{i=1}^M Z_i.
\end{equation}

As explained in the introduction, we work in the Born-Oppenheimer approximation. In a first step,
 the nuclei are fixed and pointwise, and we  study the spectrum of a Hamiltonian $H_N(Y,\cZ)$ only describing  electrons. In a second step, we
 consider the ground state energy of  $H_N(Y,\cZ)$ as an effective potential for the nuclei in order to investigate properties of reaction paths. To define $H_N(Y,\cZ)$, we need the Hilbert space associated with the electrons, which is  
\begin{equation*}
\mathcal{H}:=\bigotimes_{j=1}^N L^2\left(\R^3 \times \left\{-\frac{1}{2}, \frac{1}{2}\right\}\right),
\end{equation*}
with the inner product corresponding to the quadratic form given by
\begin{equation*}
 \langle\psi,\psi\rangle :=\sum_{s_1\in\{\pm1/2\}}...\sum_{s_N\in\{\pm1/2\}}\int_{\R^3}...\int_{\R^3}|\psi(x_1,s_1,x_2,s_2,...,x_N,s_N)|^2 \dd x_1 \dd x_2... \dd x_N.
\end{equation*}
The introduction of the set $\{-1/2,1/2\}$ is due to the fact that an electron has 2 spin states.

The usual models of quantum mechanics do not consider relativistic effects.
However, as explained in the introduction, when the  molecule contains heavy atoms, it is important to take these effects into account. Since we want to consider as well molecules without heavy atoms, we introduce both the non-relativistic  and the semirelativistic kinetic energy operators, which are, in physical units and for the $j-$th electron, 
\begin{equation}\label{def:Tj}
T_j:=\begin{cases}\sqrt{-\hbar^2c^2\Delta_j+m^2c^4}-mc^2, \quad &\text{ semirelativistic case \textbf{(SR)}}, \\ -\frac{\hbar^2}{2m} \Delta_j, \quad &\text{ non-relativistic case \textbf{(NR)}}, \end{cases}
\end{equation}
where $\Delta_j $ is the usual Laplace operator in the variable $x_j\in\R^3$, $\hbar$ is the reduced Planck's constant, $c$ the speed of light and $m$ the mass of the electron. The square root of the differential operator is defined through Fourier transformation: for $\psi\in H^1(\R^3)$ and almost every $x\in\R^3$, we have in the \textbf{(SR)} case \begin{equation}\label{def:T}T\psi(x):= \left(\mathcal{F}^{-1}
  \left[( \sqrt{c^2\hbar^2|\cdot|^2+m^2c^4}-mc^2)\mathcal{F} \psi\right]\right)(x), 
 \end{equation}
where $\mathcal{F}$ denotes the Fourier transformation on $L^2(\R^3)$ 
          \begin{equation}\label{def:Fourier}
         (\mathcal{F}\psi)(p):=\frac{1}{(2\pi)^{3/2}}\int_{\mathbb{R}^3}\mathrm{e}^{-\iu p\cdot x}\psi(x)\mathrm{d} x.
         \end{equation}
Of course, $T_j$ can also be defined with the Fourier transformation as a sesquilinear form on $H^{\frac{1}{2}}(\R^3)$.
          Note that, due to the fermionic nature of electrons, we consider the cases where all  kinetic energies are semirelativistic or all are non-relativistic. Since we consider a system with $N$ electrons, we introduce 
the $N$-particle Hamiltonian
\begin{equation}\label{def:HN}
H_{N}(Y,\cZ):= \sum_{j=1}^N T_j + \hbar c\alpha I_{N}(Y,\cZ),
\end{equation}
where $\alpha:=\frac{e^2}{4\pi\epsilon_0\hbar c}$ is a dimensionless physical constant, which is called the fine structure constant ($\alpha\approx1/137$). The electric potential is
\begin{equation}\label{def:IN}
I_{N}(Y,\cZ):= - \sum_{i=1}^M \sum_{j=1}^N \frac{ Z_i}{|x_j-y_i|} + \sum_{1 \leq i < j \leq N} \frac{1}{|x_i-x_j|} + \sum_{1 \leq i < j \leq M }\frac{Z_i Z_j}{|y_i-y_j|},
\end{equation}
where the summands respectively represent attraction between the electrons and the nuclei, repulsion between the electrons and repulsion between the nuclei.

In the non-relativistic case, $H_N(Y,\cZ)$ is always self-adjoint with  operator domain $H^2(\mathbb{R}^{3N})$ and bounded from below: see for example~\cite[Theorem~6.2]{Lewin-lecture},  or~\cite[Section~11.1]{teschl}, where the proof is given only in the case of atoms but can be easily adapted to molecules. In the semirelativistic case,
it was proved by Lieb and Yau in \cite{LY} and Fefferman and de la Llave in \cite{FeLla} that, if $\alpha$ is small enough, including the physical case $\alpha\approx1/137$, and, for all $k \in \{1,\dots,M\}$,
\begin{equation}\label{con-SA}
Z_k < \frac{2}{\pi\alpha},
\end{equation}
then there exists $b<1$ such that, for all $\psi\in H^{1/2}(R^{3N})$,
\begin{equation}\label{rel-bound}
 \langle\psi,I_{N}(Y,\cZ)\psi \rangle \leq b \Big\langle\psi,\sum_{i=1}^NT_i\psi\Big\rangle,
\end{equation}
uniformly in $Y$. 
As a consequence,
$H_{N}(Y,\cZ)$ is bounded from below and self-adjoint with form domain $H^{\frac{1}{2}}(\mathbb{R}^{3N})$. We note here that the inequality \eqref{rel-bound} does not seem to provide information on the operator domain of  $H_{N}(Y,\cZ)$. In particular, the operator domain may  contain functions that are not in $H^1$. The $\sum_{i=1}^N T_i$ operator boundedness of $I_{N}(Y,\cZ)$ with relative bound less than 1 would guarantee that the  operator domain of $H_N(Y,Z)$ is $H^1(\R^{3N})$. Such operator boundedness follows from the Hardy inequality if $\alpha$ is small enough but it is not clear to us how small $\alpha$ should be for large molecules. In particular, it is not clear if the physical value $\alpha \approx \frac{1}{137}$  would be included. For this reason, we avoid this assumption.  In fact, it is not certain that $H^1$ regularity would make our analysis easier.

In the rest of the paper, we will assume that we have chosen units such that 
\begin{equation}\label{hcm1}
\hbar=c=m=1.
\end{equation}
 Moreover, we  assume that $\alpha$ is small enough  and that~\eqref{con-SA} holds, so that $H_N(Y,\cZ)$  is bounded from below and we are not going to mention these assumptions explicitly. 

For a more precise analysis, one needs to consider the fermionic nature of the electrons. To this end, we consider  the permutation group $S_N$ of $\{1,\ldots, N\}$. 
We denote by 
$ \bigwedge_{j=1}^N L^2\left(\R^3 \times \left\{-\frac{1}{2},\frac{1}{2}\right\}\right)
$
the Hilbert space of antisymmetric square-integrable functions $\Psi(x_1,s_1,\dots ,x_N,s_N)$ with spin, that is, such that
\begin{align}
\pi\cdot \Psi(x_1,s_1, ...,x_N, s_N) :=\Psi(x_{\pi(1)}, s_{\pi(1)}, \dots , x_{\pi(N)}, s_{\pi(N)}) \label{def:piPsi} 
=(-1)^\pi\,\Psi(x_{1},s_1,\dots ,x_{N}, s_N) 
\end{align}
for any permutation $\pi\in S_N$, and we make use of the orthogonal projection 
on this space
\begin{equation*}
Q_N:  \bigotimes_{j=1}^N L^2\left(\R^3 \times \left\{-\frac{1}{2},\frac{1}{2}\right\}\right) \to \bigwedge_{j=1}^N L^2\left(\R^3 \times \left\{-\frac{1}{2},\frac{1}{2}\right\}\right),
\end{equation*}
where
\begin{equation}\label{def:Q}
Q_N \Psi 
=\frac{1}{N!} \sum_{\pi \in S_N} (-1)^{\pi}\pi\cdot\Psi,
\end{equation}
with $\pi\cdot\Psi$ defined in~\eqref{def:piPsi}. Then, since the Hamiltonian $H_N(Y,\cZ)$, defined in \eqref{def:HN}, commutes with $Q_N$, we can consider the fermionic Hamiltonian defined by
\begin{equation} \label{def:Hstat}
\hat{H}_{N}(Y,\cZ) := H_{N}(Y,\cZ) |_{\mathrm{Ran}Q_N} 
\end{equation}
which acts on the Hilbert space  
\begin{equation}\label{def:QNH}
 Q_N \mathcal{H} = \bigwedge_{j=1}^N L^2\left(\R^3 \times \left\{-\frac{1}{2},\frac{1}{2}\right\}\right).
\end{equation}
 
 The form domain of the operator $\hat{H}_N(Y,\cZ)$ is
\begin{equation}\label{def:domainHamilton}
Q\big(\hat{H}_N(Y,\cZ)\big)=\begin{cases}
H^{\frac{1}{2}}\left((\R^3\times\{\pm1/2\})^N\right) \cap Q_N \mathcal{H}, \quad &\textbf{(SR)}\\
H^1 \left((\R^3\times\{\pm1/2\})^N\right) \cap Q_N \mathcal{H}, \quad &\textbf{(NR)} 
\end{cases}
\end{equation}
where we recall that \textbf{(NR)}, \textbf{(SR)} respectively stand for non-relativistic and semirelativistic case, as defined by~\eqref{def:Tj}. 

We denote by 
\begin{equation}\label{def:gse}
\boxed{E_N(Y,\cZ):=\min\sigma\big(\hat{H}_N(Y,\cZ)\big)}
\end{equation}
the bottom of the spectrum of $\hat{H}_N(Y,\cZ)$, which is finite since the operator is bounded from below. Recall that $|\cZ|$ was defined in \eqref{nuclearcharge}.
In the \textbf{(NR)} case, when $N<|\cZ|+1$ (neutral or positively charged molecules), the HVZ~\cite{Hunziker-66,VanWinter-64,Zhislin-60} and Zhislin/Zhislin-Sigalov theorems ~\cite{Zhislin-60,ZhiSig-65}  imply that $E_N(Y,\cZ)$ is an eigenvalue of $\hat{H}_N(Y,\cZ)$, lying strictly below its essential spectrum:
$$E_N(Y,\cZ)<\min\sigma_{\rm ess}\big(\hat{H}_N(Y,\cZ)\big)=E_{N-1}(Y,\cZ),$$
and that the ground state eigenspace of $\hat{H}_N(Y,\cZ)$ is finite-dimensional.
Moreover, exponential decay of eigenfunctions associated with eigenvalues below the bottom of the essential spectrum is known,  see \textit{e.g.}~\cite{Combes-73,Simon-74,Hof-77,FroHer-82,Griesemer-04}. 
  In the \textbf{(SR)} case, analogous statements were proven in \cite{bhv} but only for atoms. Previous results on exponential decay of eigenfunctions of semirelativistic Hamiltonians had already been proven by Carmona, Masters and Simon in~\cite{CMS-90} and by Nardini in~\cite{Nardini1} and \cite{Nardini2}. Nevertheless, \cite{Nardini1} applied only in the 2-body case, \cite{CMS-90} does not include the case of Coulomb potentials and \cite{Nardini2} does not take fermionic statistics into account and was not applicable in the setting of \cite{bhv}. Closely following \cite{bhv}, we extend their results for molecules in the simpler setting where we do not need to work with irreducible representations of $S_N$. Most of the parts of these proofs are simple adaptations of the analysis of \cite{bhv} and earlier works but for convenience of the reader we shall include standard arguments. 
  
   As the theorems are known in the \textbf{(NR)} case, we state them here independently of the case.  We begin with the extension of the HVZ theorem to the \textbf{(SR)} case.
  \begin{theorem}\label{thm:HVZ}
  	We have that
  	\begin{equation*}
  	\sigma_{\text{ess}}(\hat{H}_{N}(Y,\cZ))=[\inf \sigma (\hat{H}_{N-1}(Y,\cZ)), \infty  ).
  	\end{equation*}
  \end{theorem}

The next theorem and corollary provide exponential decay of (approximate) eigenfunctions to eigenvalues below the essential spectrum.  For $x\in\R^d$, we denote by $|x|$ the Euclidean norm of $x$. The ball of radius $R >0$ centered in $y \in \mathbb{R}^d$ is denoted by
\begin{equation*}
B_R(y) := \{x \in \mathbb{R}^d\,|\; |x-y| \leq R\}.
\end{equation*}

\begin{theorem}[Exponential decay of approximate eigenfunctions]\label{thm:expdecaystat}
	We define the  ionization threshold of $\hat{H}_N(Y,\cZ)$ as
	\begin{equation}\label{ionthreshstatistic}
		\widetilde{\Sigma}:= \lim _{R \rightarrow \infty} \inf_{\substack{\psi \in D\left(\hat{H}_N(Y,\cZ)\right) \setminus\{0\},\,\\ \supp \psi \cap B_R(0)= \emptyset}} \frac{\langle \psi, \hat{H}_{N}(Y,\cZ) \psi \rangle}{\langle \psi, \psi \rangle}.
	\end{equation}
For any fixed $\mu < \widetilde{\Sigma}$, let $(\gamma_s)_{s \in \cI}$ be  a family of functions in $\Dom\big(\hat{H}_N(Y,\cZ)\big)$ satisfying  for all $s \in \cI$,
	\begin{equation}\label{gs:unif}
\|\gamma_s\|_{L^2}\leq C, 	\qquad 	( \hat{H}_{N}(Y,\cZ) -\mu) \gamma_s = \Gamma_s.
	\end{equation}
  Here $C>0$ is a constant and  $(\Gamma_s)_{s \in \cI}$ is a family of functions such that there exists  $C_1,a>0$ such that $\|e^{a |.|} \Gamma_s\|_{ L^2(\R^{3N})}\leq C_1$, for all $s \in \cI$. Then $\exists b, D>0$ such that, for all $s \in \cI$
	\begin{equation}\label{def:expdecay}
		\|e^{b |.|}  \gamma_s\|_{ L^2(\R^{3N})}\leq D.
	\end{equation}
\end{theorem}

Note that the fact that $\gamma_s$ is not assumed to be an exact eigenfunction of $\hat{H}_N(Y,\cZ)$ can be useful, as one may prove with Theorem  \ref{thm:expdecaystat} that
an exponentially decaying function remains exponentially decaying after multiplication with a resolvent of $\hat{H}_N(Y,\cZ)$. 

\medskip

\begin{coro}[Exponential decay of  eigenfunctions]\label{cor:expeigendecay}
	Choosing $ \Gamma = 0$, the above theorem implies that any  eigenfunction $\gamma$ of $\hat{H}_{N}(Y,\cZ)$ with associated eigenvalue below its essential spectrum is exponentially decaying, in the sense that~\eqref{def:expdecay} holds.
\end{coro}

\medskip
The following theorem is an extension of Zhislin's theorem to the \textbf{(SR)} case.
\begin{theorem}\label{thm:Zhislin}
	 If $|\cZ| > N - 1 $, then $ \hat{H}_{N}(Y,\cZ)$ has infinitely many eigenvalues below the infimum of its essential spectrum. In particular, it has a ground state and the ground state eigenspace is finite-dimensional.
\end{theorem}

\medskip
We prove Theorems \ref{thm:HVZ}, \ref{thm:expdecaystat} and \ref{thm:Zhislin} in Appendix \ref{pfthmspec}.
\begin{remark}
	Note that, in \cite{bhv}, a variant of Theorem \ref{thm:Zhislin} was proven for \textbf{(SR)} atoms, but only  the existence of at least one eigenvalue below the bottom of the essential spectrum was shown. Due to the nonlocality of the \textbf{(SR)} kinetic energy operator, we have to  modify their argument in order to prove that there are infinitely many eigenvalues below the bottom of the essential spectrum. 
\end{remark}
From Theorems~\ref{thm:HVZ} and~\ref{thm:Zhislin}, it follows that,  for every $Y=(y_1,...,y_M) \in \R^{3M}$ with $y_j\neq y_k$ for $j\neq k$, there exists at least one eigenfunction $\Psi\in D\big(\hat{H}_N(Y,\cZ)\big)$, such that 
\[\hat{H}_N(Y,\cZ)\Psi=E_N(Y,\cZ)\,\Psi.\]

 Moreover, the eigenspace associated with $E_N(Y,\cZ)$ has finite dimension.

\bigskip
\subsection{Isomerizations in the case of two rigid submolecules: the model}

From now on, we consider only globally neutral systems, \textit{i.e.} the case where 
\begin{equation*}
N= |\cZ|.
\end{equation*}
 
Moreover, we assume that the molecule is composed of two rigid submolecules, which can only be rotated and translated with respect to each other. This restriction is unsatisfactory as this has no physical justification and, in many isomerizations, there is splitting to more than
two solid submolecules. However, removing the assumption seems to be very hard. In some particular cases, \emph{e.g.}  reaction HCN $\to$ CNH, the assumption is reasonable. Some bigger organic molecules are made of 2 parts, which, if isomerization happens, should approximately remain rigid. This should be the case, for example, for the cis-trans isomerization of azobenzene, see \emph{e.g.}~\cite{TASVZB}. We refer to the introduction  and appendix~C of \cite{al} for some discussion of the general case of several submolecules.

Let $Y_1:=(y_1,...,y_{M_1})\in\R^{3M_1}$, $Y_2:=(y_{M_1+1},...,y_M)\in \R^{3M_2}$ be the nuclear positions of the two solid submolecules, where $M_1 + M_2=M$.  The collections of their atomic numbers are denoted by $\cZ_1:=(Z_1,...,Z_{M_1})\in\N^{M_1}$, $\cZ_2 := (Z_{M_1+1},...,Z_M)\in\N^{M_2}$.

 We assume that each of these molecules is neutral, which means that the $i$-th molecule  has $N_i=|\cZ_i|$ electrons, where $|\cZ_i|$ is defined  as in \eqref{nuclearcharge}.
 For $i=1,2$, we can define the Hamiltonian for the $i$-th molecule as in~\eqref{def:HN}
 \begin{equation}\label{def:Hi}H_i:=H_{|\cZ_i|}(Y_i,\cZ_i),\end{equation} 
  its version with statistics  as in~\eqref{def:Hstat} 
  \begin{equation}
  \label{def:Hstati}\hat{H}_i:=\hat{H}_{|\cZ_i|}(Y_i,\cZ_i)
  \end{equation}
  and its ground state energy
 \begin{equation}\label{def:Ei}E_i:=E_{|\cZ_i|}(Y_i,\cZ_i),\end{equation}
in the same way as in \eqref{def:gse}.
We then consider all possible ways of placing these two molecules in space. Without loss of generality, we can place the first molecule at the origin and the second one at a distance $L$ in the direction $e_1:=(1,0,0)$, and simply rotate the two molecules using $U,V\in \SO(3)$ (see Figure~\ref{fig:rigid_mol}). Our sole variables are therefore $(L,U,V)\in (0,\ii)\times \SO(3)\times \SO(3)$. For shortness, we introduce the new variable
$$\tau:=(L,U,V)\in (0,\ii)\times \SO(3)\times \SO(3)$$
and denote by
\begin{equation}\label{def:Ytau}
Y(\tau):=(Uy_1,\dots, U y_{M_1}, V y_{M_1+1}+Le_1,\dots, V y_{M}+Le_1 ),
\end{equation}
the nuclear positions, as well as by
\begin{equation}\label{def:Etau}\boxed{\cE_\tau:=E_N\big(Y(\tau),\cZ\big) = \inf \sigma (\hat{H}_N(Y(\tau),\cZ))}\end{equation}
the corresponding ground state energy. 
 The Hamiltonian has now the explicit expression
\begin{equation}\label{def:Htau}
H_\tau:=H_N(Y(\tau), \cZ)=  H_{1,\tau} + H_{2,\tau} + I_{\tau},
\end{equation}
where $H_{1,\tau}$ $H_{2,\tau}$ and $I_{\tau}$ are, respectively,  the Hamiltonians for the first and the second molecule and the interaction among them. Explicitly, they are given by
\begin{equation}
\label{def:H1tau}H_{1,\tau} := \sum_{j=1}^{N_1} T_j - \sum_{j=1}^{N_1}\sum_{k=1}^{M_1}\frac{ Z_k \alpha}{|x_j - Uy_k|} + \sum_{1 \leq i < j \leq N_1} \frac{\alpha}{|x_j - x_i|}  + \sum_{1\leq k<l\leq M_1}\frac{Z_kZ_l\alpha}{|Uy_k-Uy_l|},
\end{equation}
\begin{multline}
H_{2,\tau} :=  \sum_{j=N_1+1}^N T_j - \sum_{j=N_1 + 1}^{N}\sum_{k=M_1+1}^M\frac{Z_k \alpha}{|x_j - V y_k - L e_1|} + \sum_{N_1 + 1 \leq i < j \leq N} \frac{\alpha}{|x_j - x_i|}\\\label{def:H2tau}
+\sum_{M_1+1\leq k<l\leq M}\frac{Z_kZ_l\alpha}{|Vy_k-Vy_l|},
\end{multline}
\begin{multline}
I_{\tau} := -\sum_{j=1}^{N_1}\sum_{k=M_1+1}^M\frac{Z_k \alpha}{|x_j - Vy_k - L e_1|} - \sum_{j=N_1 + 1}^N \sum_{k=1}^{M_1}\frac{ Z_k \alpha}{|x_j - U y_k|} \\
+\sum_{k=1}^{M_1} \sum_{l={M_1+1}}^{M} \frac{ Z_k Z_l \alpha}{|U y_k-V y_l -Le_1|}+\sum_{i=1}^{N_1} \sum_{j=N_1 + 1}^N \frac{\alpha}{|x_i - x_j|}.\nonumber
\end{multline}
 Note that there is here a small abuse of notation, since $H_{j,\tau}$ can act either on  $L^2(\R^{3N_j})$ or on the whole $L^2(\R^{3N})$. We define as well the versions with statistics of the Hamiltonians by
\begin{equation}\label{Hhat}
\hat{H}_{\tau}:= H_{\tau}|_{\mathrm{Ran}Q_{N}}, \qquad \hat{H}_{1,\tau}:= H_{1,\tau}|_{\mathrm{Ran}Q_{N_1}}, \qquad \hat{H}_{2,\tau}:= H_{2,\tau}|_{\mathrm{Ran}Q_{N_2}},
\end{equation}
where $Q_{N_i}$ is the antisymmetrizer for $N_i$ particles, as defined in~\eqref{def:Q}.

\begin{figure}[t]
\centering
\includegraphics[width=9cm]{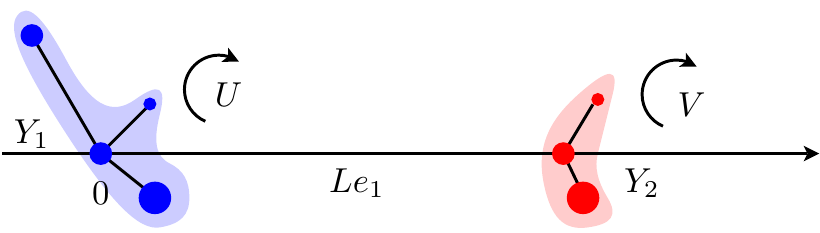}
\caption{The case of two rigid molecules (the picture was taken from \cite{al} and was created by Mathieu Lewin).\label{fig:rigid_mol}}
\end{figure}

 Since we are interested in isomerizations, we have to assume that there exist different isomers for the considered molecule. This is done with the following condition.
 
 \medskip
 
  \begin{Assumption}\customlabel{min}{\textbf{(Min)}}
  There exist two local minima $\tau_0,\tau_1\in (0,\ii)\times \SO(3)\times \SO(3)$ of $\tau \mapsto \cE_{\tau}$ that are strict with respect to the relative displacements $\tau$ that do not leave $Y(\tau)$ invariant. 
 \end{Assumption}
 
\medskip 
 
This condition is natural, as it means that the reactant and the product of the isomerization are stable molecules, in the sense that energy is needed to change their structure.

We consider the \textit{mountain pass level}
\begin{equation}
\boxed{ c:=\inf_{\substack{\tau(0)=\tau_0,\\ \tau(1)=\tau_1}}\;\max_{t\in[0,1]}\;\cE_{\tau(t)}}
 \label{eq:mountain_pass}
\end{equation}
where, as usual, it is understood that $\tau:[0,1]\to (0,\ii)\times \SO(3)\times \SO(3)$ is continuous. 
If $\tau_0$ gives the shape of the reactant of the isomerization, the difference $c-\cE_{\tau_0}$ can be understood as the least energy needed for the isomerization to happen. The existence of a continuous path $\widetilde{\tau}$ with $\widetilde{\tau}(0)=\tau_0, \widetilde{\tau}(1)=\tau_1$ and $c=\max_{t\in[0,1]}\;\cE_{\widetilde{\tau}(t)}$ would imply that the needed energy for the interaction can be indeed optimized. As we mentioned above, this is, to our best of knowledge, an open problem.

In our setting, we consider paths where the system is in a ground state at each time. However, it makes sense to let the electrons evolve along the path as well. For this reason, considering two local minima $\tau_0$ and $\tau_1$ as in Condition~\ref{min}, we choose two ground states $\overline\Psi_0$ and $\overline\Psi_1$ corresponding to the ground state energies $\cE_{\tau_0}$ and $\cE_{\tau_1}$, respectively, and define a new mountain pass level by
\begin{equation}
	\boxed{c':=\inf_{\substack{
				(\tau,\Psi)(0)=(\tau_0,\overline\Psi_0)\\ (\tau,\Psi)(1)=(\tau_1,\overline\Psi_1)}}\;\max_{t\in[0,1]} \pscal{\Psi(t),H_{\tau(t)}\Psi(t)}}
	\label{eq:mountain_pass_level_general_case_with_Psi}
\end{equation}
which is independent of the choices of $\overline\Psi_0$, $\overline\Psi_1$ and obviously larger than or equal to $c$. Recall that $H_{\tau(t)}$ is defined by~\eqref{def:Htau}.
It is understood here that $\tau(t)$ is a  continuous path as before and that $\Psi(t)$ is a continuous path of wavefunctions in the form domain of $\hat{H}_{\tau(t)}$ (which is in fact independent of $t$, as stated in~\eqref{def:domainHamilton}), such that $\|\Psi(t)\|_{L^2}=1$ for all $t\in[0,1]$. 

We note that the mountain pass problem \eqref{eq:mountain_pass_level_general_case_with_Psi} is an infinite dimensional mountain pass problem, and even if existence of an optimizer for the problem \eqref{eq:mountain_pass} were known, existence of an optimizer for  \eqref{eq:mountain_pass_level_general_case_with_Psi} would not immediately follow.
Nevertheless, it is proven in \cite[Theorem~4]{Lewin-04b} (see also~\cite[Appendix~A]{al}) that  $$c'=c.$$
Moreover, existence of a bounded minimizing sequence $(\tau_n)$ for~\eqref{eq:mountain_pass} implies existence of a bounded minimizing sequence $(\tau_n,\Psi_n)$ for~\eqref{eq:mountain_pass_level_general_case_with_Psi}. If an optimal path $(\tau_n,\Psi_n)$ exists, the point of its maximum is called a \emph{transition state} in chemistry. In the \textbf{(NR)} setting, Theorem~4 of~\cite{Lewin-04b} provides, under condition~\ref{min} and boundedness of a minimizing sequence, existence of such a transition state in a more general setting: it is not necessary that the system splits into only two rigid submolecules. We conjecture that this holds in the \textbf{(SR)} case too.

\bigskip
\subsection{Case of two rigid molecules: new results}

In this section, we state our  results, which generalize those in~\cite{al,Lewin-04b,Lewin-06}. 

We begin by giving an asymptotic expansion of the energy of the 2 interacting submolecules, before stating our main theorem.

\subsubsection{Expansion of the energy: multipole interactions and the van der Waals force}
In this section, we expand the energy $\cE_\tau=\cE_{(L,U,V)}$ for large $L$ up to order $L^{-6}$ and get the van der Waals energy as well as all the lower order multipolar energies. In particular, we prove that there always exists a non-vanishing term of order $L^{-6}$. In the analysis, error terms depending on $L, U, V$ are going to appear. Since we very often need to bound them uniformly in $U,V$ with functions of $L$, we introduce the following notation:
Let $g=g(L,U,V)$ be a function with values on a normed vector space $W$ and $f:\R^+ \to \R^+$. We say that
 \begin{equation*}
 g = O_{(L^\infty,W)} (f(L)),
 \end{equation*}
if there exist  constants $C, D >0 $ such that if $L \geq D$ then
  \begin{equation*}
  \|g(L,U,V)\|_{W} \leq C f(L),
  \end{equation*}
  uniformly in $U,V$.
  
   Furthermore, since we consider not only static but also quasistatic setting, we often need to control derivatives of quantities with respect to the position as well. To this purpose, let us introduce the set of matrices\begin{equation}\label{def:so3}
                                          \Gamma:=\{M\in\mathfrak{so}(3)\,|\;|M_{ij}|\leq 1,\; \forall i,j \in\{1,2,3\}\},
                                          \end{equation}
where recall that $\mathfrak{so}(3)$, the Lie algebra of $\SO(3)$, is the set of 3-by-3 real antisymmetric matrices. 
  We  say that
   \begin{equation}\label{C2error}
   g = O_{(C^\infty,W)} (f(L))
   \end{equation}
   if there exists a constant $D >0$ such that,  for all $n \in \N \cup \{0\}$, there exists a constant $C_n$ such that, for all $L \geq D$, $(U,V) \in \SO(3)^2$, all  $A,B \in \Gamma$ and all $t$ in a neighborhood of $0$, we have
    \begin{equation}\label{def:C2error}
  \left\|\frac{\d^n}{\d t^n}g(L,e^{tA}U,e^{tB}V)\right\|_{W}  \leq C_n f(L).
    \end{equation} 
  If $W=\R$ or $W=\C$, then $W$ is going to be omitted from the notation and we will write $O_{L^\infty} (f(L))$ or $O_{C^\infty} (f(L))$.
When we write $O_{(L^\infty,W)}(e^{-dL})$ or $O_{(C^\infty,W)}(e^{-dL})$, we will mean, even if it is not mentioned, that there exists a $d>0$ such that we have the respective estimates. The constant $d$ might change from one equation to the other.
 The definition of regularity given here is formulated differently in~\cite{al}, where geodesics are considered. Nevertheless, both definitions are equivalent, since the geodesics on $\mathrm{SO}(3)$ are given by the exponential map (in the sense of Lie groups). An explanation of this fact can be found in~\cite[Section~2.2]{AlexBett}.

In order to be able to formulate our theorems, we need to introduce the multipolar interactions. We recall here the  expansions introduced in~\cite[Section~2.1]{al}, to which it can be referred for a more detailed explanation. The reader already aware of it can directly go to Definition~\ref{def:FNVMM}.

For a bounded measure $\rho$ decaying faster than any polynomial at infinity, we can define the $2^n$-pole moment associated with this measure as a tensor of order $n$ on $\R^3$ from the Taylor expansion of $\rho*1/|\cdot|$:
\begin{equation}\label{def:multipole}
\mathcal{M}^{(n)}_\rho(h_1,...,h_n) := \frac{(-1)^n}{n!}\int_{\R^3}|x|^{2n + 1} \left[(h_1 \cdot \nabla_x)...(h_n\cdot \nabla_x)  \frac{1}{|x|} \right] \dd \rho(x).
\end{equation}

We will use multipole moments defined by~\eqref{def:multipole} in the case where the measures are charge densities, which we define now. Given a state $\psi\in Q_N\Hilb$, where recall that $Q_N\Hilb$ has been defined in~\eqref{def:QNH}, we define  the \emph{electronic density} on $\R^3$ by
\begin{equation}\label{def-elec-dens}
 \rho_\psi(x):=N\sum_{s_1\in\{\pm \frac12\}}...\sum_{s_N\in\{\pm \frac12\}}\int_{\R^3}...\int_{\R^3}|\psi(x,s_1,x_2,s_2,...,x_N,s_N)|^2\dd x_2... \dd x_N,
\end{equation}
where we multiplied the first marginal by $N$ instead of summing over all marginals, which we can do because  electrons are fermions.

Given  normalized ground states $\Psi_j$ of $\hat{H}_j$, $j=1,2$, as defined in~\eqref{def:Hstati}, we can define the \emph{molecular charge densities} $\rho_{j}^{\Psi_j}$, $j=1,2$ by
	\begin{equation}\label{def:MCD}\rho_{1}^{\Psi_1}(x):=\sum_{i=1}^{M_1}Z_{i}\delta_{Y_i}(x)-\rho_{\Psi_1}(x), \qquad \rho_{2}^{\Psi_2}(x):=\sum_{i=M_1+1}^{M}Z_{i}\delta_{Y_i}(x)-\rho_{\Psi_2}(x).\end{equation}
Note that they depend, in general, on the choice of the ground states $\Psi_1$ and $\Psi_2$.

We introduce the \textit{multipolar interaction} associated with two bounded measures $\rho_1$ and $\rho_2$ on $\mathbb{R}^3$ decaying faster than any polynomial at infinity to be 
\begin{multline}\label{multipol-coeff}
\cF^{(n,m)}(\rho_1, \rho_2):=K_{m,n}\sum_{j_1,...,j_n}\sum_{k_1,...,k_m}\mathcal{M}^{(n)}_{\rho_1}(e_{j_1},...,e_{j_n})\mathcal{M}^{(m)}_{\rho_2}(e_{k_1},...,e_{k_m})\times \\
\times \left(\partial_{z_{j_1}}...\partial_{z_{k_m}}\frac{1}{|z+e_1|}\right)\bigg|_{z=0},
\end{multline} where $K_{m,n}=(-1)^m/\Pi_{i=1}^m(2i-1)\Pi_{j=1}^n(2j-1)$, and define then the \emph{multipolar interactions} between the two molecules in respective states $\Psi_1$ and $\Psi_2$ to be $\mathcal{F}^{(n,m)}(\rho_1^{\Psi_1}, \rho_2^{\Psi_2})$.

\smallskip
The following Lemma is key in order to apply the expansion of the Coulomb interaction and exploit the dependence of the energy on the multipole moments. As we will see,  the $2^n-$pole-$2^m-$pole interaction decays as $1/L^{n+m+1}$. More explicitly, the charge-charge, dipole-charge, dipole-dipole and dipole-quadrupole interactions are proportional to $L^{-1}$, $L^{-2}$, $L^{-3}$ and $L^{-4}$, respectively. The quadrupole-quadrupole and dipole-octopole interactions are proportional to $L^{-5}$. For technical reasons it is useful to cut the ground state energies of. To this end we define
a localizing function on $\mathbb{R}^3$ by
\begin{equation}\label{def:chiL}
\chi_L(x) := \chi \left( \frac{|x|}{L}\right), \qquad\text{where } \chi (r)= \begin{cases}
1, &\text{if} \; r \leq 1/10\\
0, &\text{if} \; r \geq 1/8  
\end{cases}
\end{equation}
with $\chi \in C^{\infty}(\mathbb{R})$ and $\mathrm{Ran}\chi = [0,1]$. Let 
\begin{equation}\label{def:Phij}
\Phi_j:=\frac{\chi_L^{\otimes N_j} \Psi_j}{\|\chi_L^{\otimes N_j} \Psi_j\|}, \quad  j=1,2,
\end{equation}
where recall that $N_1,N_2$ are the numbers of the electrons of the first and second submolecule, respectively. 
Similarly as in \eqref{def:MCD} we can define $\rho_j^{\Phi_j}$ and we note that
global charge neutrality implies that 
 \begin{equation}\label{charge-neut}
\rho_{j}^{\Phi_j}(\R^3)=0, \qquad j=1,2.
 \end{equation}    
\begin{lemma}\label{thm:multipol2}
	 For any $U\in \SO(3)$, we define the rotated measure $U\rho_j^{\Phi_j}$ ($j=1,2)$ by $U\rho_j^{\Phi_j}(D):=\rho_j^{\Phi_j}(\{x, U^{-1}x\in D\})$ for any Borel set $D$. For any $K \geq 2$ and $U,V \in \SO(3)$, we have
	\begin{equation*}
	\int_{\mathbb{R}^6} \frac{\mathrm{d} \rho_1^{\Phi_1}(x) \mathrm{d} \rho_2^{\Phi_2}(y)}{| Le_1 + V y - U x|} = \sum_{2 \leq n+m \leq K-1} \frac{\cF^{(n,m)}( U\rho_1^{\Phi_1}, V\rho_2^{\Phi_2})}{L^{n+m+1}} +  O_{C^\infty}\left(\frac{1}{L^{K+1}} \right).
	\end{equation*}
\end{lemma}
\begin{proof}
 The proof is a modification of the one of Lemma~2.2 of~\cite{al}, whose statement was similar but only with  $C^2$ control of the error term. We note that for all
  $x,y \in B(0, \frac{L}{8})$ and all $U,V \in \SO(3)$ the function $s \mapsto \frac{1}{|L e_1 + s(Vy - Ux)|}$ is in $C^{\infty}([0,1])$. From Taylor's Theorem we find 
 \begin{align}\nonumber \frac{1}{|Le_1+Vy-Ux|}=\sum_{i=0}^{K-1} \frac{1}{i!}&((Vy-Ux)\cdot \nabla)^i\frac{1}{|Le_1|}  \\ \label{Taylorsatz} + \int_0^1 \int_0^{s_1} \dots & \int_0^{s_{K-1}}  ((Vy-Ux)\cdot \nabla)^{K}\frac{1}{|Le_1+s_{K}((Vy-Ux))|}\d s_{K} \dots \d s_1,
\end{align}
where $((Vy-Ux)\cdot \nabla)^i\frac{1}{|Le_1|}$ is, by abuse of notation, a directional derivative of order $i$ of the Coulomb potential at the point $Le_1$.
Since by \eqref{def:Phij} both $\rho_j^{\Phi_j}, j=1,2$ are supported in $B(0, \frac{L}{8})$ it follows after integrating with respect to the corresponding measures  that
	\begin{align}\label{Taylorintegral} 
	 \int_{\mathbb{R}^6} \frac{\mathrm{d} \rho_1^{\Phi_1}(x) \mathrm{d} \rho_2^{\Phi_2}(y)}{| Le_1 + V y - U x|} =  \sum_{i=0}^{K-1} B_i(U,V) + R, 
	 \end{align}
with
	 \begin{align}\label{def:K}
	  R:= \int_{\mathbb{R}^6} \int_0^1 \int_0^{s_1} \dots  \int_0^{s_{K-1}}  ((Vy-Ux)\cdot \nabla)^{K}\frac{1}{|Le_1+s_{K}(Vy-Ux)|}\d s_{K} \dots \d s_1 \mathrm{d} \rho_1^{\Phi_1}(x) \mathrm{d} \rho_2^{\Phi_2}(y)
	\end{align}
	and
 \begin{equation}\label{def:Bi}
  B_i(U,V):=\frac{1}{i!}\int_{\R^6}(Ux_{a_1}-Vy_{a_1})\dots(Ux_{a_i}-Vy_{a_i}) \partial_{a_1}\dots\partial_{a_i}\frac{1}{|Le_1|} \d \rho_1^{\Psi_1}(x)\d \rho_2^{\Psi_2}(y),
 \end{equation}
where we used Einstein's summation convention on the indices $a_1,\dots,a_i$ and $\partial_{a_1}\dots\partial_{a_i}\frac{1}{|Le_1|}$  means partial derivative of the Coulomb potential at the point $Le_1$.
But we know from~\cite[Appendix~B2]{al} that
\begin{equation}\label{eq:evalB}
 B_i(U,V)=\frac{1}{L^{i+1}}\sum_{j=0}^{i}\cF^{(j,i-j)}(U\rho_1^{\Phi_1},V \rho_2^{\Phi_2}).
\end{equation}
We note that the proof of \eqref{eq:evalB} is based on decomposing the integrand of the right hand side of \eqref{def:Bi} in terms that are homogeneous of degree $j$ in $x$ and of degree $i-j$ in $y$.   
From \eqref{def:multipole}, \eqref{multipol-coeff}, \eqref{charge-neut}  
it follows that 
\begin{equation}\label{B0B10}
\cF^{(i,0)}(U\rho_1^{\Phi_1},V \rho_2^{\Phi_2})=\cF^{(0,i)}(U\rho_1^{\Phi_1},V \rho_2^{\Phi_2})=0.
\end{equation}
Inserting~\eqref{eq:evalB} and \eqref{B0B10} into \eqref{Taylorintegral} we arrive at the statement of Lemma \ref{thm:multipol2} provided that
\begin{equation}\label{eq:KCinfty}
R=O_{C^\infty}\left(\frac{1}{L^{K+1}}\right),
\end{equation}
which is what remains to prove. To prove \eqref{eq:KCinfty} we first observe that for all $h \in \R^3$ and $\xi \in \R^3 \setminus \{0\}$ we have
\begin{equation}\label{dirder}
 ( h \cdot \nabla)^{K}\frac{1}{|\xi|}=\frac{g^{K}\big(h,\xi\big)}{|\xi|^{K+1}},
 \end{equation}
 where $g^{K}: \R^3 \times \R^3 \to \R$ is a smooth function which is homogeneous of degree $K$ in $h$ and homogeneous of degree 0 in $\xi$ and therefore there exists $C_K>0$ such that
 \begin{equation}\label{gKcontrol}
 \Big|g^{K}\big(h,\xi\big)\Big| \leq C_K |h|^K, \quad  \forall  h \in \R^3, \forall \xi \in \R^3 \setminus \{0\}. 
 \end{equation}
 Applying \eqref{dirder} and \eqref{gKcontrol} for $h=Vy-Ux$ and $\xi= Le_1+s(Vy-Ux)$ we find that 
 \begin{equation}\label{derivcontrol}
\bigg| ((Vy-Ux)\cdot \nabla)^{K}\frac{1}{|Le_1+s_{K}(Vy-Ux)|}\bigg| \leq C_K \frac{|Vy-Ux|^K}{|Le_1+s_{K}(Vy-Ux)|^{K+1}}.
 \end{equation}
  We also have for all $(x,y) \in \supp \rho_1^{\Phi_1} \times \supp \rho_2^{\Phi_2}$
\begin{equation}\label{denominatorcontrol}
|Le_1+s(Vy-Ux)| \geq \frac{3L}{4}.
\end{equation} 
Inserting \eqref{derivcontrol} and \eqref{denominatorcontrol} in
 \eqref{def:K} and using Corollary \ref{cor:expeigendecay}, we find that
  $R=O_{L^\infty}\left(\frac{1}{L^{K+1}}\right)$.  
  To control derivatives $R$ with respect to rotations we notice   that the dominated convergence Theorem is applicable in the right hand side of \eqref{def:K} as the densities  $\rho_j^{\Phi_j}$ are in $L^1$ and the derivatives with respect to rotations of  $((Vy-Ux)\cdot \nabla)^{K}\frac{1}{|Le_1+s_{K}(Vy-Ux)|}$   are bounded on the support of the densities.  Derivatives with respect to rotations of  $g^{K}$ can be controlled similarly as in \eqref{gKcontrol}  and any derivative of $\frac{1}{|Le_1+s_{K}(Vy-Ux)|^{K+1}}$ with respect to rotations provides an even bigger decay in $L$. We omit the details. With this we can conclude the proof of \eqref{eq:KCinfty} and therefore of Lemma \ref{thm:multipol2}.
\end{proof}

\bigskip

Lemma \ref{thm:multipol2} implies that the leading term of the expansion of the Coulomb interaction is given by the first multipole-multipole interaction which does not vanish. It is therefore important to know the first non-vanishing $2^n$-pole  moment of each submolecule.

\begin{definition}[First non-vanishing multipole moment]\label{def:FNVMM}
	For $k\in\{1,2\}$, let $\Psi_k$ be a normalized ground state of $\hat{H}_k$ and let $\rho_k^{\Psi_k}$ denote the corresponding molecular  charge density.  We call $n_k$ the smallest integer $n\geq1$ such that $\cM^{(n)}_{\rho_{\Psi_k}}\neq0$, see \eqref{def:multipole}. If all the multipole moments vanish, we let $n_k:=+\ii$.
\end{definition}

 Note that, in principle, $n_k$ depends on the choice of a normalized ground state $\Psi_k$.  Nevertheless, we will often assume that it is in fact not the case, at least for the leading term.
\begin{Assumption}\customlabel{poles}{\textbf{(Poles)}}
We assume that one of the following conditions holds
\begin{description}
\item[\customlabel{polesall}{\textbf{(Poles-all)}}] The numbers $n_1, n_2$ are independent of the choice of normalized ground states $\Psi_1, \Psi_2$ of $\hat{H}_1,\hat{H}_2$, respectively, and so are all their multipole moments $\cM^{(n)}_{\rho_{\Psi_1}}, \cM^{(n)}_{\rho_{\Psi_2}}$ for all $n\leq 4$.

\item[\customlabel{poleslead}{\textbf{(Poles-lead)}}] The numbers $n_1, n_2$ are independent of the choice of normalized ground states $\Psi_1, \Psi_2$ of $\hat{H}_1,\hat{H}_2$, respectively, and so are their leading multipole moments $\cM^{(n_1)}_{\rho_{\Psi_1}}, \cM^{(n_2)}_{\rho_{\Psi_2}}$. 
\end{description} 
\end{Assumption}

\begin{remark}
 Condition \ref{poles} follows from the assumption of the irreducibility of the ground state eigenspaces that was assumed in previous works like~\cite{al}. Condition~\ref{poles} is weaker, at least we do not see a reason why Condition \ref{polesall} or \ref{poleslead} should imply irreducibility. However, proving even Condition \ref{poleslead} is in the general case, to our best of knowledge, an open problem and it is known only for the hydrogen atom and molecule and for the helium atom,  at least in the non-relativistic case,  see \textit{e.g.}~\cite[Section~XIII.12]{ReeSim4}.
	 In these cases, irreducibility is also known. However, we find Condition \ref{poles}  more natural from a physical perspective, as chemistry handbooks speak about \emph{the} dipole or quadrupole etc. moment of a molecule, see for example~\cite{CRC}.
	 If we ignore the fermionic statistics, the ground state eigenspaces of the molecules are, at least in the non-relativistic case, non-degenerate, see again ~\cite[Section~XIII.12]{ReeSim4}. As a consequence,  Condition \ref{poles} is automatically satisfied. 	  
\end{remark}

Another assumption,  inspired by physical evidence, is that the configuration of lowest energy is the one where each of the two submolecules is neutral.
\begin{Assumption}\customlabel{neutr}{\textbf{(Neutr)}}
		The system satisfies the \emph{neutral best configuration assumption}: namely $(Y_1,\cZ_1)$ and $(Y_2,\cZ_2)$ are such that the only configuration of the system with the minimum energy possible is the  one where the charges are distributed such that the molecules are neutral:
		\begin{equation}\label{BNC}
		E_{|\cZ_1|} (Y_1, \cZ_1)+ E_{|\cZ_2|}(Y_2, \cZ_2) < \min_{\substack{N_1' \neq |\cZ_1|,\\ N_1' + N_2'=N}} \left(E_{N_1'} (Y_1, \cZ_1)+ E_{N_2'}(Y_2, \cZ_2)\right),
		\end{equation}
	where the energies appearing in the inequality are defined in the same way as  $E_N(Y,\cZ)$ in \eqref{def:gse}.
\end{Assumption} 
It is a famous conjecture that~\eqref{BNC} holds, see \emph{e.g.} the introduction of  \cite{AnaSig-17} for some discussions. If this is not the case, the mountain pass problem is actually rather easy, see \cite[Thm.~4]{Lewin-04b}.

From the results of~\cite{Morgan-79,MorSim-80,Lewin-04b}, we have
\begin{equation}
\boxed{\lim_{L\to\ii}\cE_{(L,U,V)}=E_\ii:=\min_{N_1'+N_2'=N}\big(E_{N_1'}(Y_1,\cZ_1)+E_{N_2'}(Y_2,\cZ_2)\big)}
\label{eq:limit_L_infinity}
\end{equation}
uniformly in $U,V\in \SO(3)$, where $E_{N_j}(Y_j,\cZ_j)$, $j=1,2$ is defined in the same way as $E_N(Y,\cZ)$ in \eqref{def:gse}.

In order to state the theorem about the asymptotic expansion of the energy, we introduce the following objects.
\begin{itemize}

\item $H_\infty$ is the sum of the Hamiltonians of the two submolecules, without interaction:
\begin{equation}\label{def:Hinfty}
 H_\infty:=\hat{H}_1\otimes \mathbbm{1}^{\otimes 3N_2} +   \mathbbm{1}^{\otimes 3N_1} \otimes \hat{H}_2,
\end{equation}
where $\hat{H}_1, \hat{H}_2$ were defined in \eqref{def:Hstati} and $\mathbbm{1}$ is the identity operator on $L^2(\R)$. 

	\item $f_{(U,V)}$ is the dipole interaction function 
	\begin{align}
	f_{(U,V)}(x_1, \ldots, &x_N) := U D_1(x_1, \ldots, x_{N_1})\cdot V D_2(x_{N_1 + 1}, \ldots, x_{N})  \nonumber \\
	&-3 (e_1 \cdot UD_1(x_1, \ldots, x_{N_1})) (e_1 \cdot VD_2(x_{N_1+1}, \ldots, x_{N})) \label{expansion:fUV}
	\end{align}
	where
	\begin{align*}
	D_1(x_1, \ldots, x_{N_1}) &:=  \sum_{j=1}^{N_1} x_j - \sum_{m=1}^{M_1} Z_m y_m;\\
	D_2(x_{N_1+1}, \dots, x_N) &:= \sum_{j= N_1 + 1}^{N} x_j - \sum_{m=M_1 + 1}^{M} Z_m y_m,
	\end{align*}
	are the \emph{instantaneous dipole moments} of the submolecules	and
	\begin{align*}
	U D_1(x_1, \ldots, x_{N_1}) &:= D_1(U^{-1}x_{1}, \ldots, U^{-1} x_{N_1}), \\
	V D_2(x_{N_1+1}, \ldots, x_{N}) &:= D_2(V^{-1}x_{N_1+1}, \ldots, V^{-1} x_{N}).
	\end{align*}
 Note that these instantaneous dipole moments are a function of the position, which can never be identically 0. They should not be confused with the  permanent dipole moments, which are zero if the positive and negative charges have the same barycenter. 

	\item Let 
	\begin{equation}\label{def:GSEigenspace}
	\mathcal{G}_j:=\mathrm{Ker}(\hat{H}_j-E_j),\ j=1,2,
	\end{equation}
	 denote the ground state eigenspace of the $j$-th molecule, where $\hat{H}_j$ and $E_j$ have been defined in~\eqref{def:Hstati} and~\eqref{def:Ei}. Furthermore, we denote by $\Pi$ the orthogonal projection onto $\mathcal{G}_1 \otimes \mathcal{G}_2$.
	 Let also 
	  	\begin{equation}\label{def:Hinftyperp}
	  	H_\infty^\bot:= \Pi^\bot 	H_\infty  \Pi^\bot,
	  	\end{equation}
	  	where $\Pi^\bot:=1-\Pi$ is the orthogonal projection on $(\mathcal{G}_1 \otimes \mathcal{G}_2)^\bot$. 
	\item For $\phi \in \mathcal{G}_1 \otimes \mathcal{G}_2$ with $\|\phi\|=1$, we define the \emph{van der Waals coefficient}
\begin{equation}\label{def:VdWterm}
C_{\mathrm{vdW}}(\phi,U,V) := \langle \Pi^{\bot} f_{(U,V)} \phi , (H_\infty^\bot -E_1-E_2)^{-1}  \Pi^{\bot} f_{(U,V)} \phi \rangle,
\end{equation}
with $E_\infty$ defined in~\eqref{eq:limit_L_infinity}, and 
\begin{equation}\label{eq:CvdW}
C_{\mathrm{vdW}}(U,V):= \max_{\phi \in \mathcal{G}_1 \otimes \mathcal{G}_2, \|\phi\|=1} C_{\mathrm{vdW}}(\phi,U,V).
\end{equation}

\item For the multipole-multipole terms $\cF^{(m,n)}(U \rho_{\Psi_1}, V \rho_{\Psi_2})$, where $\Psi_j \in \mathcal{G}_j$, if there is independence of the choice of normalized ground states implied by \ref{poles}, we will write $\cF^{(m,n)}(U, V)$.

\end{itemize}
We remark that the functions $f_{(U,V)}$ and thus $C_{\mathrm{vdW}}(\phi,U,V)$ for each fixed $\phi$ depend continuously on $U,V$. Note that, if~\ref{neutr} holds, then \eqref{def:Ei},~\eqref{BNC} and \eqref{eq:limit_L_infinity} imply that 
\begin{equation}\label{eq:Eiinf}
E_\infty=E_1+E_2=\min \sigma(H_\ii).
\end{equation}

Recalling that we have denoted $\tau=(L,U,V)$, we have the following results.

\begin{theorem}[Multipolar/van der Waals expansion of the energy]\label{thm:vanderwaals}\hfill
\begin{description}
\item[(a)] 
 Let $\Psi_1$ and $\Psi_2$ be any two normalized ground states of, respectively, $\hat{H}_1$ and $\hat{H}_2$. Then, 
\begin{equation}
 \cE_\tau\leq E_1+ E_2 + \sum_{2\leq n+m\leq 5}\frac{\cF^{(n,m)}(U\rho_{\Psi_1},V\rho_{\Psi_2})}{L^{n+m+1}} 
 -\frac{C_{\rm vdW}(\Psi_1 \otimes \Psi_2,U,V)}{L^6}+O_{L^\infty}\left(\frac{1}{L^7}\right).
 \label{eq:upper}
\end{equation}
Furthermore, the function 
$C_{\rm vdW}(\Psi_1 \otimes \Psi_2,U,V)$
is strictly positive for all $(U,V)$ in $\SO(3)\times \SO(3)$.

\item[(b)] If in addition Conditions \ref{neutr} and \ref{polesall} hold, then
\begin{equation}
\cE_\tau = E_\infty+\sum_{2\leq n+m\leq 5}\frac{\cF^{(n,m)}(U,V)}{L^{n+m+1}}
-\frac{C_{\rm vdW}(U,V)}{L^6}+O_{L^\infty}\left(\frac{1}{L^7}\right)
\label{eq:lowerupper}.
\end{equation}
\end{description}
\end{theorem}

\bigskip

   If, instead of Condition \ref{polesall}, we assume the weaker Condition \ref{poleslead}, we prove the following theorem providing a weaker form of equation \eqref{eq:lowerupper}, which intervenes in the proof of the main Theorem \ref{thm:with-multipoles} below. 
\begin{theorem}\label{thm:vanderwaalsbis}
	Assume Conditions \ref{neutr} and  \ref{poleslead}.   If $n_1 + n_2 <5$, then
	\begin{equation}\label{eq:thmbis}
	\cE_\tau = E_\infty +\frac{\cF^{(n_1,n_2)}(U,V)}{L^{n_1+n_2+1}} + O_{L^\infty}\left(\frac{1}{L^{n_1 + n_2 + 2}}\right).
	\end{equation}
\end{theorem}
The fact that we assume only \ref{poles} and not irreducibility as in \cite{al} for Theorems \ref{thm:vanderwaals} (b) and \ref{thm:vanderwaalsbis} creates some additional difficulties in the proof, see Section \ref{sec:lowerbound} below for explanations.

\begin{proof}[Proof of Theorems \ref{thm:vanderwaals} and \ref{thm:vanderwaalsbis}]
The proof of both  theorems can be found in Section \ref{sec:vanderwaals}. We derive an upper and lower bound in Subsections \ref{subsec:upper} and \ref{subsec:lower}, respectively, and in Subsection \ref{sec:lowerbound} we collect the results to prove the theorems.
\end{proof}

\bigskip

\subsubsection{The main theorem}

The most natural result we would like to have would be the existence of a bounded optimal path. However, we do not obtain it with our analysis. Nevertheless, we prove the existence of a bounded min-maxing sequence, as defined  below.
\begin{definition}\label{def:minmax}
 A \emph{min-maxing sequence} is a sequence of paths $$\{\tau_n(\cdot)\}\subset C^0([0,1]; \R^+\times \SO(3) \times \SO(3))$$
such that
 \begin{equation}\label{eq:minmax}
  \lim_{n\to\infty}\max_{t\in[0,1]}\cE_{\tau_n(t)}=c,
 \end{equation}
 where $c$ was defined in  \eqref{eq:mountain_pass}.
\end{definition}

Now we are ready to state our main theorem.   Recall that $n_1,n_2$ indicate the first nonvanishing multipole moments, defined in Definition \ref{def:FNVMM}, and in general they can depend on the ground state even if this is normalized. Under Condition \ref{poles}, the leading multipole moments $\cM^{(n_1)}_{\rho_{\Psi_1}}, \cM^{(n_2)}_{\rho_{\Psi_2}}$ of the molecules are independent of the choice of the normalized ground state and we will write instead $\cM^{(n_1)}_{1}, \cM^{(n_2)}_{2}$. 

\bigskip

\begin{theorem}\label{thm:with-multipoles}
Let us assume Condition~\ref{neutr}. Moreover, let us assume that Condition \ref{poleslead} holds and $n_1 + n_2 < 5$. In the case  $n_k=3$ for $k=1$ or $k=2$, we also assume the implication
\begin{equation}\label{ass:octapole}
\cM^{(n_k)}_{k}(v,\cdot,\cdot)\equiv 0\implies v=0.
\end{equation}

Then, one can find a bounded min-maxing sequence of paths  for the min-max problem~\eqref{eq:mountain_pass}: there exists $\lk>0$ and a min-maxing sequence of paths $\tau_n(t)=(L_n(t), U_n(t),V_n(t))$ such that
$$L_n(t)\leq \lk$$
for all $n \in\N$ and all $t\in[0,1]$.
\end{theorem}

\bigskip

\begin{figure}[t]
	\centering
	\includegraphics[width=9cm]{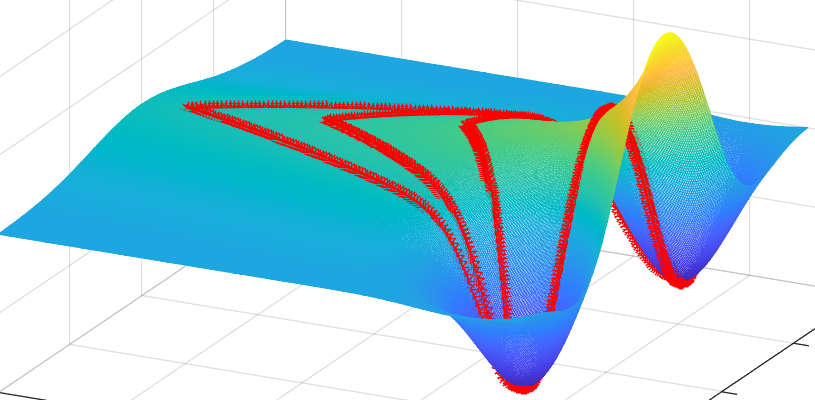}

	\caption{ We present here a hypothetical simplified graph of $\cE=\cE_\tau$. The paths which are drawn would be  the first terms of a min-maxing sequence. The distance of the submolecules along this sequence of paths cannot be uniformly bounded. As a consequence, the conclusion of Theorem \ref{thm:with-multipoles} \label{fig:unboundedness} would not hold. 
    Compare with Figure~\ref{fig:optpath}, where an optimal path exists. }

\end{figure}

\begin{remark}
    Theorem \ref{thm:with-multipoles}, aside from being an important step towards proving existence of an optimal path, also has a meaning on its own. To see this we refer to Figure \ref{fig:unboundedness} showing a hypothetical simplified graph of $\cE=\cE_\tau$ as an example in which a min-maxing sequence has to be unbounded. In this case, optimization of the energy would happen through infinite separation of the submolecules, a situation that would be unphysical. Moreover, a transition state would not exist. We recall, however, that Theorem \ref{thm:with-multipoles} together with Theorem 4 in \cite{Lewin-04b} implies existence of a transition state  in the \textbf{(NR)} case. As we mentioned above, we conjecture that Theorem 4 in \cite{Lewin-04b} holds in the \textbf{(SR)} case as well. If this is the case, then existence of a transition state follows similarly in the \textbf{(SR)} case. 
\end{remark}

\begin{remark} We believe that	the non-degeneracy assumption~\eqref{ass:octapole} on the octopole, in case this is the first non-vanishing multipole moment, is physical. We refer the interested reader to \cite{al} Lemma 1.9 and the related discussion. 
\end{remark}

 We have not been able to obtain a result if $n_1+n_2=5$, \textit{i.e.} when the leading multipole-multipole interaction has the same order as the van der Waals attraction. In the case where $n_1+n_2\geq6$, as soon as we have the static upper bound \eqref{eq:upper}, the following result of~\cite{al} can be adapted without any modification to the \textbf{(SR)} case, thus we will not give a proof of it here.

\begin{theorem}[\cite{al}, Theorem~1.7]
 Let us assume Condition~\ref{neutr} and that $\hat{H}_1, \hat{H}_2$ have ground states so that, for the leading multipole terms, we have $n_1 + n_2 \geq 6$.
 
 Then, one can find a bounded min-maxing sequence of paths $\tau_n(t)$ for the min-max problem~\eqref{eq:mountain_pass}.
\end{theorem}

 Even though the static information of Theorems \ref{thm:vanderwaals} and \ref{thm:vanderwaalsbis} is very useful, it is by no means enough to prove Theorem \ref{thm:with-multipoles} on the quasistatic behavior of the system, when the molecule slowly changes its shape. The next section describes the additional difficulties arising in the quasistatic analysis and the way that we have been able to overcome them.

\bigskip  
  
\subsection{Sketch of the proof of Theorem~\ref{thm:with-multipoles}}
In this section, we give the proof of Theorem~\ref{thm:with-multipoles}, but defer the proof of many intermediate results to the rest of the paper. The general strategy is the same as in \cite{al}. In this sketch, we will also briefly explain additional difficulties coming in the \textbf{(SR)} case.  

We will only consider the case where 
\begin{equation}\label{Ass:general}
	E_\infty=c,
\end{equation}
with $c, E_\infty$  defined in \eqref{eq:mountain_pass}, \eqref{eq:limit_L_infinity}, respectively.
If \eqref{Ass:general} is not fulfilled, then proving boundedness is easy and already known from previous works, see \emph{e.g.} \cite[Theorem~4]{Lewin-04b}. Roughly speaking, the reason is that if $E_{\infty} \neq c$, then the energy of the system does not go to the mountain pass level $c$  when the molecules go infinitely apart. Therefore, the mountain pass level $c$ should be attained at a finite distance.

  We want to prove that there exists a min-maxing sequence $(\tau_n)$ for the mountain pass problem~\eqref{eq:mountain_pass},   such that, for all $n$ and $t$, 
 \begin{equation}\label{LnLcut}
 L_n(t) \leq \lk,
 \end{equation}
 for an appropriate $\lk>0$. 
  Definition~\ref{def:minmax} and~\eqref{eq:mountain_pass} imply that there exists a min-maxing sequence
  $(\widetilde{\tau}_n(t))_{n \in \mathbb{N}}=(\widetilde{L}_n(t),\widetilde{U}_n(t),\widetilde{V}_n(t))_{n \in \mathbb{N}}$, for which $\widetilde{L}_n(t)$ might be unbounded. We will modify it to construct a min-maxing sequence $(\tau_n(t))_{n \in \mathbb{N}}=(L_n(t),U_n(t),V_n(t))_{n \in \mathbb{N}}$ such that \eqref{LnLcut} is fulfilled. We consider $\lk$ as a parameter and, for a given $n$, we
  start by looking  at the first time $t_0$ and the last time $t_1$ for which $\widetilde{L}_n(t)=\lk$. These times depend on $\lk$ and $n$, but we do not emphasize this in our notation. 

We will construct a path $\tau_n(\cdot)$  which coincides with $\widetilde{\tau}_n(\cdot)$ outside $(t_0,t_1)$ and for which  $L_n(t)\equiv \lk$ for all $t\in[t_0,t_1]$. The difficulty is to link the two molecular orientations $(\widetilde{U}_n(t_0),\widetilde{V}_n(t_0))$ and $(\widetilde{U}_n(t_1),\widetilde{V}_n(t_1))$ by a continuous path in $\SO(3)\times \SO(3)$ on which the energy does not increase too much. More precisely, we want to find, for $\lk$ large enough, a continuous path $(U(t),V(t))_{t\in[t_0,t_1]}$ of rotations such that
$$(U(t_0),V(t_0))=(\widetilde{U}_n(t_0),\widetilde{V}_n(t_0)),\qquad (U(t_1),V(t_1))=(\widetilde{U}_n(t_1),\widetilde{V}_n(t_1)),$$
and
\begin{equation}
 \max_{t\in[t_0,t_1]}\cE_{(\lk,U(t),V(t))}
 \leq \max\big\{\cE_{(\lk,\widetilde{U}_n(t_{0}),\widetilde{V}_n(t_{0}))},\cE_{(\lk,\widetilde{U}_n(t_{1}),\widetilde{V}_n(t_{1}))}, E_\infty\big\}.
 \label{eq:goal_path_rotations}
\end{equation}
An arbitrary path would not work here, since there are orientations for which
the molecules repel each other. As a consequence, the maximal energy along the path could be bigger than $E_\ii$ and thus \eqref{eq:goal_path_rotations} would not hold. This is the case when $\cF^{(n_1,n_2)}(U,V)>0$, see \eqref{eq:thmbis}. If \eqref{eq:goal_path_rotations} holds, then, due to \eqref{Ass:general} and \eqref{eq:minmax}, $\tau_n$ is going to be a min-maxing sequence fulfilling \eqref{LnLcut}, as desired. 

We proceed as follows: 

\subsubsection*{\bf Step 1. Replacing the two points by local (pseudo-)minima} 
Our first step is to rotate the molecules at each point $t_0,t_1$, without changing the distance $\lk$, so that we reach a local minimum of the energy of the system with respect to rotations. In general, this is impossible if, in some directions, the energy takes values higher and lower than that of the current point, at arbitrarily small distances. Since we cannot exclude this  pathological behavior, we introduce, as in \cite{al}, the weaker notion of local pseudo-minimum. 

\begin{definition}[Local pseudo-minimum]\label{def:locpseu}
Let $\cW$ be any continuous function on $\SO(3)\times \SO(3)$. A point $(\overline U, \overline V) \in \SO(3)\times \SO(3)$ is called a \emph{local pseudo-minimum of $\cW$} if, for any  $A,B \in \Gamma$, as defined in~\eqref{def:so3}, there exists a sequence $t_n\to 0+$ with $t_n\neq0$ such that
$\cW(e^{t_n A} \overline U, e^{t_n B} \overline V) \geq \cW(\overline U, \overline V)$.
\end{definition}

In other words, a local pseudo-minimum is a point at which one cannot find a direction in $\SO(3)^2$ in which  $\cW$ strictly decreases. If $\cW$ is $C^2$ and has a local pseudo-minimum at $(\overline{U},\overline{V})$, one must then have 
\begin{equation}\label{cond:pseudomin}\frac{\d}{\d t}\cW(e^{tA}\overline U,e^{tB} \overline V)\big|_{t=0}=0,\qquad \frac{\d^2}{\d t^2}\cW(e^{tA}\overline U,e^{tB} \overline V)\big|_{t=0}\geq0
\end{equation}
for all  $A,B \in \Gamma$, but $(\overline U,\overline V)$ does not need to be a true local minimum. 
A counterexample can be found even for functions of one variable: if $$f(x)=x^4\sin(1/x) \chi_{\{x<0\}}(x) + x^4 \chi_{\{x>0\}}(x), \qquad x \in \R,$$
where $\chi_C$ denotes the characteristic function of the set $C$, then $f$ has a local pseudo-minimum at zero but not a local minimum at 0. 

The fact that we can reach a local pseudo-minimum can be shown with the help of the following lemma, already proven in \cite{al}. We state it for simplicity on $\SO(3)\times \SO(3)$ but it is  more general.

\begin{lemma}[\cite{al}, Lemma~1.10: linking any point to a local pseudo-minimum]\label{lem:flowpseudolocmin}
Let $\cW$ be a continuous function on $\SO(3)\times \SO(3)$. Let $(U,V)$ be a point that is not a local pseudo-minimum of $\cW$. Then there exists a continuous path  on $\SO(3)\times \SO(3)$ linking $(U,V)$ to a local pseudo-minimum $(\overline U,\overline V)$, such that the maximum value of $\cW$ on this path is $\cW(U,V)$.
\end{lemma}
 
We will apply these results to the function $\cE_{\lk}$, where, for $L>2\max\{|y_j|:j \in \{1,\dots,M\}\}$,
\begin{equation}\label{ELUV}
\cE_{L}(U,V):=\cE_{(L,U,V)}.
\end{equation}
Note that the last restriction on $L$ makes sure that no nuclei coincide.  
Using Lemma \ref{lem:flowpseudolocmin},  we can connect the point $(\widetilde{U}_n(t_0), \widetilde{V}_n(t_0))$ to a local pseudo-minimum $(\overline U_0,\overline V_0)$ of $\cE_{\lk}$, with a path along which the energy stays below the value $\cE_{(\lk,\widetilde{U}_n(t_0), \widetilde{V}_n(t_0))}$. We proceed analogously for $t_1$ and connect $(\widetilde{U}_n(t_1),\widetilde{V}_n(t_1))$ to a local pseudo-minimum $(\overline U_1,\overline V_1)$ of $\cE_{\lk}$, with a path along which the energy stays below $\cE_{(\lk,\widetilde{U}_n(t_1), \widetilde{V}_n(t_1))}$. In the next steps, we will study the properties of these local pseudominima and finally connect them by a path, without increasing the energy too much.

\medskip

\subsubsection*{\bf Step 2. Properties of the leading multipolar interaction at the local pseudominima $(\overline U_0,\overline V_0)$ and $(\overline U_1,\overline V_1)$}
In Step 1, we have managed to reach two points $(\overline U_0,\overline V_0)$ and $(\overline U_1,\overline V_1)$ which are local pseudo-minima of the energy $\cE_{\lk}$. In this step we prove that 
\begin{equation}\label{Eestimates}\cE_{(\lk,\overline U_0,\overline V_0)}<E_\ii,\qquad \cE_{(\lk,\overline U_1,\overline V_1)}<E_\ii,\end{equation}
independently of $n$, provided we had chosen $\lk$ large enough. This means that
 the configurations $(\lk,\overline U_0,\overline V_0),(\lk,\overline U_1,\overline V_1)$ have an energy strictly below the dissociation threshold $c=E_\ii$ -- recall \eqref{eq:mountain_pass} and \eqref{Ass:general}. More precisely, we show that $\exists \delta>0$ such that
  \begin{equation}\label{Festimates}
\cF^{(n_1,n_2)}(\overline U_0,\overline V_0) \leq -\delta,\qquad \cF^{(n_1,n_2)}(\overline U_1,\overline V_1) \leq -\delta,
\end{equation}
where we recall that $n_1$ and $n_2$ are the indices for the first non-vanishing multipoles of the two neutral submolecules, defined in Definition~\ref{def:FNVMM} and the $\cF^{(n_1,n_2)}$ are the multipolar interactions defined by~\eqref{multipol-coeff}. The estimate \eqref{Eestimates} follows  from \eqref{Festimates} and the expansion in Theorem~\ref{thm:vanderwaalsbis}. We now explain how to prove \eqref{Festimates}. We first need the following proposition.

\medskip

\begin{proposition}[Leading multipolar energy at a local pseudo-minimum]\label{prop:fullexp}
Assume that Condition \ref{poleslead} holds. Then, for any $\delta >0$, there exists an $L_m>0$  such that, if $L>L_m$ and $(\overline U,\overline V)$ is any local pseudo-minimum of $\cE_L$, defined in \eqref{ELUV}, then we have 

\begin{equation}\label{eq:almostlocminSR}
\left|\frac{\d}{\d t}\cF^{(n_1,n_2)}(e^{tA}\overline U,e^{tB}\overline V)\big|_{t=0}\right| \leq \delta, \qquad \frac{\d^2}{\d t^2}\cF^{(n_1,n_2)}(e^{tA}\overline U,e^{tB}\overline V) \big|_{t=0} \geq -\delta
\end{equation}
 for all   $A,B \in \Gamma$, recall \eqref{def:so3}.

\end{proposition}

\begin{proof}
The proof can be found in Section \ref{sec:firstderivative}.
\end{proof}

\medskip
Proposition \ref{prop:fullexp} makes it possible to extract information on the change of the ground state energy with respect to the orientations close to a local pseudo-minimum.
Intuitively, it says that a local minimum of $\cE_L$ is close to being a local minimum of the leading multipole-multipole term $\cF^{(n_1,n_2)}$. This was already proven in \cite{al} in the \textbf{(NR)} case under an irreducibility assumption on the ground state eigenspaces $\cG_1, \cG_2$. 
Nevertheless, it has not been possible for us to adapt these arguments to the \textbf{(SR)} case. Indeed,  in the \textbf{(NR)} case, the eigenstates are in $H^2$. In the \textbf{(SR)} case, it is not even clear that they are in $H^1$. In addition, the nonlocality of the kinetic energy operator makes things more complicated too.
 We also emphasize the fact that, for the results of \cite{bhv}, it was not necessary to differentiate the energy and therefore their methods are not applicable for the proof of Proposition \ref{prop:fullexp}.  
 We overcome these difficulties by using ideas in~\cite{Hunziker-86}, that Hunziker used to prove that the ground state energy of $H_N(Y,\cZ)$ depends analytically on $Y$ in the \textbf{(NR)} case and if the fermionic statistics is ignored. We rework these ideas, so that they provide quantitative information for derivatives and are applicable in the  \textbf{(SR)} case as well. Our proof of Proposition \ref{prop:fullexp}
 in the \textbf{(NR)} case is simpler than the one in \cite{al} and provides information for higher order derivatives as well. Modifying an idea of \cite{al} we could take the fermionic statistics into account.

 The rest  of the proof is the same as in ~\cite{al}, as Propositions \ref{prop:localmin}
 and \ref{prop:connected} below were proven in \cite[Section~4]{al} independently of the kind of the kinetic energy and of the irreducibility assumption. Proving Propositions  \ref{prop:localmin}
 and \ref{prop:connected} was the hardest part of the analysis of \cite{al} but, since there are no modifications needed in our setting, we will only state them and explain how we can conclude  the proof with them. 
The first result is rather intuitive, since the average of $\cF^{(n_1,n_2)}$ over all orientations is zero.

\begin{proposition}[\cite{al}, Proposition~1.12]\label{prop:localmin}
Let $m$, $n\in\N$ with $n+m\in\{2,3,4\}$.
If $n$ or $m$ is $3$, assume that the condition on the octopole moment stated in Theorem~\ref{thm:with-multipoles} holds. Then 
there exists a $\delta>0$  such that,  if  $\cF^{(n,m)}$ fulfills \eqref{eq:almostlocminSR}, then
\begin{equation*}
\cF^{(n,m)}(\overline U,\overline V) \leq -\delta. 
\end{equation*}
\end{proposition}

\medskip

Using Propositions~\ref{prop:fullexp} and~\ref{prop:localmin}, we obtain that, if $\lk$ is large enough, there exists $\delta>0$ such that \eqref{Festimates} holds
and, in view of~\eqref{eq:thmbis},
\begin{equation*}
	\cE_{(\lk,\overline U_0,\overline V_0) }\leq E_\infty- \frac{\delta}{2}, \qquad   	\cE_{(\lk,\overline U_1,\overline V_1)} \leq E_\infty- \frac{\delta}{2},
\end{equation*}
provided that $\lk$ is large enough. 
\medskip

\subsubsection*{\bf Step 3. Linking the local pseudominima  $(\overline U_0,\overline V_0)$ and $(\overline U_1,\overline V_1)$ with a path of low energy}
 We have so far managed to show that the leading multipole-multipole interaction term of the two new points $(\overline U_0,\overline V_0)$ and $(\overline U_1,\overline V_1)$ is negative, away from zero.  We  connect $(\overline U_0,\overline V_0)$ and $(\overline U_1,\overline V_1)$ with a path along which $\cF^{(n_1,n_2)}$ stays negative, which  implies, if $\lk$ is large enough, that  $\cE_{(\lk,U(t),V(t))}<E_\infty$ for all $t$, by Theorem~\ref{thm:vanderwaalsbis}. Note that, for this part, the non-degeneracy assumption \eqref{ass:octapole} for the octopole moment is not needed. The existence of such a path is given in the following proposition. 

\begin{proposition}[\cite{al} Proposition 1.13, Connectedness of $\{\cF^{(n,m)}\leq -\delta\}$]\label{prop:connected}
Let $n,m \in \mathbb{N}$ with $n+m \in \{2,3,4\}$. Then there exists $\delta_0>0$ such that, for all $0<\delta<\delta_0$, the set $\{(\U,\V) \in \SO(3): \cF^{(n,m)}(\U,\V) <-\delta\}$ is nonempty and pathwise connected.
\end{proposition}

Thus, we can connect $(\overline U_0,\overline V_0)$ and $(\overline U_1,\overline V_1)$ with a path along which the leading term of the interaction energy is negative.
If $\lk$ is large enough, then the difference $\cE_{(\lk,U(t),V(t))}-E_\infty$ remains negative along the path connecting $(\overline U_0,\overline V_0)$ and $(\overline U_1,\overline V_1)$.
We have thus found a new sequence
$\{\tau_n\}_{n \in \mathbb{N}}$, $\tau_n = (L_n, U_n, V_n)$ where $L_n \leq L_{\text{cut}}$ for any $n \in \N$ and for which \begin{equation}\label{eq:limsequence}\max_{t\in[0,1]}\cE_{\tau_n(t)}\leq\max(E_\infty,\max_{t\in[0,1]}\cE_{\widetilde{\tau}_n(t)}).\end{equation}

But, since we have assumed that $(\widetilde{\tau}_n)_{n \in \mathbb{N}}$ is a min-maxing sequence, we know that \[\lim_{n\to\infty}\max_{t\in[0,1]}\cE_{\widetilde{\tau}_n(t)}=c.\] 
We have assumed too, in~\eqref{Ass:general}, that $E_\infty=c$. We arrive thus from~\eqref{eq:limsequence} at
\[\limsup\max_{t\in[0,1]}\cE_{\tau_n(t)}\leq c, 
\] which implies that $(\tau_n)_{n \in \mathbb{N}}$ is a min-maxing sequence as well. 
\qed

\begin{figure}[t]
\centering
\includegraphics[width=10cm]{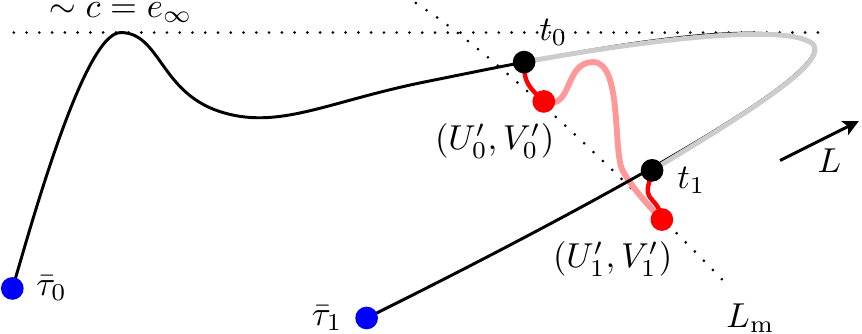}
\caption{Graphical representation of the proof of Theorem~\ref{thm:with-multipoles} in  Case 1 (the picture was taken from \cite{al} and was created by Mathieu Lewin).\label{fig:cut_path2}}
\end{figure}

\bigskip

\section{Proof of Theorems \ref{thm:vanderwaals} and \ref{thm:vanderwaalsbis}}\label{sec:vanderwaals}

\subsection{Upper bound: Proof of Theorem \ref{thm:vanderwaals}(a)}\label{subsec:upper}

In this section, we prove the first part of Theorem \ref{thm:vanderwaals} in the \textbf{(SR)} case. The analogous result for the \textbf{(NR)} case has already been proved in Section 2.2 of \cite{al}. We estimate the  interaction energy $\cE_\tau-E_1-E_2$ from above by the associated quadratic form calculated on a  suitable trial function. 

Let us introduce
\begin{equation*}
H_{\infty,\tau}:= H_{1,\tau} \otimes \mathbbm{1}^{\otimes 3N_2} +  \mathbbm{1}^{\otimes 3N_1}  \otimes H_{2,\tau}, 
\end{equation*}
with $H_{j,\tau}$, defined  in \eqref{def:H1tau}, \eqref{def:H2tau}:   $H_{\infty,\tau}$ is the Hamiltonian describing the 2 submolecules without interaction between each other. 
It will often be  useful to work in coordinates in which  $H_{\infty,\tau}$ does not depend on $\tau$.
For that reason, we introduce the unitary transformation $K_\tau: L^2(\R^{3N}) \to L^2(\R^{3N})$ defined by
\begin{align}\nonumber
(K_\tau \Psi)(x_1,\dots,x_{N_1},x_{N_1+1}, \dots,& x_N)\\ :=\label{def:Ktau}\Psi(U^{-1}x_1,\dots,U^{-1}x_{N_1}, & V^{-1} (x_{N_1+1}-Le_1),\dots,V^{-1} (x_{N}-Le_1)).
\end{align}
 In words, $K_\tau$ applies to the first $N_1$ electrons the transformation applied by $\tau$ to the nuclei of the first molecule and to the last $N_2$   electrons the transformation applied by $\tau$ to the nuclei of the second molecule. 
 We observe that
\begin{equation}\label{HinftytauHinfty}
 H_\infty = K_\tau^{-1} H_{\infty,\tau} K_\tau, 
\end{equation}
where we recall that $H_\infty$ was defined in   \eqref{def:Hinfty}. In light of \eqref{HinftytauHinfty} and \eqref{def:Htau}, it makes sense to introduce
\begin{equation}\label{def:Itilde}
\tilde{I}_{\tau} :=  K_\tau^{-1} I_\tau  K_\tau,
\end{equation}
which depends on $\tau$ even though the left-hand side of \eqref{HinftytauHinfty} does not. Explicitly, we have
\begin{multline}
\tilde{I}_{\tau} = -\sum_{i=1}^{N_1}\sum_{l=M_1+1}^M\frac{Z_l \alpha}{|U x_i - Vy_l - L e_1|} - \sum_{j=N_1 + 1}^N \sum_{k=1}^{M_1}\frac{ Z_k \alpha}{|U y_k-Vx_j -Le_1|} \\
+\sum_{k=1}^{M_1} \sum_{l={M_1+1}}^{M} \frac{ Z_l Z_k \alpha}{|U y_k-V y_l -Le_1|}+\sum_{i=1}^{N_1} \sum_{j=N_1 + 1}^N \frac{\alpha}{|Ux_i - Vx_j-Le_1|}. \label{eq:interactionmodified}
\end{multline}
Recall that $\chi_L$ was defined in \eqref{def:chiL}.
 We define on $L^2(\R^{3N})$ the operators
\begin{equation}\label{def:xi}
\xi_\tau:= (H_{\infty}^\bot -E_1-E_2)^{-1} \Pi^\bot \tilde{I}_{\tau} \,\chi^{\otimes N}_{L},
\end{equation}
and 
\begin{equation}\label{def:Wtau}
W_\tau:=\chi_L^{\otimes N}-\chi_{4L/3}^{\otimes N}\xi_\tau=\Big(1 - \chi_{4L/3}^{\otimes N} (H_\infty^\bot -E_1-E_2)^{-1} \Pi^\bot \tilde{I}_{\tau}\Big) \chi_L^{\otimes N},
\end{equation}
where $H_{\infty}^\bot, \Pi^\bot $ were defined in~\eqref{def:Hinftyperp} and below \eqref{def:Hinftyperp}, respectively.

We consider a normalized function $\Psitp \in \cG_1 \otimes \cG_2$, where we recall that $\cG_1, \cG_2$ were defined in \eqref{def:GSEigenspace}. For this section, the interesting case is $\Psitp= \Psi_1 \otimes \Psi_2$, where
 $\Psi_1$ and $\Psi_2$  are  in $\cG_1$ and $\cG_2$, respectively.  However, as most of our calculations are useful for the next sections for general normalized  $\Psitp  \in \cG_1 \otimes \cG_2$,
 we will work in this section for general $\Psitp$ as well, and we will specify the parts of the calculations for which the factorized form $\Psitp= \Psi_1 \otimes \Psi_2$ is necessary. 
  As we will prove below, in~\eqref{eq:normfunctionorder}, for $\Psi\in\cG_1\otimes\cG_ 2$, $W_\tau\Psi$ is a small perturbation of $\Psi$. More precisely, it is an approximation of the ground state given by the Feshbach map, see the brief sketch of the proof of Theorem~1.5 in~\cite{Anapolitanos-16}. The cut-off functions make some calculations easier.
 
 Since ground states of $\hat{H}_1$ and $\hat{H}_2$ have at least $H^{1/2}$ regularity, $\Psi \in \cG_1 \otimes \cG_2$ has  $H^{1/2}$ regularity as well  and, by~\cite[Theorem~7.18]{LL}, this is the case for $\chi_L^{\otimes N}\Psi$ too. Since the function $\tilde{I}_\tau$  as well as all its derivatives are bounded on the support of $ \chi_L^{\otimes N}\Psitp$, we see by~\eqref{def:Wtau} that $ W_\tau \Psitp$ is the sum of 2 terms, each of them being in $H^{1/2}(\R^{3N})$. Since $K_\tau$,  defined in~\eqref{def:Ktau}, is unitary in $H^{1/2}(\R^{3N})$,   $K_\tau  W_\tau \Psitp$ is in $H^{1/2}(\R^{3N})$ as well, so it is in the form domain of $H_\tau$, see \eqref{def:domainHamilton}.  As a consequence, 
 	\begin{equation}\label{testh12}
\frac{Q_N K_\tau  W_\tau \Psitp}{\|Q_N K_\tau  W_\tau \Psitp\|} \in H^\frac{1}{2}(\R^{3N}),
   \end{equation}
 where $Q_N$ is defined by~\eqref{def:Q}, can be used as a test function. A similar test function was used in \cite{Anapolitanos-16} and \cite{al} with different notation. 
By definition of the ground state energy, see \eqref{def:Etau}-\eqref{def:Htau}, we have that
\begin{align}\label{ineq:Etau}
\mathcal{E}_\tau -E_1-E_2 &\leq \frac{\langle Q_N K_\tau  W_\tau \Psitp , (H_{\tau} -E_1-E_2) Q_N K_\tau  W_\tau \Psitp\rangle}{\|Q_N K_\tau  W_\tau \Psitp\|^2}.    
\end{align}
 Note that $Q_N$ does not commute with $K_\tau$.  

We first get rid of the antisymmetrizer $Q_N$ in the right-hand side of \eqref{ineq:Etau}. How to do this is not as obvious as in the \textbf{(NR)} case since the Hamiltonian is not local. For this reason, we will explain it in detail. 
First, since $Q_N$ is a projection and it commutes with $H_\tau-E_1-E_2$, \[
 \langle Q_N K_\tau  W_\tau \Psitp , (H_{\tau} -E_1-E_2) Q_N K_\tau  W_\tau \Psitp\rangle = \langle Q_N K_\tau  W_\tau \Psitp , (H_{\tau} -E_1-E_2)  K_\tau  W_\tau \Psitp\rangle.                                                                
                                                                \]
In view of~\eqref{def:Q}, we can write
\begin{equation}\label{Q:0case}
 \langle Q_N  K_\tau  W_\tau \Psitp , (H_{\tau} -E_1-E_2)   K_\tau W_\tau \Psitp\rangle=\sum_{\sigma\in S_N}\frac{(-1)^\sigma}{N!}\langle \sigma\cdot  K_\tau  W_\tau \Psitp , (H_{\tau} -E_1-E_2)  K_\tau  W_\tau \Psitp\rangle,
\end{equation}
where we recall that $\sigma\cdot K_\tau W_\tau \Psitp$ has been defined in~\eqref{def:piPsi}.

We now have to distinguish between two cases: when $\sigma(\{1,...,N_1\})=\{1,...,N_1\}$, \textit{i.e.} when the permutation $\sigma$ leaves invariant the set of electrons associated with each of the submolecules, and when $\sigma(\{1,...,N_1\})\neq\{1,...,N_1\}$.

In  the first case, we have by construction, for such a $\sigma$, that $\sigma\cdot K_\tau  W_\tau \Psitp=(-1)^\sigma K_\tau  W_\tau \Psitp$. Indeed, the ground states $\Psi_1$ and $\Psi_2$ are antisymmetric and $H_\infty^\bot$, $K_\tau$ and $W_\tau$   are invariant under these permutations of variables. 
As a consequence
\begin{multline}\sum_{\sigma(\{1,...,N_1\})=\{1,...,N_1\}}(-1)^\sigma\langle \sigma\cdot K_\tau  W_\tau \Psitp , (H_{\tau} -E_1-E_2)  K_\tau  W_\tau \Psitp\rangle\\=N_1!N_2!\langle  K_\tau  W_\tau \Psitp , (H_{\tau} -E_1-E_2)  K_\tau  W_\tau \Psitp\rangle.\label{Q:1case}\end{multline}

On the other hand, if $\sigma(\{1,...,N_1\})\neq\{1,...,N_1\}$, then, due to   the cut-offs $\chi_{L}^{\otimes N}$, $\chi_{4L/3}^{\otimes N}$ and the translation in the definition of $K_\tau$ in \eqref{def:Ktau},  the functions $ K_\tau  W_\tau \Psitp$ and $\sigma\cdot K_\tau  W_\tau \Psitp$ have disjoint supports. Observing additionally that the only non-local terms in $H_\tau-E_\infty$ are the kinetic energy operators, we find that 
\begin{align}\label{nonlocal}\langle \sigma\cdot K_\tau  W_\tau \Psitp , (H_{\tau} -E_1-E_2)  K_\tau  W_\tau \Psitp\rangle =\sum_{j=1}^N\langle \sigma\cdot K_\tau  W_\tau \Psitp , T_j K_\tau  W_\tau \Psitp\rangle.\end{align}
But, for each $\sigma$ and $j$, we can find $k\neq j$ such that, either $k\in\{1,...,N_1\}$ and $\sigma(k)\in \{N_1+1,...,N\}$, or  $k\in\{N_1+1,...,N\}$ and $\sigma(k)\in\{1,...,N_1\}$. Thus, $  K_\tau  W_\tau \Psitp$ and $\sigma\cdot K_\tau  W_\tau \Psitp$ have disjoint supports in the $k$-th variable and, as a consequence,  
\begin{equation}\label{Q:2case}
 \langle \sigma\cdot K_\tau  W_\tau \Psitp , T_j   K_\tau  W_\tau \Psitp\rangle=0, \qquad \forall j \in \{1,\dots,N\}. 
 \end{equation}
 From~\eqref{Q:0case}, \eqref{Q:1case}, \eqref{nonlocal} and \eqref{Q:2case}, we find that
 \begin{equation}\langle Q_N K_\tau  W_\tau \Psitp , (H_{\tau} -E_1-E_2) Q_N  K_\tau  W_\tau \Psitp\rangle=\frac{\langle  K_\tau  W_\tau \Psitp , (H_{\tau} -E_1-E_2)   K_\tau  W_\tau \Psitp\rangle}{\binom{N}{N_1}}.\label{end-antisym}\end{equation}
A similar but easier calculation gives
\begin{equation}\label{norm-QNphi}\|Q_N K_\tau  W_\tau \Psitp\|^2=\frac{\| K_\tau  W_\tau \Psitp\|^2}{\binom{N}{N_1}}.\end{equation}

With~\eqref{end-antisym} and~\eqref{norm-QNphi}, we find 
\begin{align}\label{eq:getridofQ}\frac{\langle Q_N K_\tau   W_\tau \Psitp , (H_{\tau} -E_1-E_2) Q_N  K_\tau  W_\tau \Psitp\rangle }{\|Q_N  K_\tau  W_\tau \Psitp\|^2}    
 =\frac{\langle  K_\tau  W_\tau \Psitp , (H_{\tau} -E_1-E_2)  K_\tau  W_\tau \Psitp\rangle}{\| K_\tau  W_\tau \Psitp\|^2},\end{align}
which, using the fact that $K_\tau$ is unitary together with 
 \eqref{HinftytauHinfty}, \eqref{def:Itilde} and \eqref{def:Htau}, gives 
\begin{align} \label{gs:trialfunction} \frac{\langle Q_N K_\tau  W_\tau \Psitp , (H_{\tau} -E_1-E_2) Q_N  K_\tau  W_\tau \Psitp\rangle    }{\|Q_N  K_\tau  W_\tau \Psitp\|^2}
   =\frac{\langle   W_\tau \Psitp , (H_{\infty}  +\tilde{I}_\tau-E_1-E_2)  W_\tau \Psitp\rangle}{\|  W_\tau \Psitp\|^2}.\end{align}

In the following, we compute the right-hand side of~\eqref{gs:trialfunction}.  Moreover, we prove that it depends smoothly on $U$ and $V$, which will be very important for Section~\ref{sec:firstderivative}.

 We remark that, in $ W_\tau \Psitp$, the rotations $U$ and $V$ are present only in the interaction $\tilde{I}_\tau$ in $\xi_\tau \Psitp$.

Recall that $\Pi$ is the orthogonal projection on $\cG_1 \otimes \cG_2$. Since all functions on $\Ran \Pi$ are exponentially decaying in the sense of Corollary \ref{cor:expeigendecay}, we have 
\begin{equation}
 [\Pi,\chi_{4L/3}^{\otimes N}]=O_{(C^\infty,\cB(L^2))}(e^{-dL}), \quad
 [\Pi,\chi_{L}^{\otimes N}]=O_{(C^\infty,\cB(L^2))}(e^{-dL}) \label{commut-project},
\end{equation}
where $\cB(L^2)$ denotes the bounded operators on $L^2(\R^{3N})$ and we recall the notation \eqref{C2error}.
The $C^\infty$ estimate is here totally for free, since the left-hand sides do not depend on $U$ and $V$.
 Since
 by~\eqref{def:xi} \begin{equation}
  \Pi^\bot\xi_\tau=\xi_\tau\label{Pixi},         
            \end{equation} 
  we can use \eqref{commut-project} to find that 
\begin{align}\label{psiRtau}
 \langle\chi_L^{\otimes N}\Psitp,\chi_{4L/3}^{\otimes N} \xi_\tau \Psitp\rangle= \langle\chi_L^{\otimes N}\Pi^\perp\Psitp,\chi_{4L/3}^{\otimes N} \xi_\tau \Psitp\rangle + O_{C^\infty}(e^{-dL})= O_{C^\infty}(e^{-dL}),
\end{align}
where in the last step we have used that $\Psitp\in\Ran \Pi$. From \eqref{def:Wtau} and \eqref{psiRtau},
\begin{align}
\|  W_\tau \Psitp\|^2 &= \|\chi_L^{\otimes N} \Psitp\|^2 + \|\chi_{4L/3}^{\otimes N} \xi_\tau \Psitp\|^2 +  O_{C^\infty}(e^{-dL}) \label{eq:normfunction}.
\end{align}
 By Corollary~\ref{cor:expeigendecay}, there exists some $c>0$ such that
 \begin{equation}\label{expdecayGS}
  e^{c|\cdot|}\Psitp\in L^2(\R^{3N}).
 \end{equation}
Consequently, as $\Psi$ is normalized,
 \begin{equation}\label{norm:GS}
 \chi_L^{\otimes N} \Psitp=\Psitp+O_{(C^\infty,L^2(\R^{3N}))}(e^{-dL}).
 \end{equation}
The $C^\infty$ estimate follows again, since the left-hand side does not depend on $U$ and $V$.
Using \eqref{Taylorsatz} for $N=2$ it follows that
\begin{equation}\label{interleadterm1}
 \tilde{I}_{\tau} \chi_L^{\otimes N} \Psitp=  \frac{f_{(U,V)}}{L^3} \chi_L^{\otimes N} \Psitp+ O_{(C^\infty, L^2(\R^{3N}))}\left( \frac{1}{L^4} \right),
\end{equation}
 where $f_{(U,V)}$ is defined in  \eqref{expansion:fUV}. See \textit{e.g.} in Equation~(70) of~\cite{al}, for a derivation of the leading term. The $C^\infty$ smallness of the error term can be proven similarly as equation \eqref{eq:KCinfty}.
 
 Using \eqref{interleadterm1} and \eqref{norm:GS}, we arrive at
 \begin{equation}\label{interleadterm}
 \tilde{I}_{\tau} \chi_L^{\otimes N} \Psitp=  \frac{f_{(U,V)}}{L^3}  \Psitp+ O_{(C^\infty, L^2(\R^{3N}))}\left( \frac{1}{L^4} \right).
 \end{equation}
Since $(H_\infty^\bot-E_1-E_2)^{-1} \Pi_\tau^\bot$ is bounded from $L^2(\R^{3N})$ to itself and independent of $\tau$,
Equations~\eqref{def:xi} and~\eqref{interleadterm} imply that
\begin{equation}\label{norm:Rtau} \xi_\tau \Psitp=O_{(C^\infty, L^2(\R^{3N}))}\left( \frac{1}{L^3}\right).\end{equation}
Combining~\eqref{eq:normfunction},~\eqref{norm:GS}, and~\eqref{norm:Rtau}, we obtain
\begin{align}
 W_\tau \Psitp=\Psi +  O_{(C^\infty,L^2(\R^{3N}))}\left( \frac{1}{L^3}\right), \qquad \| W_\tau \Psitp\|^2=1 +  O_{C^\infty}\left( \frac{1}{L^6}\right).\label{eq:normfunctionorder}
\end{align}

We now calculate the numerator of the right-hand side of~\eqref{gs:trialfunction}. We will start with a decomposition of 
 $(H_\infty+\tilde{I}_\tau -E_1-E_2) W_\tau \Psitp$. Using \eqref{def:Wtau}, we find
 \begin{equation}\label{eq:HIE}
 (H_\infty +\tilde{I}_\tau-E_1-E_2) W_\tau \Psitp= (H_\infty+\tilde{I}_\tau -E_1-E_2) \chi_L^{\otimes N}  \Psitp - (H_\infty+\tilde{I}_\tau -E_1-E_2) \chi_{4L/3}^{\otimes N} \xi_\tau \Psitp.
 \end{equation}

On the one hand, since $(H_\infty -E_1-E_2)\Psitp=0$ and $\chi_L^{\otimes N}$ commutes with the potentials, we have 
\begin{align}\label{eq:HminusEW}
(H_\infty +\tilde{I}_\tau-E_1-E_2) \chi_L^{\otimes N}  \Psitp=\sum_{j=1}^N[T_j, \chi_L^{\otimes N}]\Psitp +\tilde{I}_\tau  \chi_L^{\otimes N}  \Psitp  .\end{align}

On the other hand, 
\begin{align}\nonumber
 (H_\infty+\tilde{I}_\tau -E_1-E_2) \chi_{4L/3}^{\otimes N} \xi_\tau \Psitp &=\sum_{j=1}^N[T_j, \chi_{4L/3}^{\otimes N}] \xi_\tau  \Psitp
 \\ \label{eq:HIExi} &\quad + \chi_{4L/3}^{\otimes N} (H_\infty -E_1-E_2)  \xi_\tau \Psitp+  \tilde{I}_\tau \chi_{4L/3}^{\otimes N} \xi_\tau \Psitp.
 \end{align}

Using~\eqref{Pixi}  
and that $H_\infty$ commutes with $\Pi^\bot$, we find that
\begin{equation}\label{zeta2}
(H_\infty -E_1-E_2) \xi_\tau \Psitp = ( H_\infty^\bot-E_1-E_2) \xi_\tau \Psitp,
\end{equation}
where recall that $H_\infty^\bot$ was defined in \eqref{def:Hinftyperp}.
From~\eqref{def:xi} and \eqref{zeta2} we obtain that
\begin{align}\label{Hxi}
( H_\infty -E_1-E_2) \xi_\tau \Psitp=\Pi^\bot\chi_{L}^{\otimes N} \tilde{I}_\tau 	\Psitp=\chi_{L}^{\otimes N} \tilde{I}_\tau 	\Psitp - \Pi \chi_{L}^{\otimes N} \tilde{I}_\tau 	\Psitp.
\end{align}      
Since, due to \eqref{def:chiL}, $\chi_{4L/3}$ is equal to 1 on the support of $\chi_L$, we have
\begin{equation}\label{eq:chichi}
 \chi_{4L/3}\chi_L=\chi_L.
\end{equation}
The equalities \eqref{Hxi} and~\eqref{eq:chichi} imply 
\begin{align}\label{Hxi1}
\chi_{4L/3}^{\otimes N}( H_\infty -E_1-E_2) \xi_\tau \Psitp=\chi_{L}^{\otimes N} \tilde{I}_\tau 	\Psitp - \chi_{4L/3}^{\otimes N} \Pi \chi_{L}^{\otimes N} \tilde{I}_\tau 	\Psitp.
\end{align} 
Inserting \eqref{Hxi1} in \eqref{eq:HIExi}, we obtain 
\begin{align}\nonumber
(H_\infty+\tilde{I}_\tau -E_1-E_2) \chi_{4L/3}^{\otimes N} \xi_\tau \Psitp= \sum_{j=1}^N[T_j,  &
 \chi_{4L/3}^{\otimes N}] \xi_\tau  \Psitp + \chi_{L}^{\otimes N} \tilde{I}_\tau 	\Psitp
 \\  \label{eq:HminusE} & 
  - \chi_{4L/3}^{\otimes N} \Pi \chi_{L}^{\otimes N} \tilde{I}_\tau 	\Psitp +  \tilde{I}_\tau \chi_{4L/3}^{\otimes N} \xi_\tau \Psitp.
\end{align}
Inserting now~\eqref{eq:HminusEW} and~\eqref{eq:HminusE} in  \eqref{eq:HIE}, we arrive at 
\begin{equation}
(H_\infty + \tilde{I}_\tau -E_1-E_2) W_\tau  \Psitp=\sum_{j=1}^N[T_j, \chi_L^{\otimes N}]\Psitp - \sum_{j=1}^N[T_j, \chi_{4L/3}^{\otimes N}] \xi_\tau  \Psitp    + \chi_{4L/3}^{\otimes N} \Pi \chi_{L}^{\otimes N} \tilde{I}_\tau 	\Psitp - \tilde{I}_\tau \chi_{4L/3}^{\otimes N} \xi_\tau \Psitp.\label{eq:HminusEWdec}
\end{equation}
We remark that
 \begin{equation*}   
\left\|\sum_{j=1}^N[T_j ,\chi_L^{\otimes N}]\Psi\right\|^2=\left\langle \sum_{j=1}^N[T_j ,\chi_L^{\otimes N}]\Psi,\sum_{j=1}^N[T_j \,,\,\chi_L^{\otimes N}]\Psi\right\rangle.
 \end{equation*}
 But, according to  Lemma~\ref{lem:expcommut1overL}, the left-hand factor of the scalar product is in $L^2$ and its norm can be controlled uniformly in $L$. Therefore, due to \eqref{expdecayGS}, we can apply  \eqref{expsmall} and find that 
 \begin{equation}\label{eq:TcomChiPsi}
\sum_{j=1}^N[T_j ,\chi_L^{\otimes N}]\Psi=O_{(C^\infty, L^2(\R^{3N}))}(e^{-dL}).
 \end{equation} Here, the $C^\infty$ regularity is trivial since the left-hand side does not depend on $U$ and $V$. Using  \eqref{eq:HminusEWdec}-\eqref{eq:TcomChiPsi} we find
\begin{multline}
(H_\infty + \tilde{I}_\tau -E_1-E_2) W_\tau  \Psitp=- \sum_{j=1}^N[T_j, \chi_{4L/3}^{\otimes N}] \xi_\tau  \Psitp    + \chi_{4L/3}^{\otimes N} \Pi \chi_{L}^{\otimes N} \tilde{I}_\tau 	\Psitp  \\ - \tilde{I}_\tau \chi_{4L/3}^{\otimes N} \xi_\tau \Psitp + O_{(C^\infty, L^2(\R^{3N}))}(e^{-dL})\label{est:HIEW}.
\end{multline}

In order to estimate the numerator of the right-hand side of~\eqref{gs:trialfunction}, we first take the scalar product of the terms of the right-hand side of~\eqref{est:HIEW} with $\chi_{4L/3}^{\otimes N}\xi_\tau\Psitp$.
We begin by observing that, thanks to Lemma~\ref{lem:expcommut1overL} and~\eqref{norm:Rtau}, 
\begin{equation}\label{eq:OL4-1}
 \sum_{j=1}^N[T_j, \chi_{4L/3}^{\otimes N}] \xi_\tau  \Psitp=O_{(C^\infty, L^2(\R^{3N}))}\left(\frac{1}{L^4}\right)
\end{equation}
which, using~\eqref{norm:Rtau} once more, gives
\begin{equation}\label{PSxi1}
 \left\langle \chi_{4L/3}^{\otimes N}\xi_\tau\Psitp,\sum_{j=1}^N[T_j, \chi_{4L/3}^{\otimes N}] \xi_\tau  \Psitp\right\rangle=O_{C^\infty}\left(\frac{1}{L^7}\right).
\end{equation}
We find from the expression of $\tilde{I}_\tau$ in~\eqref{eq:interactionmodified} that\[\tilde{I}_\tau\chi_{4L/3}^{\otimes N}=O_{(C^\infty, L^{\infty}(\R^{3N}))}\left(\frac{1}{L}\right),\] which combined with~\eqref{norm:Rtau}, gives that
\begin{equation}\label{eq:OL4-2}
 \tilde{I}_\tau \chi_{4L/3}^{\otimes N} \xi_\tau \Psitp=O_{(C^\infty, L^2(\R^{3N}))}\left(\frac{1}{L^4}\right).
\end{equation}
\eqref{norm:Rtau} and \eqref{eq:OL4-2} imply
\begin{equation}\label{PSxi2}
\langle \chi_{4L/3}^{\otimes N} \xi_\tau \Psitp,  \tilde{I}_\tau \chi_{4L/3}^{\otimes N} \xi_\tau \Psitp\rangle=O_{C^\infty}\left(\frac{1}{L^7}\right).
\end{equation}

Moreover, by~\eqref{Pixi},
\begin{equation}\label{PSxi3}
 \langle \chi_{4L/3}^{\otimes N} \xi_\tau \Psitp, \chi_{4L/3}^{\otimes N} \Pi \chi_{L}^{\otimes N} \tilde{I}_\tau 	\Psitp\rangle =\langle \Pi^\bot \xi_\tau \Psitp, [(\chi_{4L/3}^{\otimes N})^2, \Pi] \chi_{L}^{\otimes N} \tilde{I}_\tau 	\Psitp\rangle =O_{C^\infty}(e^{-dL})
\end{equation}
since $[(\chi_{4L/3}^{\otimes N})^2, \Pi]$ can be bounded as in~\eqref{commut-project}. In the last step we have also used \eqref{interleadterm} and \eqref{norm:Rtau}. 

Collecting~\eqref{PSxi1}, \eqref{PSxi2} and \eqref{PSxi3}, we find, with the help of \eqref{norm:Rtau} and~\eqref{est:HIEW}, that
\begin{equation}\label{PSxitot}
 \langle \chi_{4L/3}^{\otimes N} \xi_\tau \Psitp,(H_\infty + \tilde{I}_\tau -E_1-E_2) W_\tau  \Psitp\rangle=O_{C^\infty}\left(\frac{1}{L^7}\right).
\end{equation}

Taking now a scalar product of the terms of the right-hand side of~\eqref{est:HIEW} with $\chi_L^{\otimes N} \Psi$, we first  find
\begin{align}\label{eq:sp1}
\left\langle \chi_L^{\otimes N} \Psi, \sum_{j=1}^N[T_j, \chi_{4L/3}^{\otimes N}] \xi_\tau  \Psitp \right\rangle= -\left\langle \sum_{j=1}^N[T_j, \chi_{4L/3}^{\otimes N}]\chi_L^{\otimes N} \Psi,  \xi_\tau  \Psitp \right\rangle =    O_{C^\infty}(e^{-dL}),
\end{align}
by arguing as in the proof of \eqref{eq:TcomChiPsi} and using~\eqref{norm:Rtau}.
Then, we see that, by~\eqref{def:xi},
\[\langle \chi_{L}^{\otimes N} 	\Psitp,\tilde{I}_\tau \chi_{4L/3}^{\otimes N} \xi_\tau \Psitp\rangle=\langle\tilde{I}_\tau\chi_L^{\otimes N} \Psitp,\chi_{4L/3}^{\otimes N}\Pi^\bot(H_\infty^\bot -E_1-E_2)^{-1} \Pi^\bot \tilde{I}_{\tau}\chi_L^{\otimes N} \Psitp\rangle.\]
As a consequence, using \eqref{eq:chichi} and applying~\eqref{interleadterm},
we arrive at
\begin{equation}\label{chilpsiI}\langle \chi_{L}^{\otimes N} 	\Psitp,\tilde{I}_\tau \chi_{4L/3}^{\otimes N} \xi_\tau \Psitp\rangle=\frac{\langle f_{(U,V)}\Psitp,\Pi^\bot(H_\infty^\bot -E_1-E_2)^{-1} \Pi^\bot f_{(U,V)} \Psitp\rangle}{L^6}+O_{C^\infty}\left(\frac{1}{L^7}\right),\end{equation}
 where we also used that $(H_\infty^\bot -E_1-E_2)^{-1}$ is bounded and independent of $\tau$.
The definition of the van der Waals constant $C_{\mathrm{vdW}}$  given in~\eqref{def:VdWterm}, together with \eqref{chilpsiI} implies
\begin{equation}\label{res:int2}\langle \chi_{L}^{\otimes N} 	\Psitp,\tilde{I}_\tau \chi_{4L/3}^{\otimes N} \xi_\tau \Psitp\rangle=\frac{C_{\mathrm{vdW}}(\Psitp,U,V)}{L^6}+O_{C^\infty}\left(\frac{1}{L^7}\right).\end{equation}
With \eqref{est:HIEW}, \eqref{eq:sp1} and~\eqref{res:int2}, we find that
\begin{multline}\nonumber
\langle \chi_{L}^{\otimes N} 	\Psitp,(H_\infty + \tilde{I}_\tau -E_1-E_2) W_\tau  \Psitp\rangle=-\frac{C_{\mathrm{vdW}}(\Psitp,U,V)}{L^6}  +\langle\chi_{L}^{\otimes N} 	\Psitp,  \chi_{4L/3}^{\otimes N} \Pi \chi_{L}^{\otimes N} \tilde{I}_\tau 	\Psitp\rangle+O_{C^\infty}\left(\frac{1}{L^7}\right),
\end{multline}
which together with~\eqref{PSxitot} and~\eqref{def:Wtau} gives
\begin{equation}\label{PSWtot2}
\langle W_{\tau} 	\Psitp,(H_\infty + \tilde{I}_\tau -E_1-E_2) W_\tau  \Psitp\rangle=-\frac{C_{\mathrm{vdW}}(\Psitp,U,V)}{L^6}  +\langle\chi_{L}^{\otimes N} 	\Psitp,  \chi_{4L/3}^{\otimes N} \Pi \chi_{L}^{\otimes N} \tilde{I}_\tau 	\Psitp\rangle+O_{C^\infty}\left(\frac{1}{L^7}\right).
\end{equation}

Using~\eqref{commut-project}, \eqref{eq:chichi}  and  the fact that  $\Pi\Psi=\Psi$, we find
\begin{equation}\label{PichiPsi}
 \langle\chi_{L}^{\otimes N} 	\Psitp,  \chi_{4L/3}^{\otimes N} \Pi \chi_{L}^{\otimes N} \tilde{I}_\tau 	\Psitp\rangle=\langle\chi_{L}^{\otimes N} 	\Psitp,    \tilde{I}_\tau \chi_{L}^{\otimes N}	\Psitp\rangle+O_{C^\infty}(e^{-dL}).
\end{equation}
\eqref{PSWtot2} and \eqref{PichiPsi} imply
\begin{align}\label{PSWtot}
\langle W_{\tau} 	\Psitp,(H_\infty + \tilde{I}_\tau -E_1-E_2) W_\tau  \Psitp\rangle=\langle\chi_{L}^{\otimes N} 	\Psitp,    \tilde{I}_\tau \chi_{L}^{\otimes N}	\Psitp\rangle-\frac{C_{\mathrm{vdW}}(\Psitp,U,V)}{L^6}+O_{C^\infty}\left(\frac{1}{L^7}\right).
\end{align}

So far, we have worked with a general normalized state $\Psi \in \cG_1 \otimes \cG_2$. Using ~Lemma~\ref{thm:multipol2} with the molecular charge densities defined by~\eqref{def:MCD}, we obtain in the special case $\Psi=\Psi_1 \otimes \Psi_2$
\begin{align}\nonumber
\langle 	\chi_L^{\otimes N} \Psi_1 \otimes \Psi_2 ,    \tilde{I}_{\tau}   	\chi_L^{\otimes N} \Psi_1 \otimes \Psi_2 \rangle 
=&\sum_{2 \leq m+n \leq 5}\frac{\mathcal{F}^{(n,m)}(U\rho_{\chi_L^{\otimes N_1}\Psi_1}, V\rho_{\chi_L^{\otimes N_2}\Psi_2})}{L^{m+n+1}} + O_{C^\infty} \left( \frac{1}{L^7}\right)\\
=&\sum_{2 \leq m+n \leq 5}  \frac{\mathcal{F}^{(n,m)}(U\rho_{\Psi_1}, V\rho_{\Psi_2})}{L^{m+n+1}} + O_{C^\infty} \left( \frac{1}{L^7}\right),\label{res:int1}
\end{align}
where, in the last step we could remove the cut-off functions due to the exponential decay of the ground states. Indeed,  \eqref{expdecayGS} implies, in view of~\eqref{multipol-coeff}, that $\mathcal{F}^{(n,m)}(U\rho_{\chi_L^{\otimes N_1}\Psi_1}, V\rho_{\chi_L^{\otimes N_2}\Psi_2})$ and $ \mathcal{F}^{(n,m)}(U\rho_{\Psi_1}, V\rho_{\Psi_2})$ are exponentially close to each other.

Inserting~\eqref{res:int1} in~\eqref{PSWtot} gives, for $\Psitp=\Psi_1 \otimes \Psi_2$,
\begin{multline}
  \langle  W_\tau \Psi_1 \otimes \Psi_2,  (H_\infty+\tilde{I}_\tau -E_1-E_2)  W_\tau \Psi_1 \otimes \Psi_2\rangle=
\sum_{1 \leq m+n \leq 5}\frac{\mathcal{F}^{(n,m)}(U\rho_{\Psi_1}, V\rho_{\Psi_2})}{L^{m+n+1}}  \\ \label{eq:est-num}  -\frac{C_{\mathrm{vdW}}(\Psi_1 \otimes \Psi_2,U,V)}{L^6}+ O_{C^\infty} \left( \frac{1}{L^7}\right).
\end{multline}
 
 Using~\eqref{ineq:Etau}, \eqref{gs:trialfunction}, and \eqref{eq:est-num}, we find, choosing $\Psitp=\Psi_1 \otimes \Psi_2$, that
 \begin{multline}
  \cE_\tau-E_1-E_2\leq\frac{1}{\| W_\tau \Psi_1 \otimes \Psi_2 \|^2}\Bigg(
  \sum_{1 \leq m+n \leq 5}\frac{\mathcal{F}^{(n,m)}(U\rho_{\Psi_1}, V\rho_{\Psi_2})}{L^{m+n+1}} \\
  -\frac{C_{\mathrm{vdW}}(\Psi_1\otimes\Psi_2,U,V)}{L^6}+ O_{C^\infty} \left( \frac{1}{L^7}\right)\Bigg)\label{preupper}. \end{multline}

Note that the norm of the trial function gives a negligible contribution because, from \eqref{eq:normfunctionorder},
\begin{equation}\label{oneovernormvarphi}
\frac{1}{\| W_\tau \Psitp\|^2} =\frac{1}{ 1  +  O_{C^\infty}\left( \frac{1}{L^6}\right) } = 1 +  O_{C^\infty}\left( \frac{1}{L^6}\right).
\end{equation}

From \eqref{preupper} and \eqref{oneovernormvarphi} for $\Psitp=\Psi_1 \otimes \Psi_2$, we immediately arrive at \eqref{eq:upper}.
\qed

\medskip

The next step is to prove that, for all $\Psi\in\cG_1\otimes\cG_2$, $U$ and $V$, the van der Waals constant $C_{\mathrm{vdW}}(\Psi,U,V)$, defined in~\eqref{def:VdWterm}, is positive. The definition clearly implies that $C_{\mathrm{vdW}}(\Psi,U,V)\geq0$ for all states $\Psi$ and for all $U$ and $V$. Moreover, it is zero if and only if $\Pi^\perp f_{(U,V)}\Psi=0$. This would mean that $f_{(U,V)}\Psi $ is a ground state of $\hat{H}_1+ \hat{H}_2$. We will prove that it is not the case, with the help of the following lemma. It is similar to \cite[Lemma~2.4]{al}, which holds in the \textbf{(NR)} case. Nevertheless, since we cannot use the Leibniz rule in the \textbf{(SR)} case, we need a modification of the proof.
\begin{lemma}\label{lem:notGS}
	Let $\Psi\in L^2(\R^{3N_k})\backslash\{0\}$ ($k=1,2$) be such that $ H_k\Psi=E\Psi$ for some  $E$, where the operators $H_k$ have been defined in~\eqref{def:Hi}. For a vector $(b_1,...,b_{N_k})\in\R^{3N_k}\backslash\{0\}$, we define the operator $B$ by
\begin{equation*}
(B\Psi)(x) := \sum_{j=1}^{N_k}b_j\cdot x_j \Psi(x).
\end{equation*}	
Then, $H_kB\Psi\neq EB\Psi$. 
\end{lemma}
\begin{proof}
 
Assume  that $H_kB\Psi= EB\Psi$. Then,  $\Psi$ satisfies the equations
	\begin{equation}\label{HPsi}H_k\Psi=\sum_{j=1}^{N_k}T_j\Psi+I_k\Psi=E\Psi
	\end{equation}
	and
	\begin{equation}\label{HBPsi}H_kB\Psi=\sum_{j=1}^{N_k}T_jB\Psi+I_k B\Psi=EB\Psi,
	\end{equation}
	where $I_k:=H_k-\sum_{j=1}^{N_k}T_j$ contains all the interaction terms.
	Applying $B$ to~\eqref{HPsi} and subtracting it from~\eqref{HBPsi}, we find that 
	\begin{equation}\label{commB}
	\sum_{j=1}^{N_k}[ T_j,B]\Psi=0.
	\end{equation}	
	Let us denote, for each $j=1,\dots,N_k$, $b_j=:(b_j^1,b_j^2,b_j^3)$ and $x_j=:(x_j^1,x_j^2,x_j^3)$.
	We have then that $[ T_j,B]=\sum_{k=1}^3b_j^k[T_j,x_j^k]$. Let us compute the commutator $[T_j,x_j^k]$.
	Using the definition of $T_j$ through the Fourier transformation in~\eqref{def:T}, we find that that
	\begin{equation*}[T_j,x_j^k]\psi=-\iu(\sqrt{-\Delta_j+1})^{-1}{\partial_{x_j^k}},
\end{equation*}
so~\eqref{commB} is equivalent to
\begin{equation*}
 \sum_{j=1}^{N_k}-\iu(\sqrt{-\Delta_j+1})^{-1}{b_j\cdot\nabla_j}\Psi=0.
\end{equation*}
Applying the Fourier transformation \eqref{def:Fourier}, we find
\begin{equation*}
\sum_{j=1}^{N_k}\frac{b_j\cdot p_j}{\sqrt{1+|p_j|^2}} \mathcal{F}\Psi(p)=0.
\end{equation*}
Assuming without loss of generality that $b_1 \neq 0$, we find that  $\partial_{p_1} \sum_{j=1}^{N_k}\frac{b_j\cdot p_j}{\sqrt{1+|p_j|^2}} \neq 0$. As a consequence, by the implicit function theorem, $\sum_{j=1}^{N_k}\frac{b_j\cdot p_j}{\sqrt{1+|p_j|^2}}$ is, when restricted on compact sets, zero only in a finite union of graphs. Therefore, it is nonzero almost everywhere. Consequently, $\mathcal{F}\Psi(p)=0$ almost everywhere, and thus $\Psi=0$, which contradicts the assumption $\Psi \neq 0$.
\end{proof}
We are now ready to prove the positivity of the van der Waals constant, following the proof of \cite[Proposition~2.3]{al}. For the sake of completeness, we give it here. Due to the definition of $f_{(U,V)}$ in~\eqref{expansion:fUV}, we have that 
\[f_{(U,V)}\Psitp=\sum_{n,m\geq0}a_{nm}\Phi_{1,n}\otimes\Phi_{2,m},\]
where  $(a_{nm})$ are some real coefficients, $\Phi_{j,0} \in \cG_j$,  and, for all $n>0$, $\Phi_{j,n} \in x_i^{k}\cG_j$, for some $i\in \{1,\dots,N\}$ and $k \in \{1,2,3\}$. Assuming that $f_{(U,V)}\Psitp$ is a ground state of $H_\infty$, we find that
\[\sum_{n,m\geq1}a_{nm}\Phi_{1,n}^\perp\otimes\Phi_{2,m}^\perp=0,\]
with $\Phi_{k,n}^\perp=\Pi_k^\perp\Phi_{k,n}$ where
$\Pi_k$ is the orthogonal projection on $\cG_k$. As a consequence, the functions $(\Phi_{1,n}^\perp\otimes\Phi_{2,m}^\perp)_{n,m\geq1}$ are linearly dependent and we have that, for dimensional reasons, at least one of the sets $(\Phi_{1,n}^\perp)_{n\geq1}$ or $(\Phi_{2,m}^\perp)_{m\geq1}$ is linearly dependent. If we assume, for example, that there exists a finite family of complex numbers $(b_n)$ such that
\[\sum_{n\geq1}b_n\Phi_{1,n}^\perp=0,\] then $\sum_{n\geq1}b_n\Phi_{1,n}$ is an eigenvector of $H_1$ associated with the eigenvalue $E_1$, which contradicts Lemma \ref{lem:notGS}.
\qed
\bigskip

\subsection{Lower bound: general case}\label{subsec:lower}
In this section, we prove the following proposition, for which we do not assume Condition \ref{poles}. We will use it as an intermediate result  to prove Theorems ~\ref{thm:vanderwaals} (b) and ~\ref{thm:vanderwaalsbis} in the next subsection.
\begin{proposition}\label{prop:nopoles}
Assume Condition \ref{neutr}. Then 	
\begin{equation}\label{eq:Etaunopoles}
\cE_\tau=E_\infty+\min_{\substack{\Psi\in\cG_1\otimes\cG_2\\\|\Psi\|=1}}\left(\langle 	\chi_L^{\otimes N}\Psi ,    \tilde{I}_{\tau}   	\chi_L^{\otimes N}\Psi\rangle-\frac{C_{\mathrm{vdW}}(\Psi,U,V)}{L^6}\right)+ O_{L^\infty} \left( \frac{1}{L^7}\right).
\end{equation}
\end{proposition}
Recall that, if \ref{neutr} holds, then $E_\infty=E_1+E_2$. We will use this equality  without mentioning it, when we refer to equations of the previous subsection.  
  Note that, in the \textbf{(NR)} case, the  proposition easily follows  from the arguments of \cite{al}, even though it was not stated in this form.
  As a consequence, we consider here only the  \textbf{(SR)} case. 

In order to prove  a lower bound on the ground state energy, we want to make use of the Feshbach map method, which we recall here: let $H$ be a self-adjoint operator on a Hilbert space $\mathscr{H}$ with domain $\Dom(H)$, $P$ an orthogonal projection on $\mathscr{H}$ such that $\mathrm{Ran}~P \subset \Dom(H)$ and $H^{\bot} := P^{\bot} H P^{\bot}$, where $P^{\bot} := 1- P$. The \emph{Feshbach map} is an operator on $\Ran P$ defined in the following way: for any $E \in \mathbb{R}$ such that $(H^{\bot}-E)$ is invertible,
\begin{equation*}
F_P(E) := (P H P - P H P^{\bot} (H^{\bot} - E )^{-1} P^{\bot} H P)|_{\mathrm{Ran}P}.
\end{equation*}

The following well-known theorem, see \emph{e.g.} \cite[Section IV]{BFS}, shows how the Feshbach map is a useful tool to estimate the ground state energy of the system that we consider.

 \begin{theorem}\label{teo:feschbach}
	Suppose $H, \mathscr{D}(H), \mathscr{H}, P$ to be as above and that, for some $E \in \mathbb{R}$, there exists $C >0$ such that
\begin{equation*}
H^{\bot} - E \geq C >0.
\end{equation*}
Then 
\begin{equation*}
E \;\text{is an eigenvalue of } H \quad \iff \quad E \; \text{is an eigenvalue of } F_P(E).
\end{equation*}
Moreover, if, for some $\phi \in \Ran P \setminus \{0\}$, we have $F_P(E)\phi= E \phi$, then
\begin{equation*}
\phi-(H^{\bot} - E)^{-1}P^\bot H \phi,
\end{equation*}
is an eigenfunction of $H$ to the eigenvalue $E$.
\end{theorem}

Furthermore,  it turns out that, if $E$ is the ground state energy of $H$, then $E$ is the ground state energy of $F_P(E)$ and therefore there exists $\psi_0 \in \mathrm{Ran}P$ with $\|\psi_0\| =1$ such that
\begin{equation}\label{FeshGS}
E = \langle \psi_0 ,F_P(E)\psi_0\rangle = \min_{\substack{\psi \in \mathrm{Ran}P \\ \|\psi\| =1}}\langle \psi, F_P(E)\psi\rangle,
\end{equation}
see \emph{e.g.} \cite[Lemma 5.6]{anbhundert}.

In order to apply the Feshbach method to our case, we choose $\mathscr{H}= Q_N \cH$, $H=\hat{H}_\tau$, recall \eqref{def:QNH}, \eqref{Hhat},  and  $P = P_{\tau}$  the orthogonal projection with 
\begin{equation}\label{eq:RanP}
\Ran P_{\tau}:=
\Big\{ Q_N K_\tau W_\tau \Psi \,\;\Big|\, \; \Psi \in \mathcal{G}_1 \otimes \mathcal{G}_2\Big\},
\end{equation}
where  $ Q_N, K_\tau, W_\tau, \mathcal{G}_j$ were respectively defined in \eqref{def:Q}, \eqref{def:Ktau}, \eqref{def:Wtau} and \eqref{def:GSEigenspace}.

The next lemma is the analogous of \cite[Proposition 4.1]{AnaSig-17} in the \textbf{(SR)} case and shows that we are in the assumptions of Theorem \ref{teo:feschbach} thanks to Condition \ref{neutr}. 
\begin{lemma}\label{lemma:ims}
Let $P_{\tau}$ be the orthogonal projection defined by~\eqref{eq:RanP} and assume Condition \ref{neutr}. Then there exists $C>0$ such that, for $L$ large enough,
\begin{equation}\label{eq:gapconditionapplied}
P^{\bot}_{\tau}  H_{\tau}P^{\bot}_{\tau}-\cE_{\tau}   \geq C .
\end{equation}
\end{lemma}

\begin{proof}
If we follow the proof of~\cite[Proposition~4.1]{AnaSig-17}, we get that there exists $C,C'>0$ such that for all $L$ large enough
\begin{equation*}
\Pi^{\bot}_{\tau}  H_{\tau}\Pi^{\bot}_{\tau}-E_\infty   \geq \left(C-\frac{C'}{L}\right),
\end{equation*}
on $ \Q_N \cH$, where recall that $E_\infty$  is given by \eqref{eq:Eiinf}.
Here $\Pi_\tau$ is the orthogonal projection on
\begin{equation}\label{def:RanPitau}
\Ran \Pi_\tau:=\left\{ Q_N K_\tau \chi_L^{\otimes N} \Psi, \; \Psi \in \mathcal{G}_1 \otimes \mathcal{G}_2\right\}.
\end{equation}
 This together with \eqref{eq:Eiinf} and \eqref{eq:upper} implies that there exists $C''>0$ so that if $L$ is large enough then
\begin{equation}\label{oldgapcondition}
\Pi^{\bot}_{\tau}  H_{\tau}\Pi^{\bot}_{\tau}-\cE_\tau   \geq \left(C-\frac{C''}{L}\right),
\end{equation}
on  $ \Q_N \cH$.
  Even if the result of \cite{AnaSig-17} is about non-relativistic Hamiltonians, the proof can easily be adapted  to the \textbf{(SR)} case. Indeed, the only part of the proof where the fact that the kinetic energy operator is a Laplacian matters is the application of the IMS localization formula, which is replaced in our case by equation~\eqref{eq:imsloc}  in  \ref{appA}. \footnote{In fact, \cite{bhv} provided in Section 3 finer estimates that work for a broader class of localization functions in the \textbf{(SR)} case  but they are not needed in our setting.}

To extract \eqref{eq:gapconditionapplied} from \eqref{oldgapcondition}, we need to bound
\begin{equation}\label{decomp-norm-H}H_\tau(\Pi_\tau^\bot-P^{\bot}_{\tau})=H_\tau(P_{\tau}-\Pi_\tau)=-(H_\tau-E_\infty)\Pi_\tau +(H_\tau-E_\infty)P_{\tau} - E_\infty(\Pi_\tau-P_{\tau}).\end{equation}

We begin by bounding $(H_\tau-E_\infty)\Pi_\tau$. 
Due to \eqref{def:RanPitau}, we have
\begin{equation}\|(H_\tau-E_\infty)\Pi_\tau\|=\max_{\Psi \in \cG_1\otimes \cG_2, \|\Psi\|=1}\frac{\|(H_\tau-E_\infty)Q_NK_\tau\chi_L^{\otimes N}\Psi\|}{\|Q_NK_\tau\chi_L^{\otimes N}\Psi\|}\label{cov:HPitau}.
\end{equation}
Let us consider some $\Psi \in \cG_1\otimes \cG_2$ such that
$\|\Psi\|=1$. We can prove, as we have done in~\eqref{norm-QNphi}, that
\begin{equation}\label{eq:QNKPsi}
\|Q_NK_\tau\chi_L^{\otimes N}\Psi\|^2=\frac{1}{\binom{N}{N_1}}\|K_\tau\chi_L^{\otimes N}\Psi\|^2=\frac{1}{\binom{N}{N_1}}\|\chi_L^{\otimes N}\Psi\|^2=\frac{1}{\binom{N}{N_1}}+O_{C^\infty}(e^{-dL}),
\end{equation}
 where in the second step we used that $K_\tau$ is unitary and in the last step we used \eqref{norm:GS}.

We next bound $(H_\tau-E_\infty)Q_NK_\tau\chi_L^{\otimes N}\Psi$.
From \eqref{HinftytauHinfty},~\eqref{def:Itilde}, \eqref{def:Htau} and the fact that $Q_N$ commutes with $H_\tau$, we have 
\begin{equation}\label{eq:commQN}
(H_\tau-E_\infty)Q_NK_\tau= Q_N K_\tau (H_\infty + \tilde{I}_\tau-E_\infty).
\end{equation}
As $Q_N$ is an orthogonal projection and $K_\tau$ is unitary, \eqref{cov:HPitau} and \eqref{eq:commQN} imply 
\begin{equation}\label{ineq:QN}\|(H_\tau-E_\infty)\Pi_\tau\|\leq\max_{\Psi \in \cG_1\otimes \cG_2, \|\Psi\|=1}\frac{1}{\|Q_NK_\tau\chi_L^{\otimes N}\Psi\|}\|(H_\infty + \tilde{I}_\tau-E_\infty)\chi_L^{\otimes N}\Psi\|.
\end{equation}
But we know from \eqref{eq:Eiinf}, \eqref{eq:HminusEW}, \eqref{eq:TcomChiPsi} and~\eqref{interleadterm} that
\begin{equation}
 (H_\infty-E_\infty+\tilde{I}_\tau)\chi_L^{\otimes N}\Psi=O_{(C^\infty,L^2(\R^{3N}))}\left(\frac{1}{L^3}\right).\label{boundterm1bis}
\end{equation}
By~\eqref{eq:QNKPsi}, \eqref{ineq:QN} and~\eqref{boundterm1bis}, we find
\begin{equation}\label{bound:HPi}
 \|(H_\tau-E_\infty)\Pi_\tau\|=O_{L^\infty}\left(\frac{1}{L^3}\right).
\end{equation}
Note that, because of the max in~\eqref{ineq:QN}, we do not have $C^\infty$ estimate on the right-hand side of \eqref{bound:HPi}.

We now bound $(H_\tau-E_\infty)P_\tau$. 
Using~\eqref{eq:RanP} and arguing as in the proof of~\eqref{ineq:QN}, we find that
\begin{equation}\label{HminusEPnorm}
\|(H_\tau-E_\infty)P_\tau\|\leq \max_{\Psi \in \cG_1\otimes \cG_2, \|\Psi\|=1} \frac{\|(H_\infty + \tilde{I}_\tau-E_\infty) W_\tau \Psi\|}{\|Q_N K_\tau W_\tau \Psi\|}.
\end{equation}
 Using~\eqref{norm-QNphi}, \eqref{eq:normfunctionorder} and that $K_\tau$ is unitary, we arrive at
\begin{equation}\label{eq:QnKtauWtau}
\|Q_N K_\tau W_\tau \Psi\|^2=\frac{1}{ {N\choose N_1}}+O_ {C^\infty}\left(\frac{1}{L^6}\right),
\end{equation}
  uniformly on $\Psi \in \cG_1\otimes \cG_2$ with $\|\Psi\|=1$.
 
 We have, from \eqref{est:HIEW}, \eqref{eq:OL4-1} and~\eqref{eq:OL4-2}, that
 \begin{equation}
  (H_\infty + \tilde{I}_\tau - E_{\infty}) W_\tau  \Psitp= \chi_{4L/3}^{\otimes N} \Pi \chi_{L}^{\otimes N} \tilde{I}_\tau 	\Psitp +O_{(C^\infty, L^2(\R^{3N}))}\left(\frac{1}{L^4}\right)\label{boundterm2bis}
  \end{equation} and therefore, by~\eqref{interleadterm},
  \begin{equation}
  (H_\infty + \tilde{I}_\tau - E_{\infty}) W_\tau  \Psitp=O_{(C^\infty, L^2(\R^{3N}))}\left(\frac{1}{L^3}\right).\label{boundterm2}
 \end{equation}
Inserting~\eqref{eq:QnKtauWtau} and~\eqref{boundterm2} in~\eqref{HminusEPnorm}, we obtain
\begin{equation}\label{bound:HP}
 \|(H_\tau-E_\infty)P_\tau\|=O_{L^\infty}\left(\frac{1}{L^3}\right).
\end{equation}

Finally, let us bound $\|P_\tau-\Pi_\tau\|$. To this end, we remind  the following well-known linear algebra lemma:
 \begin{lemma}\label{lem:decomp-proj}
  Let $V$ be finite-dimensional subspace of a Hilbert space with inner product $\langle \cdot|\cdot \rangle$ and $(v_i)_{i=1,...,n}$ a basis of $V$. Let us denote by $G$ the Gram matrix, defined by $G_{ij}=\langle v_i|v_j \rangle$.
  Then $P_V$, the orthogonal projection on $V$, satisfies\begin{equation*}
  P_V=\sum_{i=1}^n  (G^{-1})_{ij}|v_i\rangle\langle v_j|                                              .\end{equation*}
\end{lemma}
Let now $\Psi_k, k=1,\dots,\dim\cG_1\otimes\cG_2$, be an orthonormal basis of $\cG_1 \otimes \cG_2$.
Arguing similarly as in the proof of \eqref{eq:QnKtauWtau}, we can find 
 \begin{equation}
 \label{eq:spQnKtauWtau}
 \langle Q_N K_\tau W_\tau \Psi_k,  Q_N K_\tau W_\tau \Psi_l \rangle =\frac{\delta_{kl} }{ {N\choose N_1}}+O_{C^\infty}\left(\frac{1}{L^6}\right).
 \end{equation}
 A similar argument gives that
 \begin{equation}
 \label{eq:spQnKtau}
 \langle Q_N K_\tau \chi_L^{\otimes N} \Psi_k,  Q_N K_\tau \chi_L^{\otimes N} \Psi_l \rangle =\frac{\delta_{kl} }{ {N\choose N_1}}+O_{C^\infty}\left(e^{-dL}\right).
 \end{equation}
 Note that the different order of the error terms in \eqref{eq:spQnKtauWtau} and \eqref{eq:spQnKtau} originates from the difference of the order in the error terms of the right-hand sides of
 \eqref{norm:GS} and \eqref{eq:normfunctionorder}.
As a consequence, for $L$ large enough, the functions $(Q_NK_\tau \chi_L^{\otimes N}\Psi_k)$ (resp. $(Q_NK_\tau W_\tau\Psi_k)$),  $k=1,\dots,\dim\cG_1\otimes\cG_2$, are linearly independent. Thus,   due to \eqref{def:RanPitau}, \eqref{eq:RanP}, they form a basis of $\Ran \Pi_\tau$ (resp. $\Ran P_\tau$).

Therefore, we can apply Lemma \ref{lem:decomp-proj}  to find  
\begin{equation}\label{dec:Pitau}
\Pi_\tau= \sum_{k,l=1}^{\dim\cG_1\otimes\cG_2} | Q_N K_\tau \chi_L^{\otimes N}\Psi_k \rangle g_\tau^{kl} \langle  Q_N K_\tau \chi_L^{\otimes N} \Psi_l|,
\end{equation} and
\begin{equation}\label{decomp:PL}
P_\tau= \sum_{k,l=1}^{\dim\cG_1\otimes\cG_2} | Q_N K_\tau W_\tau \Psi_k \rangle G_\tau^{kl} \langle  Q_N K_\tau  W_\tau \Psi_l|,
\end{equation}
where $g_\tau^{kl}, G_\tau^{kl}$ are the elements of the inverses of the respective Gram matrices.
  From \eqref{eq:spQnKtauWtau} and \eqref{eq:spQnKtau} we find 
  \begin{equation}\label{eq:Gdiff}
 g_\tau^{kl}={N\choose N_1}\delta_{kl} + O_{C^\infty}\left(e^{-cL}\right), \qquad G_\tau^{kl}={N\choose N_1}\delta_{kl} + O_{C^\infty}\left(\frac{1}{L^6}\right).
  \end{equation}

By \eqref{dec:Pitau}, \eqref{decomp:PL} and \eqref{def:Wtau} we find
\begin{align*}P_\tau-\Pi_\tau=&
-\sum_{k,l=1}^{\dim\cG_1\otimes\cG_2} |Q_N K_\tau \chi_{4L/3}^{\otimes N} \xi_\tau \Psi_k \rangle G_\tau^{kl}  \langle  Q_N K_\tau  W_\tau \Psi_l| \\& + \sum_{k,l=1}^{\dim\cG_1\otimes\cG_2} |Q_N K_\tau \chi_{L}^{\otimes N}  \Psi_k \rangle (G_\tau^{kl} -g_\tau^{kl}) \langle  Q_N K_\tau  W_\tau \Psi_l| 
\\& -  \sum_{k,l=1}^{\dim\cG_1\otimes\cG_2}|Q_N K_\tau \chi_{L}^{\otimes N}  \Psi_k \rangle g_\tau^{kl}  \langle  Q_N K_\tau \chi_{4L/3}^{\otimes N} \xi_\tau \Psi_l|,
\end{align*}
which together with \eqref{norm:Rtau} and \eqref{eq:Gdiff} implies
\begin{equation}\label{diffPPi}
 \|P_\tau-\Pi_\tau\|=O_{L^\infty}\left(\frac{1}{L^3}\right).
\end{equation}
Inserting~\eqref{bound:HPi},  \eqref{bound:HP}, and \eqref{diffPPi} into~\eqref{decomp-norm-H}, we find that 
\[\|H_\tau(\Pi_\tau^\perp-P_\tau^\perp)\|=O_{L^\infty}\left(\frac{1}{L^3}\right),\]
which, together with~\eqref{oldgapcondition}, proves that
\eqref{eq:gapconditionapplied} is satisfied for $L$ large enough. This concludes the proof of Lemma \ref{lemma:ims}.
\end{proof}

Due to Lemma \ref{lemma:ims},
we can use Theorem~\ref{teo:feschbach}, \eqref{eq:RanP} and \eqref{FeshGS} to obtain that  
\begin{equation}\label{gse:feshbach}
\cE_\tau =\min_{\substack{\Psi \in \cG_1\otimes\cG_2\\\|\Psi\|=1}}\frac{\langle Q_NK_\tau W_\tau\Psi, F_{P_\tau}(\cE_\tau) Q_NK_\tau W_\tau\Psi\rangle}{\|Q_NK_\tau W_\tau\Psi\|^2},
\end{equation}
where the Feshbach map $F_{P_\tau}$ is given in our setting by
\begin{equation}\label{def:feshbach1}
F_{P_\tau}(E) := P_{\tau} H_{\tau} P_{\tau} - P_{\tau} H_{\tau} P^{\bot}_{\tau} (H_{\tau}^{\bot} - E)^{-1} P^{\bot}_{\tau} H_{\tau} P_{\tau}.
\end{equation}
Here $H_\tau^\bot=P_\tau^\bot H_\tau P_\tau^\bot$.
The first term of the right-hand side of~\eqref{def:feshbach1} is  not a function of $E$ and the second one is a non-linear function of $E$.   In order to find a lower bound for $\cE_\tau$, we will bound from below the quadratic form in the right-hand side of  \eqref{gse:feshbach}. We start by estimating
\begin{equation}\label{term:L}
 \frac{\langle Q_NK_\tau W_\tau\Psi, P_{\tau} H_{\tau} P_{\tau}  Q_NK_\tau W_\tau\Psi\rangle }{\|Q_NK_\tau W_\tau\Psi\|^2}.
\end{equation}
 We use \eqref{gs:trialfunction} and \eqref{eq:RanP} which
 give that for all $\Psi\in\cG_1\otimes\cG_2$, 
 \begin{equation}\label{eq:linpart}
 \frac{\langle  Q_N K_\tau W_\tau \Psi, P_{\tau} H_{\tau} P_{\tau} Q_N K_\tau W_\tau \Psi\rangle}{\|Q_N K_\tau W_\tau \Psi\|^2 }= 
 \frac{\langle  W_\tau \Psi, (H_\infty+\tilde{I}_\tau)  W_\tau \Psi\rangle}{\| W_\tau \Psi\|^2}.
 \end{equation}
 Inserting \eqref{eq:Eiinf}, \eqref{PSWtot} and \eqref{oneovernormvarphi}
 in the right-hand side of \eqref{eq:linpart}, we find that
 \begin{align}
 \frac{\langle  Q_N K_\tau W_\tau \Psi, P_{\tau} H_{\tau} P_{\tau} Q_N K_\tau W_\tau \Psi\rangle}{\|Q_N K_\tau W_\tau \Psi\|^2 } =E_\infty+\langle 	\chi_L^{\otimes N}\Psi ,    \tilde{I}_{\tau}   	\chi_L^{\otimes N}\Psi\rangle \label{eq:linpart1}-\frac{C_{\mathrm{vdW}}(\Psi,U,V)}{L^6}+ O_{C^\infty} \left( \frac{1}{L^7}\right).
 \end{align}

 We continue by bounding from above
 \begin{equation}\label{term:NL}
 \frac{\langle Q_NK_\tau W_\tau\Psi, P_{\tau} H_{\tau} P^{\bot}_{\tau} (H_{\tau}^{\bot} - \cE_\tau)^{-1} P^{\bot}_{\tau} H_{\tau} P_{\tau}  Q_NK_\tau W_\tau\Psi\rangle }{\|Q_NK_\tau W_\tau\Psi\|^2},
 \end{equation}
 for $ \Psi \in \mathcal{G}_1 \otimes \mathcal{G}_2$.
We observe that, thanks to \eqref{eq:RanP}, \eqref{eq:gapconditionapplied} and the equality $P_{\tau}^{\bot} H_{\tau} P_{\tau}= P_{\tau}^{\bot} (H_{\tau} -E_\infty)P_{\tau}$, there exists $C>0$ such that, if $L$ is large enough,
\begin{multline}
\frac{\langle Q_NK_\tau  W_\tau \Psi , P_{\tau} H_{\tau} P^{\bot}_{\tau} (H_{\tau}^{\bot} - \cE_\tau)^{-1} P_{\tau}^{\bot} H_{\tau} P_{\tau}Q_N K_\tau W_\tau \Psi \rangle}{\|Q_NK_\tau  W_\tau \Psi\|^2}  \label{NLfesh1}
\leq C^{-1}\frac{\|P^\bot_{\tau} (H_{\tau} -E_\infty) Q_NK_\tau  W_\tau \Psi\|^2}{\|Q_N K_\tau W_\tau \Psi\|^2}.
\end{multline}
Using \eqref{eq:commQN} together with~\eqref{boundterm2bis}, we find
\begin{align}\nonumber
P^\bot_{\tau} (H_{\tau} -E_\infty) Q_N K_\tau  W_\tau \Psi&= P^\bot_{\tau} Q_NK_\tau (H_\infty +\tilde{I}_\tau -E_\infty)   W_\tau \Psi\\
\label{Feshcomm2}
&=P^\bot_{\tau} Q_N K_\tau \chi_{4L/3}^{\otimes N} \Pi \chi_{L}^{\otimes N} \tilde{I}_\tau 	\Psi  +O_{(L^\infty, L^2(\R^{3N}))}\left(\frac{1}{L^4}\right).
\end{align}
Let $\Psi=\Psi_1$ be a normalized element of $\cG_1\otimes\cG_2$.
 In order to bound $P^\bot_{\tau} Q_N K_\tau \chi_{4L/3}^{\otimes N} \Pi$, we will prove that
\begin{equation*}
 P_\tau^\bot Q_NK_\tau\chi_{4L/3}^{\otimes N} \Psi_1=O_{(L^\infty, L^2(\R^{3N}))}\left(\frac{1}{L^3}\right),
\end{equation*}
 uniformly on $\Psi_1 \in \cG_1\otimes \cG_2$ with $\|\Psi_1\|=1$.
 We begin by extending $\Psi_1$ to an orthonormal basis   $\Psi_1, \dots, \Psi_n$ of $\cG_1\otimes\cG_2$.
Using  \eqref{decomp:PL}, we see that
\begin{equation}\label{eq:PtauQKPsidec}
P_\tau Q_NK_\tau\Psi_1=\sum_{k,l=1}^n | Q_N K_\tau W_\tau \Psi_k \rangle G_\tau^{kl} \langle  Q_N K_\tau  W_\tau \Psi_l, Q_NK_\tau\chi_{4L/3}^{\otimes N} \Psi_1\rangle.
\end{equation}
Arguing as in the proof of \eqref{eq:QnKtauWtau}, we find 
\begin{equation}\label{eq:getagainridofQ}
\langle  Q_N K_\tau  W_\tau \Psi_l, Q_NK_\tau\chi_{4L/3}^{\otimes N} \Psi_1\rangle = \frac{\langle   W_\tau \Psi_l, \chi_{4L/3}^{\otimes N} \Psi_1\rangle}{{ N \choose N_1}}.
\end{equation}
But we have, due to~\eqref{expdecayGS},  \begin{equation}\label{eq:PSchi}
                           \langle\chi_L^{\otimes N}\Psi_l,\chi_{4L/3}^{\otimes N} \Psi_1\rangle=\delta_{l1}+O_{C^\infty}(e^{-dL}),
                        \end{equation}
and, due to~\eqref{Pixi}, \eqref{commut-project} and~\eqref{norm:Rtau},
\begin{equation}\label{eq:PSxi}
\langle \chi_{4L/3}^{\otimes N}\xi_{\tau}\Psi_l,\chi_{4L/3}^{\otimes N} \Psi_1\rangle=\langle [(\chi_{4L/3}^{\otimes N})^2, \Pi^\perp]\xi_{\tau}\Psi_l,\Psi_1\rangle=O_{C^\infty}(e^{-dL}),
\end{equation}
since $\Pi^\bot \Psi_1=0$. 
Using \eqref{eq:getagainridofQ}~\eqref{eq:PSchi}, \eqref{eq:PSxi} and \eqref{def:Wtau}, we find that
\begin{equation}\label{eq:delta1l}
\langle  Q_N K_\tau  W_\tau \Psi_l, Q_NK_\tau\chi_{4L/3}^{\otimes N} \Psi_1\rangle=\frac{\delta_{1l}}{{ N \choose N_1}}+O_{C^\infty}(e^{-dL}).
\end{equation}
With the help of \eqref{eq:Gdiff}, \eqref{eq:PtauQKPsidec} and \eqref{eq:delta1l}, we see that
\begin{equation}\label{eq:PQKpsi}
P_\tau Q_NK_\tau\Psi_1=Q_N K_\tau W_\tau \Psi_1+O_{(L^\infty, L^2(\R^{3N}))}\left(\frac{1}{L^6}\right).
\end{equation}
We obtain from~\eqref{eq:PQKpsi} and~\eqref{def:Wtau} that
\begin{align*}
P_\tau^\perp Q_NK_\tau\Psi_1=Q_NK_\tau(1-\chi_L^{\otimes N})\Psi_1+Q_NK_\tau\chi_{4L/3}^{\otimes N}\xi_{\tau}\Psi_1+O_{(L^\infty, L^2(\R^{3N}))}\left(\frac{1}{L^6}\right)=O_{(L^\infty, L^2(\R^{3N}))}\left(\frac{1}{L^3}\right),
\end{align*}
where in the last equality we used~\eqref{norm:Rtau} and \eqref{norm:GS}.
As a consequence,
\begin{equation} \label{prodPPi}
 \|P^\bot_{\tau} Q_N K_\tau \chi_{4L/3}^{\otimes N} \Pi\|=O_{L^\infty}\left(\frac{1}{L^3}\right).
\end{equation}

Inserting \eqref{prodPPi} and  \eqref{interleadterm1} in \eqref{Feshcomm2}, we obtain 
\begin{equation}\label{norm2term}
\|P^\bot_{\tau} (H_{\tau} -E_\infty) Q_N K_\tau  W_\tau \Psi\|=O_{L^\infty}\left(\frac{1}{L^4}\right),
\end{equation}
uniformly for all normalized $\Psi \in \cG_1 \otimes \cG_2$.

Using \eqref{norm2term}, \eqref{eq:QnKtauWtau} and \eqref{NLfesh1}, we arrive at
\begin{equation}\label{norm:NLterm}
\frac{\langle Q_N  K_\tau W_\tau \Psi , P_{\tau} H_{\tau} P^{\bot}_{\tau} (H_{\tau}^{\bot} - \cE_\tau)^{-1} P_{\tau}^{\bot} H_{\tau} P_{\tau}Q_N  K_\tau W_\tau \Psi \rangle}{\|Q_N K_\tau W_\tau \Psi\|^2}  = O_{L^\infty}\left( \frac{1}{L^8}\right).
\end{equation}

Using~\eqref{eq:linpart1}, \eqref{norm:NLterm}, \eqref{def:feshbach1} and \eqref{gse:feshbach}, we obtain \eqref{eq:Etaunopoles}. 
 This concludes the proof of Proposition \ref{prop:nopoles}.
\qed

\bigskip
\subsection{Proof of Theorems \ref{thm:vanderwaals} (b) and  \ref{thm:vanderwaalsbis} }\label{sec:lowerbound}
In this section, we prove Theorems~\ref{thm:vanderwaals} (b) and~\ref{thm:vanderwaalsbis} using Proposition~\ref{prop:nopoles}. As we explained below the statement of Proposition~\ref{prop:nopoles}, the latter holds in both the \textbf{(SR)} and the \textbf{(NR)} cases. Thus, we consider here both cases simultaneously. Note that, in this step, the analysis of~\cite{al} has to be changed, due to the weaker assumption~\ref{poles} for the ground state eigenspaces.

In view of Proposition \ref{prop:nopoles}, our goal is  to estimate the quadratic form
\begin{equation}\label{def:Qtau}
Q_\tau(\Psi):=\langle \chi_L^{\otimes N}\Psi, \tilde{I}_{\tau} \chi_L^{\otimes N}\Psi \rangle
\end{equation}
 on the unit ball of $\cG_1\otimes\cG_2$, under one of the two versions of Condition~\ref{poles}.
Let us first  assume Condition~\ref{polesall}. Together with~\eqref{res:int1}, it directly implies that, for all normalized $\Psi_1\in\cG_1$ and $\Psi_2\in\cG_2$,
\begin{equation}\label{polesall:I}Q_\tau(\Psi_1\otimes\Psi_2 )=\sum_{1 \leq m+n \leq 5}\frac{\mathcal{F}^{(n,m)}(U, V)}{L^{m+n+1}} + O_{C^\infty} \left( \frac{1}{L^7}\right),\end{equation}
where the notation $\mathcal{F}^{(n,m)}(U,V)$ was introduced right after Equation \eqref{eq:CvdW}.
Let us now consider the general case, \textit{i.e.} where $\Psi$ is not a factorized state.
For such a $\Psi$, since $\mathcal{G}_j$ are finite dimensional, we know that there exist orthonormal vectors $\{\Phi_1^{(k)}\}_{k=1}^{\dim\cG_1} \subseteq \mathcal{G}_1$ and $\{\Phi_2^{(l)}\}_{l=1}^{\dim\cG_2} \subseteq \mathcal{G}_2$ such that 
\begin{equation}\label{sumakl}
\Psi = \sum_{k=1}^{\dim\cG_1}\sum_{l=1}^{\dim\cG_2} a_{k,l} \Phi_{1}^{(k)} \otimes \Phi_2^{(l)}, \qquad \{a_{k,l}\}_{(k,l)} \subseteq \mathbb{C}, \; \sum_{k,l} |a_{k,l}|^2 =1.
\end{equation}

As a consequence, we can write that
\begin{align}
Q_\tau(\Psi)=& \sum_{p,q}\sum_{l,k} \overline{a_{p,q}} a_{l,k}\left\langle \chi_L^{\otimes N}\Phi_{1}^{(p)} \otimes \Phi_2^{(q)}, \tilde{I}_{\tau} \chi_L^{\otimes N}\Phi_{1}^{(l)} \otimes \Phi_2^{(k)} \right\rangle  \nonumber\\
=& \sum_{p,q}|a_{p,q}|^2Q_\tau\left(\Phi_{1}^{(p)} \otimes \Phi_2^{(q)}\right)  \nonumber \\
&+ \sum_{(p,q)\neq(k,l)}  \overline{a_{p,q}} a_{l,k}\left\langle \chi_L^{\otimes N}\Phi_{1}^{(p)} \otimes \Phi_2^{(q)}, \tilde{I}_{\tau} \chi_L^{\otimes N}\Phi_{1}^{(k)} \otimes \Phi_2^{(l)} \right\rangle\label{eq:I-Fesh}.
\end{align}

For the diagonal part of~\eqref{eq:I-Fesh},  we find, with the help of~\eqref{polesall:I} and~\eqref{sumakl}, that  
\begin{equation}\label{diagpart}
\sum_{p,q}|a_{p,q}|^2\Q_\tau\left(\Phi_{1}^{(p)} \otimes \Phi_2^{(q)}\right)  
=  \sum_{1 \leq n+m \leq 5}\frac{\mathcal{F}^{(n,m)}(U,V)}{L^{m+n+1}}+ O_{C^\infty}\left( \frac{1}{L^7}\right).
\end{equation}

We now prove that, for all $(p,q)\neq(k,l)$, 
\begin{equation}\label{outdiagpart}
 \left\langle \chi_L^{\otimes N}\Phi_{1}^{(p)} \otimes \Phi_2^{(q)}, \tilde{I}_{\tau} \chi_L^{\otimes N}\Phi_{1}^{(k)} \otimes \Phi_2^{(l)} \right\rangle=O_{C^\infty}\left( \frac{1}{L^7}\right).
\end{equation}
By the complex polarization identity, we find that
\begin{align}\nonumber
 4\left\langle \chi_L^{\otimes N} \Phi_{1}^{(p)}  \otimes \Phi_2^{(q)},  \tilde{I}_{\tau} \chi_L^{\otimes N}\Phi_{1}^{(k)} \otimes \Phi_2^{(l)} \right\rangle &=\left\langle \chi_L^{\otimes N}(\Phi_{1}^{(p)}+\Phi_{1}^{(k)}) \otimes \Phi_2^{(q)}, \tilde{I}_{\tau} \chi_L^{\otimes N}(\Phi_{1}^{(p)}+\Phi_{1}^{(k)}) \otimes \Phi_2^{(l)} \right\rangle \nonumber
 \\ &-\left\langle \chi_L^{\otimes N}(\Phi_{1}^{(p)}-\Phi_{1}^{(k)}) \otimes \Phi_2^{(q)}, \tilde{I}_{\tau} \chi_L^{\otimes N}(\Phi_{1}^{(p)}-\Phi_{1}^{(k)}) \otimes \Phi_2^{(l)} \right\rangle\nonumber 
  \\&+\iu \left\langle \chi_L^{\otimes N}(\iu\Phi_{1}^{(p)}+\Phi_{1}^{(k)}) \otimes \Phi_2^{(q)}, \tilde{I}_{\tau} \chi_L^{\otimes N}(\iu\Phi_{1}^{(p)}+\Phi_{1}^{(k)}) \otimes \Phi_2^{(l)} \right\rangle\nonumber
 \\-\iu\Big\langle \chi_L^{\otimes N}&(-\iu\Phi_{1}^{(p)}+\Phi_{1}^{(k)}) \otimes \Phi_2^{(q)}, \tilde{I}_{\tau} \chi_L^{\otimes N}(-\iu\Phi_{1}^{(p)}+\Phi_{1}^{(k)}) \otimes \Phi_2^{(l)} 
 \Big\rangle .\label{polident}
\end{align}

Let us first consider the case where $q=l$, and thus $p\neq k$.
We can rewrite the previous equation in the following way:
\begin{align*}
 4\left\langle \chi_L^{\otimes N}\Phi_{1}^{(p)} \otimes \Phi_2^{(q)}, \tilde{I}_{\tau}\chi_L^{\otimes N}\Phi_{1}^{(k)} \otimes \Phi_2^{(q)} \right\rangle &=Q_\tau\left((\Phi_{1}^{(p)}+\Phi_{1}^{(k)}) \otimes \Phi_2^{(q)}\right)-Q_\tau\left((\Phi_{1}^{(p)}-\Phi_{1}^{(k)}) \otimes \Phi_2^{(q)}\right)\\&\quad +\iu Q_\tau\left((\iu\Phi_{1}^{(p)}+\Phi_{1}^{(k)}) \otimes \Phi_2^{(q)}\right)-\iu Q_\tau\left((-\iu\Phi_{1}^{(p)}+\Phi_{1}^{(k)}) \otimes \Phi_2^{(q)}\right).
\end{align*}

But, since the four vectors $\Phi_{1}^{(p)}+\Phi_{1}^{(k)}$, $\Phi_{1}^{(p)}-\Phi_{1}^{(k)}$, $\iu\Phi_{1}^{(p)}+\Phi_{1}^{(k)}$, and $-\iu\Phi_{1}^{(p)}+\Phi_{1}^{(k)}$ are in $\cG_1$ and have  the same norm $\sqrt{2}$, we have, by~\eqref{polesall:I}, that the four values $Q_\tau\left((\Phi_{1}^{(p)}+\Phi_{1}^{(k)}) \otimes \Phi_2^{(q)}\right)$, $Q_\tau\left((\Phi_{1}^{(p)}-\Phi_{1}^{(k)}) \otimes \Phi_2^{(q)}\right)$, $Q_\tau\left((\iu\Phi_{1}^{(p)}+\Phi_{1}^{(k)}) \otimes \Phi_2^{(q)}\right)$, and $Q_\tau\left((-\iu\Phi_{1}^{(p)}+\Phi_{1}^{(k)}) \otimes \Phi_2^{(q)}\right)$ are all equal up to an error of order $O_{C^\infty}(1/L^7)$. As a consequence,
\begin{equation}\label{qeql}
 \left\langle \chi_L^{\otimes N}\Phi_{1}^{(p)} \otimes \Phi_2^{(q)}, \tilde{I}_{\tau} \chi_L^{\otimes N}\Phi_{1}^{(k)} \otimes \Phi_2^{(q)} \right\rangle=O_{C^\infty}\left( \frac{1}{L^7}\right).
\end{equation}

Let us now consider the case $q\neq l$ and $p\neq k$.
In view of~\eqref{polident}, we will prove that \[\left\langle \chi_L^{\otimes N}(\Phi_{1}^{(p)}+\Phi_{1}^{(k)}) \otimes \Phi_2^{(q)}, \tilde{I}_{\tau} \chi_L^{\otimes N}(\Phi_{1}^{(p)}+\Phi_{1}^{(k)}) \otimes \Phi_2^{(l)} \right\rangle=O_{C^\infty}\left( \frac{1}{L^7}\right).\]
All the other terms of the right-hand side of~\eqref{polident} can be similarly bounded, since $-\Phi_{1}^{(k)}$ and $\pm\iu\Phi_{1}^{(k)}$ are normalized ground states of $\hat{H}_1$ too.
We apply again a polarization identity, which gives
\begin{align*}
 4\Big\langle \chi_L^{\otimes N}(\Phi_{1}^{(p)} &+\Phi_{1}^{(k)})  \otimes \Phi_2^{(q)}, \tilde{I}_{\tau} \chi_L^{\otimes N}(\Phi_{1}^{(p)}+\Phi_{1}^{(k)}) \otimes \Phi_2^{(l)} \Big\rangle \\ = & Q_\tau\left((\Phi_{1}^{(p)}+\Phi_{1}^{(k)}) \otimes (\Phi_2^{(q)}+\Phi_2^{(l)})\right)-Q_\tau\left((\Phi_{1}^{(p)}+\Phi_{1}^{(k)}) \otimes (\Phi_2^{(q)}-\Phi_2^{(l)})\right)\\
 &+\iu Q_\tau\left((\Phi_{1}^{(p)}+\Phi_{1}^{(k)}) \otimes (\iu\Phi_2^{(q)}+\Phi_2^{(l)})\right)-\iu Q_\tau\left((\Phi_{1}^{(p)}+\Phi_{1}^{(k)}) \otimes (-\iu\Phi_2^{(q)}+\Phi_2^{(l)})\right)
\end{align*}
With the same argument as for~\eqref{qeql}, we find that  the 4  quadratic forms in the right-hand side of the last equation are equal up to  $O_{C^\infty}(1/L^7)$ and therefore
\begin{equation*}
\left \langle \chi_L^{\otimes N}\Phi_{1}^{(p)} \otimes \Phi_2^{(q)}, \tilde{I}_{\tau} \chi_L^{\otimes N}\Phi_{1}^{(k)} \otimes \Phi_2^{(l)} \right\rangle=O_{C^\infty}\left( \frac{1}{L^7}\right).
\end{equation*}
The remaining case $p=k$, $q\neq l$ can be handled just in the same way as the case $p\neq k$, $q=l$. This concludes the proof of~\eqref{outdiagpart}.

From~\eqref{def:Qtau}, \eqref{eq:I-Fesh}, \eqref{diagpart}, and~\eqref{outdiagpart}, we find that
\begin{equation}\label{Ipart}\left\langle 	\chi_L^{\otimes N}\Psi ,    \tilde{I}_{\tau}   	\chi_L^{\otimes N}\Psi\right\rangle=\sum_{1 \leq m+n \leq 5}\frac{\mathcal{F}^{(n,m)}(U, V)}{L^{m+n+1}} + O_{C^\infty} \left( \frac{1}{L^7}\right).
\end{equation}
Inserting~\eqref{Ipart} 
in~\eqref{eq:Etaunopoles}, we find
\begin{equation*}
\cE_\tau = E_{\infty} + \sum_{1 \leq m+n \leq 5}\frac{\mathcal{F}^{(n,m)}(U,V)}{L^{m+n+1}}
 - \frac{C_{\text{vdW}}(U,V)}{L^6} +  O_{L^\infty} \left( \frac{1}{L^7}\right),
\end{equation*}
where we recall that $C_{\text{vdW}}( U,V)$ was defined in~\eqref{eq:CvdW}.
This concludes the proof of  point (b) of Theorem~\ref{thm:vanderwaals}.

We now prove Theorem \ref{thm:vanderwaalsbis}. In the case where $n_1+n_2<5$ and we assume only Condition~\ref{poleslead},  the proof of~\eqref{eq:thmbis} is almost the same as the one of~\eqref{eq:lowerupper}. The only change is that~\eqref{polesall:I}, \eqref{diagpart} and therefore~\eqref{Ipart} are no longer  valid. We  have the following 
 modifications:
 \eqref{polesall:I} becomes
\begin{equation*}Q_\tau\left(\Psi_1\otimes\Psi_2 \right)=\frac{\mathcal{F}^{(n_1,n_2)}(U, V)}{L^{n_1+n_2+1}} + O_{C^\infty} \left( \frac{1}{L^{n_1+n_2+2}}\right),\end{equation*} and thus \eqref{diagpart} becomes
\begin{equation*}
\sum_{p,q}|a_{p,q}|^2\Q_\tau\left(\Phi_{1}^{(p)} \otimes \Phi_2^{(q)}\right)  
=  \frac{\mathcal{F}^{(n_1,n_2)}(U,V)}{L^{n_1+n_2+1}}+ O_{C^\infty}\left( \frac{1}{L^{n_1+n_2+2}}\right).\end{equation*}
Similarly, \eqref{outdiagpart} becomes
\begin{equation*}
 \left\langle \chi_L^{\otimes N}\Phi_{1}^{(p)} \otimes \Phi_2^{(q)}, \tilde{I}_{\tau} \chi_L^{\otimes N}\Phi_{1}^{(k)} \otimes \Phi_2^{(l)} \right\rangle=O_{C^\infty}\left( \frac{1}{L^{n_1+n_2+2}}\right),\end{equation*}
for all $(p,q)\neq(k,l)$.
As a consequence, instead of \eqref{Ipart} we have
\begin{equation}\label{IpartPoleslead}\left\langle 	\chi_L^{\otimes N}\Psi ,    \tilde{I}_{\tau}   	\chi_L^{\otimes N}\Psi\right\rangle=\frac{\mathcal{F}^{(n_1,n_2)}(U, V)}{L^{n_1+n_2+1}} + O_{C^\infty} \left( \frac{1}{L^{n_1+n_2+2}}\right).
\end{equation}
Inserting \eqref{IpartPoleslead} 
in \eqref{eq:Etaunopoles},
we arrive at~\eqref{eq:thmbis}.
\qed

\bigskip

\section{Proof of Proposition \ref{prop:fullexp}: A smoothed version of the ground state energy and expansion of its derivatives}\label{sec:firstderivative}
In order to prove Proposition \ref{prop:fullexp}, we would like to use that, for each $L$ large enough, any local pseudo-minimum of the function $(U,V) \mapsto \cE_L(U,V)$, recall \eqref{def:Etau} and \eqref{ELUV}, should be a critical point with a positive semidefinite Hessian but, in principle, $\cE_L$ could be non-differentiable, due to the infimum in its definition.

We treat both the \textbf{(SR)} and \textbf{(NR)} cases. We could not adapt   the analysis of \cite{al} in Section 3 so that it works for both cases, as in the \textbf{(SR)} case major additional difficulties occur. The first one is, like before, that the kinetic energy is a nonlocal operator. Moreover, the  singularities of the Coulomb potentials in $H_\tau$ move when  $\tau$ changes and, therefore, the second derivative of the interaction terms with respect to rotations  gives a term that is not locally $L^1$. In \cite{al}, this problem was compensated using that the  ground states of $H_{1,\tau}$ and $H_{2,\tau}$ are in  $H^2$ in the  \textbf{(NR)} case.   
In the \textbf{(SR)} case, at most $H^1$ regularity can be expected. But the proof of~\cite{al} cannot be adapted to this weaker regularity. 
 Moreover, in this case, it is not even clear whether the ground states are in $H^1$, since, as we explained below Equation \eqref{rel-bound}, the domain of $H_\tau$ may contain functions that are not in $H^1$.   For this double reason, the strategy of \cite{al} is not applicable in our setting  and we have to proceed completely differently.  

To overcome these difficulties, we will use ideas of the proof of Theorem 1 in \cite{Hunziker-86}.  Hunziker proved  in the \textbf{(NR)} case that there is a family of unitary transformations $U(\Xi)$ such that
$U^{-1}(\Xi) H_N(Y + \Xi,\cZ) U(\Xi)$ depends smoothly on $\Xi$, when $\Xi$ is in a neighborhood of $0$.  The unitary transformations were chosen so that the motion of the singularities of the Coulomb potentials was reversed. 
Here, we will adapt this method so that it firstly provides quantitative information for the derivatives and secondly  becomes applicable in the \textbf{(SR)} case as well. We note that the semirelativistic Laplacian is known to be analytic with respect to dilatations, see \cite{Weder} and in particular Lemma 1 for a precise statement.

In fact, in the \textbf{(NR)} case, our strategy  provides a simpler proof than the one of Section~3 of~\cite{al} and makes it possible to provide information not only for the first two derivatives, but for higher derivatives as well, see Proposition \ref{reg:Etilde} below. Recall that, in our proof, we treat both  the \textbf{(SR)} and the \textbf{(NR)} cases at the same time and we specify when we consider only one of them.

We assume that Condition \ref{neutr} is satisfied. 
Let ${\tb}:=(L,\overline{U},\overline{V})$ be such that the function $\cE_L$, defined in \eqref{ELUV}, has a local pseudo-minimum in $(\overline{U}, \overline{V})$ -- recall Definition \ref{def:locpseu}. Recall that  $\hat{H}_{\tb}$ was defined in \eqref{Hhat}. By Theorem \ref{teo:feschbach} and~\eqref{eq:RanP}, there exists some $\Psitb\in\cG_1\otimes\cG_2$ with $\|\Psi_{\tb}\|=1$ such that a ground state of $\hat{H}_{\tb}$ is given by
\begin{equation}\label{Phitau0}
\Phi_{\tb}:=Q_N K_{\tb} W_{\tb}\Psitb - R_{\tb},
\end{equation}
where
\begin{equation}\label{Rtau0}
R_{\tb}:=(H_{\tb}^\bot-\cE_{\tb})^{-1}  P_{\tb}^\bot  (H_{\tb}-E_\infty) Q_N K_{\tb} W_{\tb}\Psitb.
\end{equation}
 In other words, 
\begin{equation*}
\cE_{\tb}=\inf\sigma(\hat{H}_{\tb})=\frac{\langle \Phi_{\tb}, H_{\tb} \Phi_{\tb} \rangle}{\|\Phi_{\tb}\|^2}.
\end{equation*}

Note that in \eqref{Rtau0} we could subtract $E_\infty$ since $P_{\tb}^\bot Q_N K_{\tb} W_{\tb} \Psi_{\tb}=0$ by \eqref{eq:RanP}.
A key observation here is that, if for some open neighborhood $\Omega\subset \SO(3)\times \SO(3)$, $(\Phi_\tau)_{\tau=(L,U,V),(U,V)\in\Omega}$ is a family of functions so that $\Phi_{\tb}$ is given by \eqref{Phitau0}, then $(\overline{U},\overline{V})$ is a local pseudo-minimum not only of $\cE_{\tau}$ but also of the new function 
\begin{equation}\label{eq:trick_F-S-upper-bound}
\tilde{\cE}_{\tau}:=\frac{\langle \Phi_\tau, H_\tau \Phi_\tau \rangle}{\|\Phi_\tau\|^2}.
\end{equation}
The observation simply follows from the fact that $\tilde{\cE}_{\tau} \geq \cE_{\tau}$, by the definition of $\cE_{\tau}$ in \eqref{def:Etau}, and the equality $\cE_{\tb}=\tilde{\cE}_{\tb}$.

We want to choose the family  $\Phi_\tau$ appropriately so that $\tilde{\cE}_{\tau}$ has the same leading order  as $\cE_\tau$ when $L$ is large and, at the same time, is smooth. This is an idea that already appeared in \cite{al}, but here we choose a  different family. Moreover, the rest of the proof completely changes. 

     We  consider the family of vector fields $\mt:\R^3 \to \R^3$ with
\begin{multline}
\mt(x):=U \overline{U}^{-1} x \chi_\frac{8L}{3}(x)  + \Big[Le_1 + V\overline{V}^{-1}(x-Le_1) \Big]\chi_\frac{8L}{3}(x-Le_1)   \label{def:mut}
+ x\Big(1-\chi_\frac{8L}{3}(x)-\chi_\frac{8L}{3}(x-Le_1)\Big),
\end{multline}
where  $\chi_\frac{8L}{3}$ was defined in \eqref{def:chiL}.
 Note that, by \eqref{def:chiL},
 \begin{equation}\label{chichichi}
 \chi_\frac{8L}{3} \chi_{2L}= \chi_{2L}, \qquad  \chi_{2L}\chi_{\frac{4L}3}=\chi_{\frac{4L}3}
 \end{equation}
 and
 \begin{equation}\label{chichisupp}
 \supp{\chi_{\frac{8L}{3}}} \cap  (\supp\chi_{\frac{8L}{3}}-Le_1)= \emptyset.
 \end{equation}
The idea of the definition \eqref{def:mut} is that, if $x$ is close to the $j$-th molecule ($j=1$ or 2), then it is rotated together with the nuclei around the center of the $j$-th molecule. Indeed, using \eqref{def:mut}-\eqref{chichisupp}, we find that
\begin{equation}\label{mtclose}
\mt(x)=U \overline{U}^{-1} x \text{ on } \supp \chi_{2L}, \text{ } \mt(x)=Le_1 + V\overline{V}^{-1}(x-Le_1) \text{ on }  (\supp \chi_{2L}-Le_1),
\end{equation}
compare with \eqref{def:Ytau}. 
If $x$ is far from both molecules, namely outside of the supports of the functions
$\chi_\frac{8L}{3}$ and $\chi_\frac{8L}{3}(.-Le_1)$, then  $\mt(x)=x$ by \eqref{def:mut}.  

A simple computation gives that
\begin{equation}\label{eq:mtvt}
\mt(x)=x + v_\tau(x),
\end{equation}
where
\begin{equation}\label{def:vt}
v_\tau(x):=(U-\overline{U})\overline{U}^{-1}x \chi_\frac{8L}{3}(x) + (V-\overline{V})\overline{V}^{-1}(x-Le_1) \chi_\frac{8L}{3}(x-Le_1).
\end{equation}
We have that, for all $(U,V) \in \SO(3)\times \SO(3)$ and $L>1$,
\begin{equation}\label{eq:vLip}
\|v_\tau'\|_{L^\infty(\R^3,\R^{3\times3})}\leq C  \sup_{y \in \R^3, \|y\|=1}(\|(U-\overline{U})y\|_{\R^3} + \|(V-\overline{V})y\|_{\R^3}),
\end{equation}
where 
\begin{equation*}
C:=\|(x \chi_{\frac{8L}{3}}(x))'\|_{L^\infty
(\R^3)}
\end{equation*}
is uniformly bounded with respect to $L$ due to \eqref{def:chiL}. 
 In particular, there exists some open neighborhood $\Omega$ of $(\overline{U},\overline{V})$ such that
 \begin{equation}\label{omegaproperty}
  (U,V)\in\Omega \implies v_\tau(\cdot) \text{ is a contraction uniformly in } L>1. 
  \end{equation}

Following ideas of \cite{Hunziker-86}, we  consider the family of  transformations $\ut:L^2(\R^{3N}) \to L^2(\R^{3N})$ defined for $(U,V)\in\Omega$, through
\begin{equation}\label{utdef}
(\ut \Psi)(x_1,\dots,x_N):= \left(\prod_{j=1}^N J_{\mt}(x_j)\right)\Psi(\mt(x_1),\dots,\mt(x_N)),
\end{equation} 
where $J_f(x)$ denotes the determinant of the Jacobian matrix of $f$ at the point $x$. 
 It follows from~\eqref{omegaproperty} that, if $(U,V)\in\Omega$, then $\mt(\cdot)$ is a bijection. Indeed, for any $y \in \R^3$, the vector field $x \mapsto y-v_\tau(x)$ is a contraction and therefore has a unique fixed point. As a consequence, $\mt$ has an inverse map. Since $\mt \in C^\infty$ and, by~\eqref{eq:mtvt} and~\eqref{eq:vLip}, $\mt'$ is invertible, it follows, by the implicit function theorem, that $\mt$ is a $C^\infty$ diffeomorphism, which implies that the transformations $\ut$ are unitary.

We now define
\begin{equation}\label{Phitautdef}
\Phi_\tau:=Q_N K_{\tau} W_{\tau}\Psitb - \uti R_{\tb},
\end{equation}
where recall that $R_{\tb}$ was defined in \eqref{Rtau0}.  When $\tau=\tb$, \eqref{Phitautdef} coincides with \eqref{Phitau0}. In \cite{Hunziker-86}, a more general family of unitary transformations was considered by choosing the family $\mt$ in a  more general way. Here, we choose $\mt$ more specifically such that the electrons close to a molecule rotate together with its nuclei.  This observation will enable us to obtain quantitative information for the derivatives with the help, \textit{e.g.}, of equation \eqref{eq:utKtauKtau0}  below.

We remark that, because of \eqref{norm2term} and the boundedness of the resolvent $(H_{\tb}^\bot-\cE_{\tb})^{-1}$ from $L^2(\R^{3N})$ to the form domain of $H_{\tb}$, which is $H^{1/2}(\R^{3N})$ in the \textbf{(SR)} case and $H^{1}(\R^{3N})$ in the \textbf{(NR)} case,
\begin{equation}\label{norm:Rtau0}
 R_{\tb}=O_{(C^\infty,H^{q})}\left(\frac{1}{L^4}\right), \quad
	q = \begin{cases} 
		1, &\text{in the \textbf{(NR)} case}\\
		\frac{1}{2}, &\text{in the \textbf{(SR)} case}.
	\end{cases}
\end{equation}
 From \eqref{norm:Rtau0}, \eqref{testh12}, which in the \textbf{(NR)} case holds with $H^1$ regularity and the smoothness of the functions in the definition of $\uti$, it follows that $\Phi_{\tau}$ is in the form domain of $H_{\tau}$. The $C^\infty$ estimate follows here for free, as $R_{\tb}$ does not depend on $\tau$.  Note that, in the \textbf{(SR)} case, the uniformity with respect to $\tb$ for large $L$ follows from \eqref{rel-bound}.

We choose to insert the family defined by \eqref{Phitautdef} in \eqref{eq:trick_F-S-upper-bound}. We find that
\begin{equation}\label{EtauEinf}
\tilde{\cE}_{\tau}-E_\ii=\frac{\langle \Phi_{\tau}, (H_{\tau}-E_\ii) \Phi_{\tau} \rangle}{\|\Phi_{\tau}\|^2}.
\end{equation}
We emphasize the fact that, on the right-hand side of \eqref{Phitautdef}, $\Psitb$ and $R_{\tb}$ remain fixed when $\tau$ varies. This will help us to prove that $\tilde{\cE}_{\tau}$ is a smooth function of $\tau$   and we will moreover provide quantitative information for its derivatives.
We prove the following result:
\begin{proposition}\label{reg:Etilde}
	Let us assume that Conditions~\ref{poleslead} and \ref{neutr} are satisfied. 
 We assume that $n_1+n_2<5$, where recall that $n_1$ and $n_2$ are the indices of the first non-vanishing multipole moments of the molecules.	
For $\tilde{\cE}_\tau$ defined by~\eqref{eq:trick_F-S-upper-bound} and~\eqref{Phitautdef}, we have
 \begin{equation}\label{eq:regEtilde}
  \tilde{\cE}_{\tau}=E_\ii+\frac{\mathcal{F}^{(n_1,n_2)}(U,V)}{L^{n_1+n_2+1}}  +  O_{C^\infty} \left( \frac{1}{L^{n_1+n_2+2}}\right),
 \end{equation}
in the sense of~\eqref{C2error}-\eqref{def:C2error}.
\end{proposition}

\begin{proof}[Proof of Proposition~\ref{reg:Etilde}]
 In view of~\eqref{EtauEinf}, we begin by estimating $\|\Phi_{\tau}\|^2$. 
Using \eqref{Phitautdef} together with  \eqref{norm:Rtau0} and the fact that  $\ut$ commutes with $Q_N$, we find 

\begin{equation}\label{Phitautnorm}
\|\Phi_{\tau}\|^2=\|Q_N K_{\tau} W_{\tau}\Psitb\|^2   - 2 \text{Re} \langle R_{\tb},  Q_N \ut K_{\tau} W_{\tau}\Psitb \rangle+ O_{C^\infty}\left(\frac{1}{L^8}\right).
\end{equation}
 
From  \eqref{def:Ktau}, \eqref{def:mut} and \eqref{utdef}, it follows that
\begin{multline}
(u_\tau K_{\tau} \Psi)(x_1,\dots,x_N) \label{utKtau}= \left(\prod_{j=1}^N J_{\mt}(x_j)\right)
 (K_{\tb} \Psi)(\elt(x_1),\dots, \elt(x_{N_1}),\nt(x_{N_1+1}),\dots,\nt(x_{N})),
\end{multline}
where
\begin{equation}\label{def:eltnt}
\elt(x):=\overline{U}U^{-1} \mt(x), \qquad \nt(x):=\overline{V}V^{-1}(\mt(x)-Le_1) + Le_1.
\end{equation}
Using \eqref{mtclose} and \eqref{def:eltnt}, it follows that
\begin{equation}\label{eltntsuppchi}
\elt(x)=x, \quad \forall x \in \supp \chi_{2L},
 \qquad \nt(x)=x, \quad\forall x \in (\supp \chi_{2L}-Le_1).
\end{equation}
Moreover, \eqref{mtclose} implies the equality
\begin{equation}\label{mtJacclose}
J_{\mt}(x)=1, \quad \forall x \in \supp \chi_{2L} \cup (\supp \chi_{2L}-Le_1).
\end{equation}
Combining \eqref{utKtau}, \eqref{eltntsuppchi} and \eqref{mtJacclose}, we arrive at
\begin{equation}\label{eq:utKtauKtau0}
\ut K_{\tau}\Phi=K_{\tb}\Phi,   \text{ when } \supp \Phi \subset\supp \chi_{2L}^{\otimes N}.
\end{equation}
In particular, since by \eqref{def:Wtau}   $W_{\tau}\Psitb$ is supported on 
 $\supp \chi_{\frac{4L}{3}}^{\otimes N} \subset \supp \chi_{2L}^{\otimes N}$, \eqref{eq:utKtauKtau0} implies that
\begin{equation}\label{utktcancelation}
\ut K_{\tau} W_{\tau}\Psitb= K_{\tb}  W_{\tau}\Psitb.
\end{equation}
   Using \eqref{norm:GS}, \eqref{eq:normfunctionorder},  \eqref{norm:Rtau0} and \eqref{utktcancelation},
   we find
\begin{equation}\label{Phitautnorm1}
 \langle R_{\tb},  Q_N \ut K_{\tau} W_{\tau} \Psitb \rangle=\langle R_{\tb},  Q_N  K_{\tb} \chi_L^{\otimes N} \Psitb \rangle + O_{C^\infty}\left(\frac{1}{L^7}\right).
\end{equation}
Since by \eqref{Rtau0} $R_{\tb}=P_{\tb}^\bot R_{\tb}$, using \eqref{diffPPi} and \eqref{norm:Rtau0}, we see that
\begin{equation}\label{Phitautnorm2}
\langle R_{\tb},  Q_N K_{\tb} \chi_L^{\otimes N} \Psitb \rangle = \langle R_{\tb}, \Pi_{\tb}^\bot Q_N K_{\tb}  \chi_L^{\otimes N} \Psitb \rangle + O_{C^\infty}\left(\frac{1}{L^7}\right)=O_{C^\infty}\left(\frac{1}{L^7}\right),
\end{equation}
where in the last step we used that  $\Pi_{\tb}^\bot Q_N K_{\tb}  \chi_L^{\otimes N} \Psitb=0$, which is immediate consequence of \eqref{def:RanPitau}.
The $C^\infty$ estimate follows for free, since the terms appearing in \eqref{Phitautnorm2} are $\tau$ independent.  
From \eqref{Phitautnorm}, \eqref{Phitautnorm1} and \eqref{Phitautnorm2}, we find that
\begin{equation}\label{Phitautnormsquare}
\|\Phi_{\tau}\|^2=\|Q_N K_{\tau} W_{\tau}\Psitb\|^2  + O_{C^\infty}\left(\frac{1}{L^7}\right).
\end{equation}

We now estimate the numerator of the right-hand side of \eqref{EtauEinf}. By~\eqref{Phitautdef}, we can decompose 
\begin{align}
 \label{num:term1}\langle \Phi_{\tau}, (H_{\tau}-E_\ii) \Phi_{\tau} \rangle&=\langle  Q_N K_{\tau} W_{\tau}\Psitb, (H_{\tau}-E_\ii)Q_N K_{\tau} W_{\tau}\Psitb\rangle \\\label{num:term2}&\quad-2\Re\langle \uti R_{\tb},(H_{\tau}-E_\ii)Q_N K_{\tau} W_{\tau}\Psitb\rangle\\&\quad + \langle R_{\tb},\ut(H_{\tau}-E_\ii)\uti R_{\tb}\rangle.\label{num:term3}
\end{align}
Equations~\eqref{eq:linpart1}, \eqref{IpartPoleslead} and~\eqref{Phitautnormsquare}  imply that, when $n_1+n_2<5$,
\begin{align}\nonumber
 \frac{\langle  Q_N K_{\tau} W_{\tau}\Psitb, (H_{\tau}-E_\ii)Q_N K_{\tau} W_{\tau}\Psitb\rangle}{\|\Phi_{\tau}\|^2}
 &
 =\frac{\langle  Q_N K_{\tau} W_{\tau}\Psitb, (H_{\tau}-E_\ii)Q_N K_{\tau} W_{\tau}\Psitb\rangle}{\|Q_N K_{\tau} W_{\tau}\Psitb\|^2  + O_{C^\infty}\left(\frac{1}{L^7}\right)}\\\label{eq:sec2der} &=   \frac{\mathcal{F}^{(n_1,n_2)}(U,V)}{L^{n_1+n_2+1}} 
  +  O_{C^\infty} \left( \frac{1}{L^{n_1+n_2+2}}\right).
\end{align}
 Note that this is the only point where we use Condition~\ref{poles} in the proof of Proposition~\ref{reg:Etilde}.
 
We will now estimate the terms~\eqref{num:term2} and~\eqref{num:term3}.  
We begin with~\eqref{num:term2}. By~\eqref{eq:commQN} and the fact that $\ut$ is unitary and commutes with $Q_N$, we have
\begin{align}\nonumber
 \langle \uti R_{\tb},(H_{\tau}-E_\ii)Q_N K_{\tau} W_{\tau}\Psitb\rangle =&
\langle  R_{\tb},Q_N\ut  K_{\tau}(H_{\infty}+\widetilde{I}_{\tau}-E_\ii)  W_{\tau}\Psitb\rangle\\
=&\langle  R_{\tb},\ut  K_{\tau}\chi_{2L}^{\otimes N}(H_{\infty}+\widetilde{I}_{\tau}-E_\ii) W_{\tau}\Psitb\rangle\label{term:close}\\
&+\langle  R_{\tb},\ut  K_{\tau}  (1-\chi_{2L}^{\otimes N})(H_{\infty}+\widetilde{I}_{\tau}-E_\ii) W_{\tau}\Psitb\rangle,\label{term:faraway}
\end{align}
where in the last step we used the equality $Q_N R_{\tb}=R_{\tb}$ and the decomposition $1= \chi_{2L}^{\otimes N} + (1-\chi_{2L}^{\otimes N})$. 
On the one hand, using \eqref{eq:utKtauKtau0}, we find
\begin{align}
 \langle  R_{\tb},\ut K_{\tau}\chi_{2L}^{\otimes N}(H_{\infty}+\widetilde{I}_{\tau}-E_\ii) W_{\tau}\Psitb\rangle \label{eq:predercutoff}=\langle  R_{\tb}, K_{\tb} \chi_{2L}^{\otimes N}(H_{\infty}  +\widetilde{I}_{\tau}-E_\ii) W_{\tau}\Psitb\rangle.
\end{align}
The only part which depends on $\tau$ in the right-hand side of \eqref{eq:predercutoff} is $(H_{\infty}+\widetilde{I}_{\tau}-E_\ii) W_{\tau}\Psitb$. Therefore, inserting \eqref{boundterm2} and~\eqref{norm:Rtau0} in \eqref{eq:predercutoff}, we find
\begin{equation}\label{eq:dercutoff}
 \langle  R_{\tb}, K_{\tb} \chi_{2L}^{\otimes N}(H_{\infty}+\widetilde{I}_{\tau}-E_\ii) W_{\tau}\Psitb\rangle=O_{C^\infty}\left(\frac{1}{L^7}\right).
\end{equation}

Let us now bound~\eqref{term:faraway}. Here, we have to distinguish between the two types of kinetic energy. In the \textbf{(NR)} case, since $ W_{\tau}\Psitb$ and $(1-\chi_{2L}^{\otimes N})$ have disjoint supports and the Hamiltonian is a local operator, we immediately see that this term is 0.
For the~\textbf{(SR)} case, we need to work more. The only nonlocal terms in the Hamiltonian are the kinetic energy operators. Furthermore, the functions $W_{\tau}\Psitb$ and $(1-\chi_{2L}^{\otimes N})K_{\tau}^{-1}\ut^{-1}R_{\tb}$ are antisymmetric with respect to permutations of variables leaving invariant the  set $\{1,\dots,N_1\}$. Therefore, we can write
\begin{align}
 \langle  R_{\tb},\ut K_{\tau}(1-\chi_{2L}^{\otimes N})(H_{\infty}+\widetilde{I}_{\tau}-E_\ii) W_{\tau}  \Psitb\rangle&=N_1\langle  R_{\tb},\ut K_{\tau}(1-\chi_{2L}^{\otimes N})T_1  W_{\tau}\Psitb\rangle\label{der:1stmol}
\\&\quad+N_2\langle  R_{\tb},\ut K_{\tau}(1-\chi_{2L}^{\otimes N})T_N W_{\tau} \Psitb\rangle,\label{der:2ndmol}
\end{align}
 where recall that $T_i$ is the semirelativistic kinetic energy operator applied to the $i$-th electron.

 We begin by taking some fixed $\Phi\in L^2(\R^{3N})$. By \eqref{def:chiL}, the distance between the supports of $\chi_{4L/3}^{\otimes N}$ and $1-\chi_{2L}^{\otimes N}$ is bigger than $7L/12$. Therefore, we see with Lemma~\ref{expdecay:realT} that $(1-\chi_{2L}^{\otimes N}) T_1 \chi_{4L/3}^{\otimes N} \Phi$ is in $L^2(\R^{3N})$.
Moreover,~\eqref{int-kern} implies that, for almost all $(x_1,X) \in \R^3 \times \R^{3(N-1)}$  where $X:=(x_2,\dots,x_N)$, we have
 \begin{align}\label{eq:Ktwo}
 \Big((1-\chi_{2L}^{\otimes N}) T_1 \chi_{4L/3}^{\otimes N} \Phi\Big)(x_1,X)
  = -\frac{1-\chi_{2L}(x_1)}{2 \pi^2} \int_{\R^3} \frac{K_2(|x_1-y|)}{|x_1-y|^2} (\chi_{4L/3}^{\otimes N} \Phi)(y,X) \d y,
  \end{align}
    where $K_2$ is the modified Bessel function of the second kind. Note that we could replace $1-\chi_{2L}^{\otimes N}(x_1,X)$ with $1-\chi_{2L}(x_1)$ in the right-hand side of \eqref{eq:Ktwo} due to \eqref{chichichi}. Equation
\eqref{eq:Ktwo}  together with \eqref{utKtau}, \eqref{eltntsuppchi} and \eqref{mtJacclose} implies
 \begin{align}
\ut K_\tau \Big((1-\chi_{2L}^{\otimes N}) T_1 \chi_{4L/3}^{\otimes N} \Phi\Big)(x_1,X)
 \label{eq:yintegration}= -\frac{1-\chi_{2L}(x_1)}{2 \pi^2} J_{\mt}(x_1) \int_{\R^3} 
  \frac{K_2(|\elt(x_1)-y|)}{|\elt(x_1)-y|^2} [K_{\tb}(\chi_{4L/3}^{\otimes N} \Phi)](y,X) \d y.
 \end{align}
 
 Let us now consider the integral in the right-hand side of~\eqref{eq:yintegration}.
 Note that, by \eqref{def:chiL}, for $|x_1|\leq L/5$, we have $1-\chi_{2L}(x_1)=0$ and thus the right-hand side of \eqref{eq:yintegration} is zero. Therefore, we consider only the case $|x_1|\geq L/5$. This implies that
 \begin{equation}\label{elxfar}
 |\elt(x_1)| \geq \frac{L}{5},
 \end{equation}
 because $\elt$ is bijective and, due to \eqref{eltntsuppchi}, $\elt(B(0,L/5))=B(0,L/5)$.

  On the other hand, $[K_{\tb}(\chi_{4L/3}^{\otimes N} \Phi)](y,X)$ is not 0 only if $y$ is in the support of $\chi_{4L/3}$, \textit{i.e.}, when $|y|\leq L/6$.
 Therefore by \eqref{elxfar}, it follows that $|\elt(x_1)-y| \geq \frac{L}{30}$ in all cases where the right-hand side of \eqref{eq:yintegration} is not 0.

Controlling the derivatives with respect to rotations of the quantity in~\eqref{eq:yintegration} is not as obvious as in the other cases and we need to go back to the definition of~\eqref{C2error}. We choose arbitrary  $A,B \in \Gamma$, as defined in \eqref{def:so3}, and introduce
\begin{equation}\label{eq:taut}
\taut:=(L,e^{At} \overline{U},e^{Bt} \overline{V}), \qquad t \in \R.
\end{equation}

Recall that $\Omega$ is a neighborhood of $(\overline{U}, \overline{V})$ such that \eqref{omegaproperty} holds. There exists $\eta>0$ such that, for all $t\in(-\eta,\eta)$ and all $A,B \in \Gamma$, $(e^{At} \overline{U},e^{Bt} \overline{V})\in \Omega$, so we can apply what we have proven until now.
 Using~\cite[Formula~9.6.29]{AS} and~\cite[Formula~9.7.2]{AS}, we find that, for all $n \in \N_0:=\N \cup \{0\}$, there exists $C_n >0$ such that for all $z\geq 1$ 
 \begin{equation*}
 |K_2^{(n)}(z)| \leq C_n e^{-z/2}.
 \end{equation*}
Moreover, by \eqref{def:mut} and \eqref{def:eltnt}, for all $n\in\N_0$, there exist $D_n>0$, independent of $A$ and $B$, such that for all $t\in(-\eta,\eta)$,
\begin{equation*}
\left|\frac{\d^n}{\d t^n} l_{\tau_t}(x)\right| \leq D_n (1+|x|), \quad \forall x \in \R^3.
\end{equation*}

As a consequence, for all $x_1$ and $y$ such that $|x_1|\geq L/5$ and $|y|\leq L/6$, the function 
\begin{equation}\label{gdef}
g_{(x_1,y)}(t): = \frac{K_2(|l_{\tau_t}(x_1)-y|)}{|l_{\tau_t}(x_1)-y|^2},
\end{equation}
 is infinitely differentiable. Moreover, there exists a $d>0$ such that, for all $n \in \N_0$, there exists $M_n>0$, such that,  for all $L$ large enough, for all $|x_1|\geq L/5$ and $|y|\leq L/6$, all $A$, $B\in\Gamma$ and all $t\in(-\eta,\eta)$,
\begin{equation}\label{gder}
|g_{(x_1,y)}^{(n)}(t)|\leq M_n e^{-|l_{\tau_t}(x_1)-y|/2}\leq M_n e^{-d|x_1-y|},
\end{equation}
where, in the last step, we used the fact that, if $|x_1|\geq L/5$ and $|y|\leq L/6$, then $|l_{\tau_t}(x_1)-y|\geq \frac{1}{2}|x_1-y|$,  provided that $t$ is in a small neighborhood of zero, which is independent of $A$ and $B$. In fact, we have  $|l_{\tau_t}(x_1)-y|\geq \frac{1}{11}|x_1-y|$ uniformly for all times.
Moreover, from \eqref{def:mut}, it follows that
\begin{equation}\label{Jtaubound}
J_{\mt}=O_{(C^\infty,L^\infty)}(1).
\end{equation}
 Equation \eqref{gdef} together with \eqref{gder} enables the application of the dominated convergence theorem in \eqref{eq:yintegration} and together with \eqref{Jtaubound}
we can extract that  
\begin{equation}
\ut K_\tau
 (1-\chi_{2L}^{\otimes N}) T_1 \chi_{4L/3}^{\otimes N} = O_{(C^\infty,\cB(L^2))} (e^{-dL})\label{der:int-kern},
\end{equation}
because of Young's inequality and the fact that $\|e^{-d|.|} \chi_{|x| \geq \frac{L}{30}}\|_{L^1}$ is exponentially decaying in $L$. 

From~\eqref{norm:Rtau0}, \eqref{der:int-kern}, \eqref{eq:normfunctionorder} and the Leibniz rule, we find that, 
\begin{equation}
 \langle  R_{\tb},u_{\tau} K_{\tau}(1-\chi_{2L}^{\otimes N})T_1 W_{\tau}\Psitb\rangle = O_{C^\infty}(e^{-dL}).\label{der:1-chi}
\end{equation}
This bounds the derivatives of the term~\eqref{der:1stmol}. 
It similarly follows that
\begin{equation}
\langle  R_{\tb},\ut K_{\tau}(1-\chi_{2L}^{\otimes N})T_N W_{\tau}\Psitb\rangle = O_{C^\infty}(e^{-dL}).\label{der:2-chi}
\end{equation}
Inserting \eqref{der:1-chi} and \eqref{der:2-chi} in \eqref{der:1stmol}-\eqref{der:2ndmol}, we find that
\begin{equation}\label{eq:rtutKt}
\langle  R_{\tb},\ut K_{\tau}(1-\chi_{2L}^{\otimes N})(H_{\infty}+\widetilde{I}_{\tau}-E_\ii) W_{\tau}  \Psitb\rangle=  O_{C^\infty}(e^{-dL}).
\end{equation}
 Inserting \eqref{eq:dercutoff} and \eqref{eq:rtutKt}  into~\eqref{term:close}--\eqref{term:faraway}, we find that
\begin{equation}\label{der:crossterm}
 \langle \uti R_{\tb},(H_{\tau}-E_\ii)Q_N K_{\tau} W_{\tau}\Psitb\rangle=O_{C^\infty}\left(\frac{1}{L^7}\right),
\end{equation}
for both \textbf{(NR)} and the \textbf{(SR)} cases, because in the  \textbf{(NR)} case the left-hand sides of \eqref{der:1-chi} and \eqref{der:2-chi} are zero, as we discussed above.

We now bound~\eqref{num:term3}. Here, we treat the \textbf{(NR)} and \textbf{(SR)} cases simultaneously.
We  directly see with~\eqref{def:HN} that
\begin{equation}\label{decomp:Hunz}
 \ut(H_{\tau}-E_\ii)\uti=\sum_{i=1}^N\ut T_i\uti+\alpha\ut I_N(Y(\tau),\cZ)\uti-E_\infty.
\end{equation}

We begin with the potentials, which are the same in the \textbf{(NR)} and in the \textbf{(SR)} cases. We start with the electron-electron interactions. From \eqref{utdef} and \eqref{eq:mtvt}, we see that, for all $x_i$ and $x_j$, $i,j=1,...,N$, $i\neq j$ and for all $\tau$ close enough to $\tb$,
\begin{equation*}
 \ut\frac{1}{|x_i-x_j|}\uti=\frac{1}{|x_i-x_j+v_\tau(x_i)-v_\tau(x_j)|}=(1+\omega)^{-1/2}\frac{1}{|x_i-x_j|},
\end{equation*}
where
\begin{equation*}
 \omega:=\frac{1}{|x_i-x_j|^2}\left(2(x_i-x_j)\cdot(v_\tau(x_i)-v_\tau(x_j))+|v_\tau(x_i)-v_\tau(x_j)|^2\right).
\end{equation*}
Since $\omega$ is a polynomial function of values of the function $v_\tau$ and  $v_\tau$  is smooth in $\tau$, $\omega$  also depends smoothly on $U$ and $V$. Moreover, due to~\eqref{eq:vLip}, there exists an open neighborhood $\Omega'$ of $(\overline{U},\overline{V})$ such that, for all $x$, $y\in\R^3$ and  all $(U,V)\in\Omega'$,  $|v_\tau(x)-v_\tau(y)|< |x-y|/3$, which implies that $|\omega|<7/9$. Consequently, $(1+\omega)^{-1/2}$ is a smooth function of $U$ and $V$ on $\Omega'$.  Arguing as in the proof of \eqref{eq:vLip}, it follows that every derivative of $v_\tau$ with respect to $x$ is bounded uniformly in $L>1$. As a consequence, for all $n \in \N_0$, there exists $C_n>0$ such that, for all $t$ close enough to $0$ and all $A,B \in \Gamma$,
\begin{equation}\label{eq:elelderivatives}
\left| \frac{\d^n}{\d t^n} u_{\taut}\frac{1}{|x_i-x_j|}u_{\taut}^{-1} \right|\leq  \frac{C_n}{|x_i-x_j|}, \qquad \forall i,j \in \{1,\dots,N\}, \text{ with } i \neq j,
\end{equation}
where recall that $\taut$ was defined in \eqref{eq:taut}. 
To deal with the electron-nucleus attractions, we observe, with the help of \eqref{mtclose}, \eqref{def:Ytau} and recalling that in the current setting $U=e^{tA}\overline{U}, V=e^{tB}\overline{V}$, that
\begin{equation}\label{Ytautransformed}
 Y(\tau)=(m_\taut (y_1),\dots,m_\taut(y_M)).
 \end{equation}
 Moreover,
\begin{equation*}
u_{\taut}\frac{1}{|x_i-m_\taut(y_j)|}u_{\taut}^{-1}=\frac{1}{|m_\taut(x_i)-m_\taut(y_j)|},
\end{equation*}
so we can argue as in the proof of \eqref{eq:elelderivatives} to conclude that, for all $n \in \N_0$, there exists $D_n>0$ such that, for all $t$ close enough to $0$,
\begin{equation}\label{eq:elnuderivatives}
\left| \frac{\d^n}{\d t^n}u_{\taut}\frac{1}{|x_i-\mt(y_j)|}u_{\taut}^{-1}\right| \leq  \frac{D_n}{|x_i-y_{j,0}|}, \qquad \forall i  \in \{1,\dots,N\}, \forall j  \in \{1,\dots,M\},
\end{equation}
where $y_{j,0}$ is the position of the  $j$-th nucleus in the configuration $Y(\tb)$. 
We now consider the interactions between the nuclei. Since we can write that
\[\left\langle R_{\tb},u_{\taut}\frac{1}{|m_\taut(y_k)-m_\taut(y_l)|}u_{\taut}^{-1} R_{\tb}\right\rangle=\frac{1}{|m_\taut(y_k)-m_\taut(y_l)|}\langle R_{\tb},R_{\tb}\rangle,
\]
we obtain that 
\begin{equation}\label{eq:nunu1L8}
 \left\langle R_{\tb},u_{\tau}\frac{1}{|m_\tau(y_k)-m_\tau(y_l)|}u_{\tau}^{-1} R_{\tb}\right\rangle =O_{C^\infty}\left(\frac{1}{L^8}\right).
\end{equation}
 Using \eqref{def:IN}, \eqref{eq:elelderivatives}, \eqref{Ytautransformed}, \eqref{eq:elnuderivatives} and \eqref{eq:nunu1L8},  we can derive that
\begin{equation}\label{eq:derivative2ndterm}
\langle R_{\tb}, u_{\tau} I_N(Y (\tau),\cZ)u_{\tau}^{-1} R_{\tb} \rangle  \leq O_{C^\infty}( \|R_{\tb}\|_{H^\frac{1}{2}}^2),
\end{equation}
where in the last step we also used Kato's inequality $\frac{1}{|x|}\leq \frac{\pi}{2}\sqrt{1-\Delta}$, see \cite[Theorem 2.1]{Herbst}. Using \eqref{eq:derivative2ndterm} and \eqref{norm:Rtau0}, we find that
\begin{equation}\label{eq:derivative2ndtermagain}
  \langle R_{\tb}, \ut I_N(Y (\tau),\cZ)\uti R_{\tb} \rangle  = O_{C^\infty}\left(\frac{1}{L^8}\right),
\end{equation}
for both the \textbf{(NR)} and \textbf{(SR)} cases. 

Finally, we bound the kinetic energy terms of the right-hand side of~\eqref{decomp:Hunz}. Let us take $t\in\R$ and $A,B \in \Gamma$. For $t$ small enough so that $(e^{tA},e^{tB})\in\Omega$, where $\Omega$ was defined in \eqref{omegaproperty}, we define the following operator on $L^2(\R^3)$:
\begin{equation*}
D(t):= 1 - u_\taut \Delta u_\taut^{-1}=u_\taut(1 -  \Delta )u_\taut^{-1},
\end{equation*}
where recall that $\taut$ was defined in \eqref{eq:taut}. 
Since, for all $t$, $u_\taut$ is unitary and $-\Delta\geq0$, we have that $D(t)\geq 1$ in the sense of quadratic forms.
We find, following~\cite[Theorem~1]{Hunziker-86}, that there exist functions $(a^{ij}_t)$, $(b_t^k)$, $c_t$, which are all $C_c^\infty$ with compact support uniformly in $t$ and depend smoothly on $t$, such that
\begin{equation}\label{eq:Bt}
 D(t)=-\Delta+\sum_{i,j}a^{ij}_t(x)\frac{\partial}{\partial x_i}\frac{\partial}{\partial x_j}+\sum_{k}b^{k}_t(x)\frac{\partial}{\partial x_k}+c_t(x).
\end{equation}
Moreover, when $t$ goes to 0, all the functions $a_t^{ij}, b_t^k, c_t$ converge to 0 in $C^n$ norm for all $n \in \N_0$, uniformly in  $A,B \in \Gamma$. 
This together with \eqref{norm:Rtau0} implies that, in the~\textbf{(NR)} case, each of the terms $\langle R_{\tb},\ut T_i \uti R_{\tb}\rangle$ depends smoothly on $t$. Moreover, for all $n\geq 1$,
\begin{align*}
 \frac{\d^n}{\d t^n}\langle R_{\tb},u_\taut T_i u_\taut^{-1} R_{\tb}\rangle=
 \sum_{i,j}\left\langle R_{\tb},\frac{\d^na^{ij}_t}{\d t^n}\frac{\partial}{\partial x_i}\frac{\partial}{\partial x_j}R_{\tb}\right\rangle
 +\sum_{k}\left\langle R_{\tb},\frac{\d ^nb^{k}_t}{\d t^n}\frac{\partial}{\partial x_k}R_{\tb}\right\rangle+\left\langle R_{\tb},\frac{\d^nc_t}{\d t^n}R_{\tb}\right\rangle.
\end{align*}
As a consequence, in the \textbf{(NR)} case, we immediately obtain that
\begin{equation}\label{der:LaplaceNR}
 \langle R_{\tb},\ut T_i \uti R_{\tb}\rangle =O_{C^ \infty}\left(\frac{1}{L^8}\right),
\end{equation}
due to~\eqref{norm:Rtau0}.

For the \textbf{(SR)} case, we use the well-known identity $\sqrt{u}=1/\pi \int_0^\infty u/[(u+w)\sqrt{w}]\d w$ for all $u>0$.
It follows,  by the spectral theorem, that for all $t$
\begin{equation}\label{eq:rootrepr}
u_\taut T_i u_\taut^{-1} + 1 =u_\taut \sqrt{1 -\Delta} u_\taut^{-1}=\sqrt{D(t)}=\frac{1}{\pi} \int_0^\infty D(t)(D(t)+w)^{-1}\frac{\d w}{\sqrt{w}}.
\end{equation}
Therefore, with the help of the product rule, we find
\begin{multline}
\frac{\d}{\d t}\sqrt{D(t)}=\frac{1}{\pi} \int_0^\infty D'(t)(D(t)+w)^{-1}\frac{\d w}{\sqrt{w}}  \label{sqrtder}- \frac{1}{\pi} \int_0^\infty  D(t)(D(t)+w)^{-1}  D'(t)(D(t)+w)^{-1}\frac{\d w}{\sqrt{w}},
\end{multline}
in the sense of quadratic forms on $H^{\frac{1}{2}}(\R^{3})$. A rigorous justification can be achieved if we first consider quadratic forms on $H^2(\R^{3})$ and apply the dominated convergence theorem. We omit the details.  
Using the decomposition $D(t)=(D(t)+w)-w$,
 we find that
\begin{align}\nonumber
 \frac{1}{\pi} \int_0^\infty  D(t)(D(t)+w)^{-1}  D'(t)(D(t)+w)^{-1}\frac{\d w}{\sqrt{w}}&\\ \label{sqrtsecondterm}
 =  \frac{1}{\pi} \int_0^\infty D'(t)(D(t)+w)^{-1}\frac{\d w}{\sqrt{w}}  - & \frac{1}{\pi} \int_0^\infty   (D(t)+w)^{-1}  D'(t)(D(t)+w)^{-1} \sqrt{w} \,\d w.
\end{align}
Combining \eqref{sqrtder} and \eqref{sqrtsecondterm}, we obtain that 
\begin{equation}
\frac{\d}{\d t}\sqrt{D(t)}=  \label{rootderdec1} \frac{1}{\pi} \int_0^\infty   (D(t)+w)^{-1}  D'(t)(D(t)+w)^{-1} \sqrt{w} \,\d w.
\end{equation}
Using the commutator formula $[A^{-1},C]= A^{-1} [C,A] A^{-1}$, we find
\begin{align}\nonumber
 (D(t)+w)^{-1} D'(t) = D'(t) (D(t)+w)^{-1}  + (D(t)+w)^{-1}  [D'(t),D(t)](D(t)+w)^{-1},
 \end{align}
and, using the same commutator formula once more, we arrive at
\begin{multline}
(D(t)+w)^{-1} D'(t) = D'(t) (D(t)+w)^{-1}  + [D'(t),D(t)] (D(t)+w)^{-2} \\ \label{commutatorformula} + (D(t)+w)^{-1} [[D'(t),D(t)],D(t)] (D(t)+w)^{-2}.
\end{multline}
Inserting  the equalities $\int_0^\infty (D(t)+w)^{-2}\sqrt{w}\d w=\frac{\pi}{2}(D(t))^{-\frac{1}{2}}$ and  $\int_0^\infty (D(t)+w)^{-3}\sqrt{w}\d w=\frac{\pi}{8}(D(t))^{-\frac{3}{2}}$ together with \eqref{commutatorformula}  in  \eqref{rootderdec1}, we find that
 \begin{align}\nonumber
 \frac{\d}{\d t}\sqrt{D(t)}= F(t) \sqrt{D(t)} + C(t), 
 \end{align}
 where 
 $$F(t):=\frac{1}{2}  D'(t) D^{-1}(t) + \frac{1}{8}[D'(t),D(t)] D(t)^{-2},$$
 and  $$C(t):=\frac{1}{\pi} \int_0^\infty  (D(t)+w)^{-1} [[D'(t),D(t)],D(t)] (D(t)+w)^{-3}  \sqrt{w}\,\d w.$$ 
 But for all $s \in \R$ the operator $D'(t)$ is bounded from $H^{s+2}$ to $H^{s}$ due to~\eqref{eq:Bt}, as it is a second order differential operator with $C_c^\infty$ coefficients. 
  On the other hand, for $t$ in a neighborhood of 0 independent of the matrices $A$ and $B$, the operator $D(t)$ is elliptic, since it is a small perturbation of the Laplacian. Thus, by elliptic regularity, $D^{-1}(t)$ is bounded from $H^s$ to $H^{s+2}$ for all $s \in \R$. Consequently, the operator $F(t)$ is bounded from $H^{1/2}(\R^3)$ to itself for $t$ small enough.
  
 Concerning  $C(t)$, one can see that the second commutator $[[D'(t),D(t)],D(t)]$ is a 4-th order differential operator and it is bounded from $H^s(\R^3)$ to $H^{s-4}(\R^3)$, for all $s \in \R$. Consequently, the  operator $D(t)^{-1} [[D'(t),D(t)],D(t)] D(t)^{-1}$ is bounded from $H^{1/2}(\R^3)$ to itself. Moreover,  $(D(t)+w)^{-1}D(t) \leq 1$  for all  $w \geq 0$. There exists, therefore, a $\Theta>0$ such that, for all $\psi \in L^2(\R^{3})$  and all $t$ in a neighborhood of 0 which is independent of the matrices $A, B\in\Gamma$,
 \begin{equation*}
  \|C(t)\psi\|\leq \Theta\left\|\int_0^\infty(D(t)+w)^{-2}\sqrt{w}\d w \psi\right\|.
 \end{equation*}
Since $D(t)\geq 1$ for all $t$ and $\int_0^\infty(1+w)^{-2}\sqrt{w}\d w<\infty$, the operator $C(t)$ is bounded.

 Let $F_1(t):=F(t)$ and $C_1(t):=C(t)$. In order to bound the $n$-th derivative of $\sqrt{D(t)}$
 for all $n\in\N$, we define inductively   $F_{n+1}(t):=F_n'(t)+F_n(t)F(t)$ and $C_{n+1}(t):=F_n(t)C(t)+C_n'(t)$. Then, for each $n$, $F_n$ and $C_n$ are smooth families of operators. 
 We get, by induction, that:
 \begin{equation*}
  \frac{\d^n}{\d t^n}\sqrt{D(t)}= F_n(t) \sqrt{D(t)} + C_n(t).
 \end{equation*}
 Moreover, for each $t$ close enough to 0, $F_n(t)$ is bounded from $H^{1/2}(\R^3)$ to itself and $C_n(t)$ is bounded from $L^2(\R^3)$ to itself, as one can show by induction too, since $D^{(n)}(t)$ is bounded from $H^{s+2}(\R^3)$ to $H^s(\R^3)$ for all $s \in \R$ and $n \in \N$.
This implies that there exists a constant $\Lambda_n$ such that
\begin{equation}\label{der:LaplaceSR}
 \left|\frac{\d^n}{\d t^n}\langle R_{\tb},\sqrt{D(t)}R_{\tb}\rangle\right|\leq \Lambda_n\|R_{\tb}\|_{H^{1/2}}^2=O\left(\frac{1}{L^8}\right),
\end{equation}
where we used again~\eqref{norm:Rtau0}.
 Inserting~\eqref{eq:derivative2ndtermagain} and~\eqref{der:LaplaceNR} (in the \textbf{(NR)} case) or~\eqref{der:LaplaceSR} (in the \textbf{(SR)} case) in~\eqref{decomp:Hunz},
 we find that
 \begin{equation}\label{der:3rdterm}
\langle R_{\tb},\ut(H_{\tau}-E_\ii)\uti R_{\tb}\rangle=O_{C^\infty}\left(\frac{1}{L^8}\right).
 \end{equation}
 Finally, inserting~\eqref{eq:sec2der}, \eqref{der:crossterm} and~\eqref{der:3rdterm} in~\eqref{num:term1}--\eqref{num:term3},
 we find that
 \begin{equation}
\frac{\langle \Phi_{\tau}, (H_{\tau}-E_\ii) \Phi_{\tau} \rangle}{\|\Phi_{\tau}\|^2}= \frac{\mathcal{F}^{(n_1,n_2)}(U,V)}{L^{n_1+n_2+1}} 
+  O_{C^\infty} \left( \frac{1}{L^{n_1+n_2+2}}\right) + \frac{1}{\|\Phi_{\tau}\|^2} O_{C^\infty}\left(\frac{1}{L^7}\right),
 \end{equation}
 which together  with \eqref{eq:spQnKtauWtau}, \eqref{EtauEinf}, \eqref{Phitautnormsquare} and the hypothesis $n_1 + n_2 <5$
 leads to~\eqref{eq:regEtilde}. This concludes the proof of Proposition \ref{reg:Etilde}.
   \end{proof}

  To conclude the proof of Proposition~\ref{prop:fullexp},  we use the fact that $\tilde{\cE}_{\tau}$ has a local pseudo-minimum at $(\overline{U},\overline{V})$. Thus, if we take any matrices $A$ and $B$ as in~\eqref{cond:pseudomin}, and define $\tau_t:=(L,e^{tA}\overline{U},e^{tB}\overline{V})$, we find that
  \[\frac{\d\tilde{\cE}_{\taut}}{\d t}\Big|_{t=0}=0,\qquad \frac{\d^2\tilde{\cE}_{\taut}}{\d t^2}\Big|_{t=0}\geq0.\]
  This implies, together with Proposition~\ref{reg:Etilde}, that, for any $\delta>0$, we can find an $L_m>0$ such that, for all $L \geq L_m$,
  \begin{equation*}
  \left| \frac{\d\mathcal{F}^{(n_1,n_2)}(e^{tA}\overline{U},e^{tB}\overline{V})}{\d t}\Big|_{t=0}\right|\leq\delta,\qquad \frac{\d^2\mathcal{F}^{(n_1,n_2)}(e^{tA}\overline{U},e^{tB}\overline{V})}{\d t^2}\Big|_{t=0}\geq-\delta,
  \end{equation*}
  for all  $A,B\in\Gamma$ independently of the position of the local pseudominimum $(\overline{U},\overline{V})$. 
This concludes the proof of Proposition~\ref{prop:fullexp}.
\bigskip
\appendix
\section{Some properties of the semirelativistic kinetic energy operator}\label{appA}
In this appendix, we give some properties of the semirelativistic kinetic energy operator $T:=\sqrt{1-\Delta}-1$  on $L^2(\R^3)$, as defined in \eqref{def:T}.  Since it is a pseudodifferential and not a differential operator, it is  not as easy to deal with as the Laplace operator. One reason is that it is  not a local operator. The following lemma enables us to control this nonlocality. 

\begin{lemma}\label{expdecay:realT}
 Let $R>0$.	Then, for all  $f,g\in L^{2}(\R^3)$ such that $\dist(\supp f, \supp g)\geq R$,
 \begin{equation}\label{int-kern}
 \langle f, Tg\rangle=-\frac{1}{2\pi^2}\int_{\R^3}\int_{\R^3}\frac{K_2(|x-y|)}{|x-y|^2}\overline{f(x)}g(y)\,\d x\d y,
 \end{equation}
 where $K_2$ is the modified Bessel function of the second kind.
 Moreover, 
 for all $f\in L^{2}(\R^3)$, the linear form $Tf\in H^{-1}(\R^3)$ can be extended to a bounded linear form on \[\{g\in L^2(\R^3),\, \dist(\supp f, \supp g)\geq R\}.\] Furthermore, there exists $C>0$ such that, for $R$ large enough, the above linear form has norm less than $Ce^{-R}\|f\|_{L^2}$.
\end{lemma}
\begin{proof}
	
	We know from~\cite[Section 7.11]{LL} that, for all $f$, $g\in L^2(\R^3)$ and all $t>0$,
	\[\langle f, e^{-tT}g\rangle=\frac{e^t}{2\pi^2}\int_{\R^3}\int_{\R^3}\frac{t}{t^2+|x-y|^2}K_ 2\Big(\sqrt{t^2 + |x-y|^2}\Big)\overline{f(x)}g(y)\,\d x\d y,\]
where the factor $e^t$ originated from the 1 that was subtracted from the operator $\sqrt{1-\Delta}$ in the definition of $T$.
	
	If $f$ and $g$ have disjoint supports, we can write that
	\begin{align}\nonumber
		\frac{1}{t}\left(\langle f, g\rangle-\langle f, e^{-tT}g\rangle\right)&=-\frac{1}{t}\langle f, e^{-tT}g\rangle\\
		&=\frac{-e^t}{2\pi^2}\int_{\supp g}\int_{\supp f}\frac{K_ 2\Big(\sqrt{t^2 + |x-y|^2}\Big)}{t^2+|x-y|^2}\overline{f(x)}g(y)\d x\d y.\label{eq:Tlimit}
	\end{align}
Moreover, for all $t>0$, $x\in \supp f$ and $y\in \supp g$, 
	\begin{align}
	\left|\frac{K_ 2\Big(\sqrt{t^2 + |x-y|^2}\Big)}{t^2+|x-y|^2}\overline{f(x)}g(y)\right| \leq 	\frac{K_ 2(|x-y|)}{|x-y|^2}\left|\overline{f(x)}g(y)\right|,
	\end{align}
	because $K_2$ is a positive and decreasing function, see~\cite[Formula~9.6.29]{AS}.
		If  $\dist(\supp f, \supp g)\geq R$, then, by Young's inequality
		\begin{equation}\label{bound:kern}
		\int_{\supp g}\int_{\supp f} \frac{K_ 2(|x-y|)}{|x-y|^2}\left|\overline{f(x)}g(y)\right| \d x \d y \leq \|f\|_{L^2} \|g\|_{L^2} \int_{|x|\geq R}\frac{K_2(|x|)}{|x|^2} \d x.
		\end{equation}
		But the integrand on the right hand side of \eqref{bound:kern} is continuous on the domain of integration. 	
		 Moreover, as we can see for example in the result 9.7.2 of~\cite{AS}, $K_2(x)\sim\sqrt{\frac{\pi}{2x}}e^{-x}$ for large $x>0$. It follows that the right hand side of \eqref{bound:kern} is finite. 		 
We can thus use the dominated convergence theorem to take the limit when $t\to0$ in \eqref{eq:Tlimit} and find  that for all such $f$ and $g$,  $\langle f,Tg\rangle$ is well-defined and \eqref{int-kern} holds.	
	Moreover, if $R$ is large enough then  \eqref{int-kern},\eqref{bound:kern} and the asymptotic behavior of $K_2(x)$ for large $x$ lead to
	\begin{equation*}
		|\langle f,Tg\rangle|\leq C e^{-R}\|f\|_{L^2}\|g\|_{L^2},
	\end{equation*}	as desired.
\end{proof}
\bigskip

Another difficulty with the semirelativistic kinetic energy operator is that there is no Leibniz rule. As a consequence, commutators between this operator and a multiplication operator cannot be explicitly computed. Nevertheless, we are able to bound some commutators with the following lemma. 

\bigskip

\begin{lemma}\label{lem:expcommut1overL}
	Let $\zeta \in C_0^{\infty}(\mathbb{R}^3) $, with  supp $\zeta\subset B_{1/4}(0)$
and denote $\zeta_R(x): = \zeta\left( \frac{x}{R}\right)$. Then,
	there exists $C>0$ such that we have for all $R>1$, in operator norm,  
\begin{equation*}
		\|[T,\zeta_R]\|_{\cB(L^2)} \leq \frac{C}{R}.
	\end{equation*}
\end{lemma}
\begin{proof}
	We  first consider $\psi\in H^1(\R^3)$. In view of the definition of $T$ in~\eqref{def:T} and of \eqref{hcm1} we use the Fourier transformation as defined in \eqref{def:Fourier}. Let us denote, for all $p\in\R^3$, $\widetilde{T}(p):=\sqrt{|p|^2+1}$. We define
	\[\Xi_R:=\widetilde{T}\cF(\zeta_R\psi)-\frac{1}{(2\pi)^{3/2}}\cF({\zeta}_R) *(\widetilde{T}\cF(\psi)).\]
	Then $\Xi_R= \cF\big([T,\zeta_R ]\psi\big)$ so, by Plancherel identity,
	\begin{equation}\label{commPlanch}
	 \|[T,\zeta_R ]\psi\|_{L^2(\R^3)}=\|\Xi_R\|_{L^2(\R^3)}.
	 \end{equation}
	
	For almost every $p$, we have
	\begin{align*}
		(2\pi)^{3/2}\Xi_R(p)&=\widetilde{T}(p)(\cF(\zeta_R) *\cF(\psi))(p)-\cF(\zeta_R) *(\widetilde{T}\cF(\psi))(p)\\
		&=\widetilde{T}(p)\int_{\R^3} \cF(\zeta_R) (q)\ \cF(\psi)(p-q)\d q-\int_{\R^3}\cF(\zeta_R) (q)\ \widetilde{T}(p-q)\ \cF(\psi)(p-q)\d q\\
		&=\int_{\R^3}\cF(\zeta_R) (q) \ \cF(\psi)(p-q)\ (\widetilde{T}(p)-\widetilde{T}(p-q))\d q.
	\end{align*}
	Therefore, since the function $\widetilde{T}$ is Lipschitz with Lipschitz constant 1,\[(2\pi)^{3/2}|\Xi_R(p)|\leq\int_{\R^3}|\cF(\zeta_R) (q)||\cF(\psi)(p-q)||\widetilde{T}(p)-\widetilde{T}(p-q)|\d q\leq\int_{\R^3}|q||\cF(\zeta_R) (q)||\cF(\psi)(p-q)|\d q\]
	so we have \[(2\pi)^{3/2} \|\Xi_R\|_{L^2(\R^3)}\leq \||\cdot|\cF(\zeta_R) \|_{L^1(\R^3)}\|\cF(\psi)\|_{L^2(\R^3)}.\] Since $\|\cF(\psi)\|_{L^2(\R^3)}=\|\psi\|_{L^2(\R^3)}$, using \eqref{commPlanch} we find
	\[\|[T,\zeta_R ]\psi\|_{L^2(\R^3)}\leq\frac{1}{(2\pi)^{3/2}}\||\cdot|\cF(\zeta_R) \|_{L^1(\R^3)}\|\psi\|_{L^2(\R^3)}.\]
	By density of $H^1$ in $L^2$, we can extend this inequality to all $\psi\in L^2(\R^3)$, which implies that
	\begin{equation}\label{TzetaR}
	\|[T,\zeta_R ]\|_{\cB(L^2)}\leq\frac{1}{(2\pi)^{3/2}}\||\cdot|\cF(\zeta_R) \|_{L^1(\R^3)}.
	\end{equation}
	
	Let us estimate $\||\cdot|\cF(\zeta_R) \|_{L^1(\R^3)}$.
	We have that
	\begin{align*}
		\int_{\mathbb{R}^3}|q||\cF(\zeta_R)(q)|\d q&=\frac{1}{(2\pi)^{3/2}}\int_{\mathbb{R}^3}|q|\left|\int_{\mathbb{R}^3}e^{-\iu q\cdot x}\zeta\left(\frac{x}{R}\right)\d x\right|\d q\nonumber
		=\frac{1}{(2\pi)^{3/2}}\int_{\mathbb{R}^3}\left|\frac{p}{R}\right|\left|\int_{\mathbb{R}^3}e^{-\iu p\cdot y}\zeta(y)\d y\right|\d p,
	\end{align*}
	where we have performed the changes of variable $y=\frac{x}{R}$ and $p=Rq$. We find that the norm is equal to $R^{-1}(2\pi)^{-3/2}\int_{\R^3}|p|\left|\int_{\R^3} e^{-\iu p\cdot y}\zeta(y)\d y\right|\d p$. Since $\zeta$ is a Schwartz function, we know that this integral is finite. As a consequence,
	\begin{align}\label{est-chi}
		\int_{\mathbb{R}^3}|q||\cF(\zeta_R)(q)|\d q \leq \frac{C}{R},
	\end{align}
	which together with \eqref{TzetaR} concludes the proof. 
\end{proof}
One consequence of this result is that we can get a result analogous to what is called IMS localization formula in the non-relativistic case, see \textit{e.g.} \cite[Theorem~3.2]{CFKS}.

Recall that $\chi_1$ was defined in \eqref{def:chiL}.
Let us define $v_1:=\chi_1$, $v_{2}:=v_1(\cdot-e_1)$ and  $v_{0}:=1-v_{1}-v_{2}$. We define too $V:=\sqrt{v_0^2+v_1^2+v_2^2}$, which by construction satisfies $V\geq1/\sqrt{2}$. For $R>0$ and  $i=0,1,2$, we define the functions
\begin{equation}\label{def:JiR}
	J_{i,R}(x):=\frac{v_i(x/R)}{V(x/R)}.
\end{equation}
We have that \[\sum_{i=0}^2J_{i,R}(x)^2=1\] for all $x$ and $R$.

We  define the \emph{one-particle localization error} for $\psi \in H^{1/2}(\mathbb{R}^{3})$ by
\begin{equation*}
	\text{Err}[\psi] := \langle \psi,T\psi \rangle - \sum_{i=0}^2 \langle J_{i,R}\psi,T J_{i,R}\psi\rangle.
\end{equation*}
\begin{lemma}
	There exists $C>0$  such that, for all $R>0$ and $\psi\in H^{1/2}(\R^3)$,
	\begin{equation*}\left|\mathrm{Err}[\psi]\right|\leq\frac{C}{R}\|\psi\|_{L^2(\R^3)}^2.\end{equation*}
\end{lemma}
\begin{proof}
	We remark that $\langle \psi,T\psi \rangle = \sum_{i=0}^2 \langle \psi,T J_{i,R}^2\psi\rangle$ so
	\[|\text{Err}[\psi]|\leq\sum_{i=0}^2 |\langle \psi,[T ,J_{i,R}]J_{i,R}\psi\rangle|\leq \sum_{i=0}^2\|[T ,J_{i,R}]\|_{\cB(L^2)}\|\psi\|_{L^2(\R^3)}^2\]
	since $\|J_{i,R}\|_{L^\infty}\leq1$ for all $i$ and $R$.
	For $i=1$ or 2, the function $J_{i,R}$ satisfies the hypothesis of Lemma~\ref{lem:expcommut1overL} so $\|[T ,J_{i,R}]\|_{\cB(L^2)}\leq C/R$. Concerning the term $J_{0,R}$, we use the fact that $[T ,J_{0,R}]=-[T ,1-J_{0,R}]$. But $1-J_{0,R}$ is compactly supported. Thus, $1-J_{0,R}$ satisfies the hypotheses of Lemma~\ref{lem:expcommut1overL}, and we can conclude with the same argument.
\end{proof}
This result can be extended to the case of $N$ electrons.
Let $\gamma=(\gamma_1,...,\gamma_N)\in\{0,1,2\}^N$.

 We define
\begin{equation}\label{def:Jgamma}
	J_{\gamma,R}(x_1,\dots,x_N) :=  \prod_{i=1}^N J_{\gamma_i,R}(x_i),
\end{equation}
where $J_{\gamma_i,R}$ was defined in \eqref{def:JiR}.

We  define the \emph{$N$-particle localization error} for $\psi \in H^{1/2}(\mathbb{R}^{3N})$ by
\begin{equation*}
	\text{Err}[\psi] := \Big\langle \psi,\sum_{i=1}^NT_i\psi \Big\rangle - \sum_{\gamma \in \{0,1,2\}^N} \Big\langle J_{\gamma,R}\psi,\sum_{i=1}^NT_i J_{\gamma,R}\psi\Big\rangle.
\end{equation*}
\begin{proposition}\label{thm:imsloc}There exists $C>0$  such that, for all $R>0$ and $\psi\in H^{1/2}(\R^{3N})$,
	\begin{equation}\label{eq:imsloc}
		|\mathrm{Err}[\psi]|\leq\frac{C}{R}\|\psi\|^2_{L^2(\R^{3N})}.
	\end{equation}
	
\end{proposition}
\begin{proof}
	The proof is similar to the one-particle case.
\end{proof}
\begin{remark}
	As explained in~\cite[Theorem~3.1]{bhv}, it is even possible to have a control of the localization error by $C/R^2$. Nevertheless, we do not need it in our work and it requires a  longer proof.
\end{remark}

The next lemma states that, even if the commutator is not local, we are  able to control this nonlocality  in the sense that, if we apply it to a function which has some decay at infinity, then we obtain something which is decaying, at least in a quadratic form sense. 
\begin{lemma}\label{lem:expcommutdecay}
	Let us consider, for $R>0$, the function $\chi_R$ defined  in~\eqref{def:chiL}. Then there exists $C>0$ such that,
	for all functions $F:\R^+\to [1,\infty)$, all $R$ large enough and all $\xi,\eta \in L^2(\mathbb{R}^3)$  such that $F(|\cdot|)\eta \in L^2(\mathbb{R}^3)$,
	\begin{equation*}
		|\langle \eta , [T,\chi_R]\xi\rangle| \leq C\max\left\{\frac{1}{R}\sup_{r\in[R/20,+\infty)}\frac{1}{F(r)},e^{-3R/80}\right\}\|F(|\cdot|)\eta\|_{L^2(\R^3)}\|\xi\|_{L^2(\R^3)}.
	\end{equation*}
	
\end{lemma}

\begin{remark}
	Applying Lemma \ref{lem:expcommutdecay} for $F(x)=e^{ax}$, $a>0$ we find
	\begin{equation}\label{expsmall}
		|\langle \eta , [T,\chi_R]\xi\rangle| \leq Ce^{-\min(a/20,3/80)R}\|e^{a|\cdot|}\eta\|_{L^2(\R^3)}\|\xi\|_{L^2(\R^3)}.
	\end{equation}
\end{remark}

\begin{proof}
	
	We introduce the function $\chi_{R/2}^c := 1 - \chi_{R/2}$ and decompose
	\begin{equation}\label{exp:commut}
		\langle \eta , [T ,\chi_R] \xi\rangle = \langle\chi_{R/2} \eta , [T ,\chi_R] \xi\rangle + \langle \chi_{R/2}^c \eta ,[T ,\chi_R] \xi\rangle .\end{equation}
	
	On the one hand, 
	\begin{align}
		|\langle\chi_{R/2}^c \eta , [T ,\chi_R] \xi\rangle|&\leq\|F(|\cdot|)\eta\|_{L^2(\R^3)}\left\|\frac{1}{F(|\cdot|)}\chi_{R/2}^c\right\|_{L^\infty(\R^3)}\|[T ,\chi_R] \xi\|_{L^2(\R^3)}\nonumber\\
 		&\leq \frac{C}{R}\left\|\frac{1}{F(|\cdot|)}\chi_{R/2}^c\right\|_{L^\infty(\R^3)}\|F(|\cdot|)\eta\|_{L^2(\R^3)}\|\xi\|_{L^2(\R^3)}\nonumber
	\\&\leq\frac{C}{R}\sup_{r \in [R/20,+\infty)}\frac{1}{F(r)}\|F(|\cdot|)\eta\|_{L^2(\R^3)}\|\xi\|_{L^2(\R^3)},\label{bound:1stcommut}
	\end{align}
	where we used Lemma \ref{lem:expcommut1overL} to go from the first to the second line.
	
	On the other hand, 
	\begin{equation*}
		\langle\chi_{R/2} \eta , [T ,\chi_R] \xi\rangle=-\langle\chi_{R/2} \eta , [T ,\chi_R^c] \xi\rangle=-\langle\chi_{R/2} \eta , T \chi_R^c \xi\rangle
	\end{equation*}
	since $\chi_R^c\chi_{R/2}=0$. But $\chi_{R/2} \eta$ and $\chi_R^c \xi$ are both $L^2$ functions and the distance between their support is $3R/80$. Hence, we get from Lemma~\ref{expdecay:realT} that, for $R$ large enough:
	\begin{equation}\label{bound:2ndcommut}
		|\langle\chi_{R/2} \eta , T \chi_R^c \xi\rangle|\leq Ce^{-3R/80} \|\eta\|_{L^2(\R^3)} \|\xi\|_{L^2(\R^3)}.
	\end{equation}
	Inserting~\eqref{bound:1stcommut} and~\eqref{bound:2ndcommut} into~\eqref{exp:commut} concludes the proof.
\end{proof}

\bigskip

\section{Proof of Theorems \ref{thm:HVZ}, \ref{thm:expdecaystat} and \ref{thm:Zhislin}}\label{pfthmspec}

\subsection{Proof of Theorem \ref{thm:HVZ}}
The proof  closely follows the one which was given for atoms in~\cite[Appendix~A]{bhv}. It is split into 2 parts. In the first one, we prove that 
\[[\inf \sigma (\hat{H}_{N-1}(Y,\cZ)), \infty  ) \subset \sigma_{\mathrm{ess}}(\hat{H}_{N}(Y,\cZ))\] and in the second one we prove the reverse inclusion.

``Easy part'':
We  prove here
\begin{equation}\label{HVZEasydirection}
	[\inf \sigma (\hat{H}_{N-1}(Y,\cZ)), \infty  ) \subset \sigma(\hat{H}_{N}(Y,\cZ)). 
\end{equation}
Since the left-hand side is an interval, it will be in the essential spectrum.

We denote $E_{N-1}:= \inf \sigma (\hat{H}_{N-1}(Y,\cZ))$. Let us pick $y \in [E_{N-1},\infty)$. We want to prove that $y \in \sigma(\hat{H}_{N}(Y,\cZ))$. By assumption on $y$, 
there exists $a \geq 0$ such that $y= E_{N-1} +a$. Let $\epsilon >0$. Then, since $E_{N-1} \in  \sigma (\hat{H}_{N-1}(Y,\cZ))$,
there exists by Weyl's criterion (cf.~\cite[Theorem~VII.12]{RS1}) a $\phi_{N-1,\epsilon} \in \Dom(\hat{H}_{N-1}(Y,\cZ))$, normalized,  such that
\begin{equation}\label{phimistake}
	\|(H_{N-1}(Y,\cZ)-E_{N-1}) \phi_{N-1,\epsilon}\| \leq \frac{\epsilon}{3\sqrt{N}}.
\end{equation}
 But compactly supported functions are dense in $\Dom(\hat{H}_{N-1}(Y,\cZ))$ in the graph norm. Indeed, if we consider a function $\zeta\in\Dom(\hat{H}_{N-1}(Y,\cZ))$ and, for $R>0$, the function $\chi_R$ defined in \eqref{def:chiL}, we have that
\[\|\hat{H}_{N-1}(Y,\cZ)(\chi_R^{\otimes N-1}\zeta-\zeta)\|\leq\|(1-\chi_R^{\otimes N-1})\hat{H}_{N-1}(Y,\cZ)\zeta\|+\left\|\sum_{i=1}^N[T_i,\chi_R^{\otimes N-1}]\zeta\right\|.\] Since  $\zeta\in\Dom(\hat{H}_{N-1}(Y,\cZ))$, we have $\hat{H}_{N-1}(Y,\cZ)\zeta\in L^2(\R^{3(N-1)})$ and therefore 
$$\|(1-\chi_R^{\otimes N-1})\hat{H}_{N-1}(Y,\cZ)\zeta\|\to 0, \quad \text{when } R\to\infty. $$
On the other hand, by Lemma~\ref{lem:expcommut1overL}, $\|\sum_{i=1}^N[T_i,\chi_R^{\otimes N-1}]\zeta\|\leq C/R\|\zeta\|$, which goes to 0 as well when $R$ goes to $\infty$. As a consequence of this density argument, we can assume, without loss of generality, that $\phi_{N-1,\epsilon}$ is compactly supported.
Moreover, since $\sigma(T) = [0,+\infty) $ and therefore $a \in \sigma(T)$, there exists $f_\epsilon \in C^\infty_c(\R^3)$, normalized, such that 
\begin{equation}\label{Tmistake}
	\|(T-a)f_\epsilon\|\leq  \frac{\epsilon}{3\sqrt{N}}.
\end{equation}

For $h\in\R^3$, we consider the function $f_{\epsilon,h}$ defined on $\R^3$ by
$f_{\epsilon,h}(x):= f_\epsilon(x-h)$.
We further define
\begin{equation}\label{def:geh}
	g_{\epsilon,h}:=\frac{Q_N (\phi_{N-1,\epsilon} \otimes f_{\epsilon,h}) }{\| Q_N(\phi_{N-1,\epsilon} \otimes f_{\epsilon,h})\|}.
\end{equation}

We will prove that, if we choose $|h|$ large enough, $g_{\epsilon,h}$ is in the domain of $\hat{H}_{N}(Y,\cZ)$ and
\begin{equation*}
	\|(H_{N}(Y,\cZ)-y)g_{\epsilon,h}\|< \epsilon.
\end{equation*}

Let us first compute  the norm which appears in the denominator of~\eqref{def:geh}. We observe that, if $|h|$ is large enough, then the  electronic densities (defined similarly to~\eqref{def-elec-dens}) of $\phi_{N-1,\epsilon}$ and  $f_{\epsilon,h}$ have disjoint supports. We can therefore argue similarly to the proof of~\eqref{norm-QNphi}, with $N_1=N-1$, to find   

\begin{equation*}
	\| Q_N(\phi_{N-1,\epsilon} \otimes f_{\epsilon,h})\|=\frac{1}{\sqrt{N}}.
\end{equation*}

As a consequence we have that
\begin{align}\label{HNg}
	\|(H_{N}(Y,\cZ)-y)g_{\epsilon,h}\|
	\leq \sqrt{N}  \|(H_{N}(Y,\cZ)-y)(\phi_{N-1,\epsilon} \otimes f_{\epsilon,h})\|,
\end{align}
where in the last step we used that $Q_N$ commutes with $H_N(Y,\cZ)$ and that it is an orthogonal projection. 
But we have that
\begin{equation*}
	(H_{N}(Y,\cZ)-y) =(H_{N-1}(Y,\cZ)-E_{N-1})
	+(T_N - a) + \sum_{i=1}^{N-1} \frac{\alpha}{|x_i-x_N|} - \sum_{j=1}^{M} \frac{Z_j\alpha}{|y_j-x_N|},
\end{equation*}
so we get from \eqref{phimistake}, \eqref{Tmistake} and~\eqref{HNg} that
\begin{equation*}
	\|(H_{N}(Y,\cZ)-y)g_{\epsilon,h}\| \leq \frac{2 \epsilon }{3} + \sqrt{N} \left\|\left(\sum_{i=1}^{N-1}\frac{\alpha}{|x_i-x_N|} - \sum_{j=1}^{M} \frac{Z_j \alpha}{|y_j-x_N|} \right)(\phi_{N-1,\epsilon} \otimes f_{\epsilon,h}) \right\|.
\end{equation*}
We can choose $|h|$ large enough so that, for all $x_N\in\supp f_{\epsilon,h}$ and $(x_1,...,x_{N-1})\in\supp \phi_{N-1,\epsilon}$,  $ \sum_{i=1}^{N-1}\frac{\alpha}{|x_i-x_N|}<\epsilon/(6\sqrt{N})$ and $\sum_{j=1}^{M} \frac{Z_j \alpha}{|y_j-x_N|}<\epsilon/(6\sqrt{N})$. In such a case, we have that $g_{\epsilon,h}$ is in the domain of $\hat{H}_N(Y,\cZ)$ and
\begin{equation*}
	\|(\hat{H}_{N}(Y,\cZ)-y)g_{\epsilon,h}\| 
	<  \epsilon, 
\end{equation*}
concluding the proof of \eqref{HVZEasydirection} since $\epsilon$ is arbitrarily small.

\medskip

``Hard part": Let us prove that, for all $N$, $Y$, $\cZ$, 
\begin{equation*}
	\sigma_{\text{ess}}(\hat{H}_{N}(Y,\cZ)) \subset [E_{N-1}, \infty).
\end{equation*}
Let $\psi_n \in \text{Ran}~Q_N$ be a Weyl sequence for $\hat{H}_{N}(Y,\cZ)$: for some  $\lambda\in\sigma_{\text{ess}}(\hat{H}_{N}(Y,\cZ))$, $\|(\hat{H}_{N}(Y,\cZ)-\lambda)\psi_n\|\to0$, with $\|\psi_n\|=1$ for all $n$ and $\psi_n\rightharpoonup 0$.
We will show that
\begin{equation}\label{HVZhardgoal}
	\lim_{n \rightarrow \infty} 
	\langle \psi_n, \hat{H}_{N}(Y,\cZ) \psi_n \rangle \geq E_{N-1}. 
\end{equation}

To this purpose, we use a method similar to the one given in Proposition~\ref{thm:imsloc}. The difference is that we consider only 2 functions instead of 3. We define again $v_1:=\chi_1$ and we choose $v_0:=1-v_1$. We define then, for $\beta\in\{0,1\}^N$, $J_{\beta,R}$ similarly as in \eqref{def:Jgamma}.  On  $\supp J_{\beta,R}$,  $x_i$ is close to the origin if $\beta_i=1$ and far from the origin if $\beta_i=0$. 
We have the same bound on the error as in Proposition~\ref{thm:imsloc}, 
implying that there exist $C>0$ such that,  for all $n\in\N$ and $R>1$,
\begin{equation}\label{semIMSHVZ}
	\langle \psi_n, H_{N}(Y,\cZ) \psi_n \rangle \geq \sum_{\beta\in\{0,1\}^N} \langle J_{\beta,R} \psi_n, \hat{H}_{N}(Y,\cZ) J_{\beta,R} \psi_n \rangle - \frac{C}{R}.
\end{equation}

Let us estimate the sum in the right-hand side of~\eqref{semIMSHVZ}.
First, we can observe that, for $B:=(1,...,1)$, corresponding to the case where all electrons are close to the origin, the function $J_{B,R}\psi_n\in\mathrm{Ran}Q_N$ and thus
\begin{equation}\label{semIMSHVZclose}
	\langle J_{B,R} \psi_n, H_{N}(Y,\cZ) J_{B,R} \psi_n \rangle \geq E_N \|J_{B,R} \psi_n\|^2,
\end{equation}
where we recall that $E_N=\inf \sigma(\hat{H}_{N}(Y,\cZ))$. 

Now, for $\beta\neq B$, \emph{i.e.} when at least one electron is far away,
\begin{multline}\label{decomp}
	\langle J_{\beta,R} \psi_n, H_{N}(Y,\cZ) J_{\beta,R} \psi_n \rangle =\langle J_{\beta,R} \psi_n, H_{\beta}(Y,\cZ) J_{\beta,R} \psi_n\rangle \\+ \langle J_{\beta,R} \psi_n, (H_{N}(Y,\cZ)-H_{\beta}(Y,\cZ)) J_{\beta,R} \psi_n\rangle,
\end{multline} 
where $H_{\beta}(Y,\cZ)$ is the Hamiltonian describing the electrons with index $i$ such that $\beta_i=1$, interacting  between them and with all nuclei, namely \[H_{\beta}(Y,\cZ):=\sum_{\{i,\beta_i=1\}}\left(T_i-\sum_{j=1}^M\frac{Z_j \alpha}{|x_i-y_j|}\right)+\sum_{\substack{i<j\\\beta_i=\beta_j=1}}\frac{\alpha}{|x_i-x_j|}+\sum_{1 \leq k<l\leq M}\frac{Z_kZ_l\alpha}{|y_k-y_l|}.\]

On the one hand, \[H_{N}(Y,\cZ)-H_{\beta}(Y,\cZ)\geq -\sum_{\{i,\beta_i=0\}}\sum_{j=1}^M\frac{Z_j\alpha}{|x_i-y_j|}.\] But, for $R$ large enough, if $x_i$ is in  the support of $v_0$, there exists a constant $C'$ such that, for all $j$, $|x_i-y_j|\geq C'R$. As a consequence, there exists $C>0$ such that
\begin{equation}\label{term:diff}
	\langle J_{\beta,R} \psi_n, (H_{N}(Y,\cZ)-H_{\beta}(Y,\cZ)) J_{\beta,R} \psi_n\rangle\geq -\frac{C}{R}.
\end{equation}

Let $K=\sum_{i=1}^{N} \beta_i$, namely $K$ describes the number of electrons that are close to the origin. We observe that
 $J_{\beta,R} \psi_n$ is antisymmetric with respect to the $K$ variables $x_i$ for which $\beta_i=1$. As a consequence, 
\begin{equation}\label{terme1}
	\langle J_{\beta,R} \psi_n, H_{\beta}(Y,\cZ) J_{\beta,R} \psi_n\rangle \geq  
	E_{K} \| J_{\beta,R} \psi_n \|^2
	\geq E_{N-1} \| J_{\beta,R} \psi_n \|^2
\end{equation} since, by \eqref{HVZEasydirection}, the sequence $(E_k)$ is nonincreasing. 
Inserting~\eqref{term:diff} and \eqref{terme1} into~\eqref{decomp} and combining it with \eqref{semIMSHVZ} and \eqref{semIMSHVZclose},  we obtain that
\begin{equation*}
	\langle \psi_n, H_{N}(Y,\cZ) \psi_n \rangle \geq E_N \|J_{B,R} \psi_n\|^2  + E_{N-1} \sum_{\beta\neq B}  \|J_{\beta,R} \psi_n\|^2 - \frac{C}{R},
\end{equation*}
and, since $ \sum_{\beta\neq B} J_{\beta,R}^2= 1-J_{B,R}^2$ and $\|\psi_n\|=1$, we find that
\begin{equation*}
	\langle \psi_n, H_{N}(Y,\cZ) \psi_n \rangle \geq E_{N-1} +  (E_N - E_{N-1}) \|J_{B,R} \psi_n\|^2 -\frac{C}{R}.
\end{equation*}

We will now prove that, for all $R>0$, we have that
\begin{equation*}
	\lim_{n\to \infty}\| J_{B,R} \psi_n \|  = 0
\end{equation*}
and this will conclude the proof of \eqref{HVZhardgoal} and therefore of Theorem~\ref{thm:HVZ}. 
This comes from the fact that, for $c>0$ large, we have the following decomposition: 
\begin{equation*}
J_{B,R} \psi_n = \underbrace{J_{B,R} \left(\sum_{i=1}^N T_{i} +c\right)^{-\frac{1}{2}}}_{\text{compact operator }}  \underbrace{\left(\sum_{i=1}^N T_{i} +c\right)^\frac{1}{2} (H_{N}(Y,\cZ) +c)^{-\frac{1}{2}}}_{\text{ bounded operator }} \underbrace{(H_{N}(Y,\cZ) +c)^{-\frac{1}{2}}}_{\text{ bounded}} \underbrace{(H_{N}(Y,\cZ) +c)\psi_n}_{\rightharpoonup 0}.
\end{equation*}
 The weak convergence $(H_{N}(Y,\cZ) +c)\psi_n \rightharpoonup 0$ follows from the decomposition \[H_{N}(Y,\cZ) +c=(H_{N}(Y,\cZ)-\lambda) +(\lambda+c).\] We know from the definition of $(\psi_n)$ that $(H_{N}(Y,\cZ)-\lambda)\psi_n\to0$ and that $(\lambda+c)\psi_n\rightharpoonup0$. The operator  $(H_{N}(Y,\cZ) +c)^{-\frac{1}{2}}$ is bounded since $H_{N}(Y,\cZ) + c$ is strictly positive for $c$ large enough. The boundedness of 
$(\sum_{i=1}^N T_i +c)^\frac{1}{2} (H_{N}(Y,\cZ) +c)^{-\frac{1}{2}}$ follows from~\eqref{rel-bound}. Finally, the operator $J_{B,R} \left(\sum_{i=1}^N T_{i} +c\right)^{-\frac{1}{2}}$ is compact, since it is the norm limit of the sequence    of Hilbert-Schmidt operators $J_{B,R} \left(\sum_{i=1}^N T_{i} +c\right)^{-\frac{1}{2}} \mathbbm{1}_{\{\sum_{i=1}^N T_{i} +c<n\}}$. We thus obtain that
$\|J_{B,R} \psi_n\| \rightarrow 0$, concluding the proof of the theorem. 
\qed

\bigskip

\subsection{Proof of Theorem \ref{thm:expdecaystat}}

We prove in the following  exponential decay of  approximate eigenfunctions for the Hamiltonian $\hat{H}_N$. The proof closely follows the proof of Theorem 2.1 in \cite{bhv}.

Let $(\gamma_s)$ be as in the statement of Theorem~\ref{thm:expdecaystat}. For $\nu\geq 0$ and $\epsilon \geq 0$,
we set for $r\in\R^+$ \begin{equation*}
	G_{\nu,\epsilon}(r):=\frac{\nu r}{1+\epsilon r}.	          \end{equation*}
We then define  on $\R^{3N}$ the function $F_{\nu,\epsilon}$ by
\begin{equation*}
	F_{\nu,\epsilon}(x_1,...,x_N):=\sum_{i=1}^NG_{\nu,\epsilon}(|x_i|).
\end{equation*}
Moreover, for the function $\chi_R$ defined in \eqref{def:chiL}, we introduce
\begin{equation}\label{def:xieps}
	\xi_{\eps} := (1-\chi_R^{\otimes N}) e^{F_{\nu,\eps}}.
\end{equation}

By the definition of $\widetilde{\Sigma}$ in \eqref{ionthreshstatistic}, there exists a function $\theta(R)$ with $\lim_{R \rightarrow + \infty} \theta(R) = 0$ such that for all $s \in \cI$
\begin{align*}
	(\widetilde{\Sigma} - \mu - \theta(R))\; \|\xi_{\eps}\gamma_s\|^2 
	\leq \langle \xi_{\eps}\gamma_s, (H_{N}(Y,\cZ) -\mu) \xi_{\eps}\gamma_s\rangle 
	 = \Re  \langle \gamma_s, \xi_{\eps}^2 \Gamma_s \rangle + \Re  \left\langle \xi_{\eps}\gamma_s,  \sum_{i=1}^N[T_i,\xi_{\eps}] \gamma_s \right\rangle,
\end{align*}
where in the last step we used \eqref{gs:unif} and that $\xi_\epsilon$ commutes with potentials.
Using the assumption that, for all $s \in \cI$, $\|e^{a |\,\cdot\,|} \Gamma_s\|_{L^2(\mathbb{R}^{3N})}\leq C_1$, we can bound, for all $R$, $\epsilon$, $\nu<\frac{a}{2\sqrt{N}}$ and $s$, the first term on the right-hand side and find
\begin{equation}\label{midstep:exp}
	(\widetilde{\Sigma} - \mu - \theta(R))\; \|\xi_{\eps}\gamma_s\|^2 \leq C_1 \|\gamma_s \| + \Re  \left\langle \xi_{\eps}\gamma_s,  \sum_{i=1}^N[T_i,\xi_{\eps}] \gamma_s \right\rangle. 
\end{equation}

By \cite[Lemma~C.10, Remark C.7]{bhv}, we know that for all $\nu\in(0,1)$, there exist $C_{\nu } > 0$ such that, for all $R$,
\begin{equation}\label{lem:C10}
	\left|\Re  \left\langle \xi_{\eps}\gamma_s, \sum_{i=1}^N[T_i,\xi_{\eps}] \gamma_s \right\rangle\right| \leq NC_\nu\left(\frac{L_\chi}{R}+\nu\right)^2 \left\|  e^{F_{\nu,\eps}} \gamma_s\right\|^2 ,  
\end{equation}
where $L_\chi$ is the Lipschitz constant of $\chi$, recall \eqref{def:chiL}. Moreover, $C_\nu$ goes to a finite value when $\nu$ goes to 0.
Let us introduce the abbreviation $\delta:=NC_\nu\left(\frac{L_\chi}{R}+\nu\right)^2$.
Plugging~\eqref{lem:C10} into \eqref{midstep:exp} yields, in view of~\eqref{def:xieps},
\begin{align*}
	(\widetilde{\Sigma} - \mu - \theta(R))\; \|\xi_{\eps}\gamma_s\|^2 &\leq  C_1 \|\gamma_s \| +\delta \left\|  e^{F_{\nu,\eps}} \gamma_s\right\|^2 \\
	&\leq C_1 \|\gamma_s \| +2\delta \|\xi_{\eps}\gamma_s\|^2 + 2\delta \left\|  \chi_R^{\otimes N} e^{F_{\nu,\eps}} \gamma_s\right\|^2  \\
	&\leq C_1 \|\gamma_s \| +2\delta \|\xi_{\eps}\gamma_s\|^2 +2\delta  e^{N\nu R/4}\|\gamma_s\|^2 ,
\end{align*}
where in the last line we used that, for all $x \in \R^{3N}$,
\begin{equation}\label{expcomp}
\chi_R^{\otimes N} (x) e^{F_{\nu,\eps}(x)} \leq e^{N\nu R/8}.
\end{equation}
 We have therefore obtained
\begin{equation}\label{laststep:exp}
	(\widetilde{\Sigma} - \mu - \theta(R) -2\delta) \|\xi_{\eps} \gamma_s\|^2 \leq C_1 \|\gamma_s \| + 2\delta  e^{N\nu R/4}\|\gamma_s\|^2 .
\end{equation}
But
\[\lim_{R\to\infty}\theta(R)+2\delta=2NC_\nu\nu^2,\]
so, for all $\nu$ such that $2NC_\nu\nu^2<\widetilde{\Sigma} - \mu$, we can find a radius $R$ such that 

\begin{equation*}
	g:=(\widetilde{\Sigma} - \mu - \theta(R) -2\delta) >0,
\end{equation*}
and together with \eqref{laststep:exp} we obtain:
\begin{equation*}
	g\|\xi_{\eps} \gamma_s\|^2 < C_1 \|\gamma_s\| +2\delta e^{N\nu R/4}\|\gamma_s\|^2 \leq C_R,
\end{equation*}
where in the last step we used the uniform boundedness of $\|\gamma_s\|$, see \eqref{gs:unif}. The constant $C_R>0$ does not depend on $\eps$. 
Since $\xi_{\eps}$ converges monotonically to $(1-\chi_R^{\otimes N})e^{\nu |x|}$ for $\eps \rightarrow 0$, the monotone convergence theorem for the previous inequality gives that for all $s$ and $\nu$, we have $\|(1-\chi_R^{\otimes N})e^{\nu |x|}\gamma_s\|^2\leq g^{-1}C_R$, which together with \eqref{expcomp} concludes the proof of Theorem \ref{thm:expdecaystat}.
\qed

\bigskip

\subsection{Proof of Theorem \ref{thm:Zhislin}}

We begin by showing that there exists at least one eigenvalue below $\sigma_{\text{ess}}(\hat{H}_{N}(Y,\cZ))$. As \textit{e.g.} in \cite[Appendix~B]{bhv}, we prove it by induction on the number of electrons $N$.

Let us first consider the case $N=1$. Here, we do not need to take symmetry  into account. Let $f \in C_c^\infty (\R^3)$ be such that  $\|f\|_{L^2}=1$ and $\supp f \subset \{x \in \R^3: \frac{5}{4} \leq |x| \leq \frac{7}{4}\}$.  For $R>0$, we define
\begin{equation*}
	f_R(x):=\frac{1}{R^{\frac{3}{2}}} f\left(\frac{x}{R}\right).
\end{equation*}

Then, for all $R>0$, $\|f_R\|=1$ and
\begin{align} \nonumber
	\langle f_R, H_{1}(Y,\cZ) f_R \rangle 
	= & \left\langle f_R, \left(T_1 - \sum_{j=1}^M \frac{Z_j \alpha}{|x_1-y_j|} \right) f_R \right\rangle + \sum_{1\leq i<j \leq M} \frac{ Z_i Z_j \alpha}{|y_i-y_j|}  \\ \nonumber
	\leq & \left\langle f_R,  \left(-\Delta_1 - \sum_{j=1}^M \frac{Z_j \alpha}{|x-y_j|}\right)  f_R \right\rangle + \sum_{1\leq i<j\leq M} \frac{ Z_i Z_j \alpha}{|y_i-y_j|} \\
	\leq  & \frac{1}{R^2}  \langle f,  -\Delta_1  f\rangle- \sum_{j=1}^M \left\langle f, \frac{Z_j \alpha}{|Rx_1-y_j|}  f \right\rangle + \sum_{1\leq i<j \leq M} \frac{ Z_i Z_j\alpha}{|y_i-y_j|},\nonumber
\end{align}
where we have used that $T_1 \leq -\Delta_1$ and a change of variables in the second and third line, respectively. 
With a Taylor expansion, we see that the second term decays as $-\sum_{j=1}^M \alpha Z_j/R$ when $R$ becomes large. This proves that, for $R$ large enough,
\[\langle f_R, H_{1}(Y,\cZ) f_R \rangle< \sum_{1\leq i<j\leq M} \frac{ Z_i Z_j \alpha}{|y_i-y_j|}= \inf \sigma_{\text{ess}}(H_1(Y,\cZ)),\]
implying that $\hat{H}_1(Y,\cZ)$ has at least one eigenvalue below the infimum of its essential spectrum.	

We consider now the induction step. Let us assume, for some $N \in \N$ with $N<|Z|+1$,
\[E_{N-1}:=\inf \sigma(\hat{H}_{N-1}(Y,\cZ))< \inf\sigma_{\text{ess}}(\hat{H}_{N-1}(Y,\cZ)),\] which implies that there exists a ground state $\Psi_{N-1}\in \Dom(\hat{H}_{N-1}(Y,\cZ))$ with $\|\psi_{N-1}\|=1$ such that
$H_{N-1}(Y,\cZ) \Psi_{N-1}= E_{N-1} \Psi_{N-1}$.
By the assumption of Theorem \ref{thm:Zhislin}, $|\cZ|>N-1$ so  we can choose and fix $\delta >0$ with 
\begin{equation}\label{deltachoice}
\delta<\frac{|\cZ| + 1 - N}{2(N-1 + |\cZ|)}.
\end{equation}
For $R$ large enough, $\chi_{\delta R}^{\otimes N-1}\Psi_{N-1}$ is not identically 0, where  $\chi_{\delta R}$ was defined in \eqref{def:chiL}.
Thus, we can define the sequence of trial functions:
\begin{equation}\label{def:PhiNR}
	\Phi_{N,R}:=\frac{Q_N((\chi_{\delta R}^{\otimes N-1}\Psi_{N-1} )\otimes f_R)}{\|Q_N((\chi_{\delta R}^{\otimes N-1}\Psi_{N-1}) \otimes f_R)\|}.
\end{equation}
For all $R$, we have $\Phi_{N,R}\in H^{1/2}(\R^{3N})\cap\Ran Q_N$.

Arguing as in the proof of~\eqref{eq:getridofQ}, we find that
\begin{align}\label{HN-ENPhi}
	\langle \Phi_{N,R},(\hat{H}_N-E_{N-1})\Phi_{N,R}\rangle=\frac{\langle (\chi_{\delta R}^{\otimes N-1}\Psi_{N-1} )\otimes f_R,(H_N-E_{N-1})(\chi_{\delta R}^{\otimes N-1}\Psi_{N-1} )\otimes f_R \rangle}{\|(\chi_{\delta R}^{\otimes N-1}\Psi_{N-1}) \otimes f_R\|^2}.
\end{align}

Let us first estimate the denominator of the right-hand side of~\eqref{HN-ENPhi}. Since, by Corollary~\ref{cor:expeigendecay}, $\Psi_{N-1}$ is exponentially decaying and  we have chosen $\Psi_{N-1}$ and $f_R$ with norm 1, we find that 
\begin{equation}\label{norm:chiPsi}\|(\chi_{\delta R}^{\otimes N-1}\Psi_{N-1}) \otimes f_R\|=1+O(e^{-cR}).\end{equation}

To estimate the numerator of the right-hand side of~\eqref{HN-ENPhi}, let us split 
\begin{equation}\label{split}
	H_N -E_{N-1} = H_N - H_{N-1} + H_{N-1} - E_{N-1}.
\end{equation}

Using that $( H_{N-1} - E_{N-1})$ does not act on the coordinate $x_N$ and that $( H_{N-1} - E_{N-1})\Psi_{N-1}=0$, we find 
\begin{align*} \langle(\chi_{\delta R}^{\otimes N-1}&\Psi_{N-1}\otimes f_R),(H_{N-1}-E_{N-1})(\chi_{\delta R}^{\otimes N-1}\Psi_{N-1}\otimes f_R)\rangle
=\Big\langle\chi_{\delta R}^{\otimes N-1}\Psi_{N-1},\sum_{j=1}^{N-1}[T_j,\chi_{\delta R}^{\otimes N-1}]\Psi_{N-1}\Big\rangle.
\end{align*}

As a consequence, since, by Corollary~\ref{cor:expeigendecay}, $\Psi_{N-1}$  is exponentially decaying, we see, using \eqref{expsmall}, that 

\begin{equation}\label{diag0}
	\langle(\chi_{\delta R}^{\otimes N-1}\Psi_{N-1}\otimes f_R),(H_{N-1}-E_{N-1})(\chi_{\delta R}^{\otimes N-1}\Psi_{N-1}\otimes f_R)\rangle =O(e^{-dR}).
\end{equation}
 On the other hand, using  $T\leq\frac{-\Delta}{2}$, we find
\begin{multline}\label{HNHN-1}
	\langle(\chi_{\delta R}^{\otimes N-1}\Psi_{N-1}\otimes f_R),(H_{N}-H_{N-1})(\chi_{\delta R}^{\otimes N-1}\Psi_{N-1}\otimes f_R)\rangle  \\
	\leq \left\langle\chi_{\delta R}^{\otimes N-1}\Psi_{N-1}\otimes f_R,\left(-\frac{\Delta_N}{2}-\sum_{j=1}^M\frac{ Z_j \alpha}{|x_N-y_j|}+\sum_{j=1}^{N-1}\frac{\alpha}{|x_N-x_j|}\right)(\chi_{\delta R}^{\otimes N-1}\Psi_{N-1}\otimes f_R)\right\rangle.
\end{multline}

Assume that $R$ is large enough such that for all $j$, $|y_j|\leq \delta R$.  For $x_N$ in the support of $f_R$, we have that $|x_N|=(1+\theta)R$, for some $\theta\in [\frac14,\frac34]$. Therefore, 
\begin{equation}\label{boundattr}
 \frac{1}{|x_N-y_j|}\geq\frac{1}{|x_N|+|y_j|}\geq\frac{1}{(1+\h + \delta)R}.
\end{equation}

Similarly, for $(x_1,...,x_{N-1})$ in the support of $\chi_{\delta R}^{\otimes N-1}\Psi_{N-1}$, we have $|x_j|<\delta R$  for all $j=1,...,N-1$. Therefore, 
\begin{equation}\label{boundrep}
\frac{1}{|x_N-x_j|}\leq\frac{1}{|x_N|-|x_j|}\leq\frac{1}{(1 +\h -\delta) R}.
\end{equation}
Using \eqref{boundattr} and \eqref{boundrep} and recalling \eqref{def:PhiNR} we find that, on $\supp \Phi_{N,R}$,
\begin{align*}
 -\sum_{j=1}^M \frac{ Z_j }{|x_N-y_j|}+\sum_{j=1}^{N-1}\frac{1}{|x_N-x_j|}\leq 
 \frac{(1+\h)(N-1-|\cZ|) + \delta(N-1+|\cZ|)}{(1 +\h + \delta)(1 +\h -\delta)R}.
\end{align*}
As a consequence, using \eqref{deltachoice}, we find, on $\supp \Phi_{N,R}$,
\begin{equation}\label{bound:pot}
 -\sum_{j=1}^N\frac{ Z_j }{|x_N-y_j|}+\sum_{j=1}^{N-1}\frac{1}{|x_N-x_j|}\leq \frac{(\frac{1}{2}+\h)(N-1-|\cZ|)}{(1 +\h + \delta)(1 +\h -\delta)R} \leq \frac{N-1-|\cZ|}{8R},
\end{equation}
where in the last step we used that $\frac{1}{2} + \h \geq \frac12$ and that
$(1+\h +\delta)(1+\h-\delta) \leq 4$ for all $\h \in [1/4,3/4]$ as well as the assumed inequality  $N-1-|\cZ|<0$.

Inserting  first~ $\langle f_R, (-\Delta_{x})f_R \rangle= \frac{1}{R^2}\langle f_1, (-\Delta_{x})f_1 \rangle$ 
	 and~\eqref{bound:pot} into~\eqref{HNHN-1} and then combining it with \eqref{diag0} and ~\eqref{split}, we find that
\begin{equation}\label{1ev}
	\langle(\chi_{\delta R}^{\otimes N-1}\Psi_{N-1}\otimes f_R),(H_{N}-E_{N-1})(\chi_{\delta R}^{\otimes N-1}\Psi_{N-1}\otimes f_R)\rangle  \leq -\frac{\alpha(|\cZ|+1-N)}{8R}+\frac{D_1}{R^2},\end{equation}
for some constant $D_1>0$.

Using~\eqref{def:PhiNR}, \eqref{norm:chiPsi} and~\eqref{1ev}, we find that there exists $C>0$ such that, for $R$ large enough,
\begin{equation}\label{ineq:HEPhi}
	\langle \Phi_{N,R},(\hat{H}_N-E_{N-1})\Phi_{N,R}\rangle\leq-\frac{C}{R}+\frac{D_1}{R^2},
\end{equation}
which implies that $\hat{H}_N$ has an eigenvalue below its essential spectrum since $E_{N-1}=\inf\sigma_{ess}(\hat{H}_N)$.


Let $n\in\N$.  We will prove that $\hat{H}_N-E_{N-1}$ has at least $n$ negative eigenvalues.  As far as we know, this part of the proof is new, even for atoms, since~\cite{bhv} does not consider this question. Let us consider a normed linear combination of $\Phi_{N,2^k}$'s with $k$ between $2n+1$ and $3n$: \begin{equation}\label{formXi}\Xi:=\sum_{i=1}^{n}c_i\Phi_{N,2^{2n+i}}, \quad \text{with} \quad \sum_{i=1}^n|c_i|^2=1.
\end{equation}
 Then, 
\begin{multline}\label{decompXi}\langle\Xi,(\hat{H}_N-E_{N-1})\Xi\rangle=\sum_{i=1}^n|c_i|^2
	\langle\Phi_{N,2^{2n+i}},(\hat{H}_N-E_{N-1})\Phi_{N,2^{2n+i}}\rangle
\\	+\sum_{i\neq j} \overline{c_i} c_j \langle\Phi_{N,2^{2n+i}},(\hat{H}_N-E_{N-1})\Phi_{N,2^{2n+j}}\rangle.
\end{multline}

Using~\eqref{ineq:HEPhi}, we find that, if $n$ is large enough,
\begin{align}\label{negdiagterms}
	\sum_{i=1}^n|c_i|^2
	\langle\Phi_{N,2^{2n+i}},(\hat{H}_N-E_{N-1})\Phi_{N,2^{2n+i}}\rangle \leq-\sum_{i=1}^n|c_i|^2\frac{C}{2^{2n+i}}+\frac{D_1}{4^{2n}}\leq -\frac{C}{8^n}+\frac{D_1}{16^n}.
\end{align}

For the { terms where $i\neq j$, we know that the functions $\Phi_{N,2^{2n+i}}$ and $\Phi_{N,2^{2n+j}}$ have disjoint supports in the $N$-th variable. As a consequence, the only parts which are not 0 are of the form
\[\langle\Phi_{N,2^{2n+i}},T_N \Phi_{N,2^{2n+j}}\rangle=\langle f_{2^{2n+i}},T f_{2^{2n+j}}\rangle.\]
 Since the supports of $f_{2^{2n+i}}$ and $f_{2^{2n+j}}$ have distance at least $2^{2n}$, we get from Lemma~\ref{expdecay:realT} that there exists $C>0$ such that\begin{equation}\label{smallnondiag} \langle\Phi_{N,2^{2n+i}}, T_N \Phi_{N,2^{2n+j}}\rangle\leq Ce^{-2^{2n}}.\end{equation}
Inserting~\eqref{negdiagterms} and~\eqref{smallnondiag} into~\eqref{decompXi}, we get that \[\langle\Xi,(\hat{H}_N-E_{N-1})\Xi\rangle<0\] for all such $\Xi$ in the form~\eqref{formXi}, given that $n$ is large enough.

The linear space spanned by the $\phi_{N,2^k}$ for $k$ between $2n+1$ and $3n$ is thus an $n$-dimensional space on which the quadratic form associated with $\hat{H}_N-E_{N-1}$ is negative. Therefore, $\hat{H}_N$ has at least $n$ eigenvalues below its essential spectrum. Since this is true for arbitrarily large $n$, there are infinitely many eigenvalues below the bottom of the essential spectrum.
\qed
\section*{Acknowledgements} 
 The research of all authors was funded by the Deutsche Forschungsgemeinschaft (DFG, German Research Foundation) - Project-ID 258734477  - SFB 1173. The one of S.Z. was also  supported by the Basque Government through the BERC 2022-2025 program and by the Ministry of Science and Innovation: BCAM Severo Ochoa accreditation SEV-2017-0718 and PID2020-112948GB-I00 funded by MCIN/AEI/10.13039/501100011033 and by "ERDF A way of making Europe". The research of M.O. was partially supported by the grant 0135-00166B from the Independent Research Fund Denmark.
  I.A. is grateful to  Mathieu Lewin for numerous inspiring discussions  on the problem of isomerizations.  He also acknowledges discussions with Roland Schnaubelt, which improved the presentation of the Introduction and of Section 1. We thank Semjon Vugalter and Jean-Marie Barbaroux for suggesting to us the problem of extending the results of \cite{al} in the semirelativistic case. All authors are grateful to Dirk Hundertmark for communicating  that the results of \cite{bhv} work if the irreducibility assumption is replaced by the assumption that all ground states of all atoms are spherically symmetric and for discussing the importance of this. This was an inspiration for us to relax the irreducibility assumption, as well.
\bigskip

\bibliographystyle{elsarticle-num.bst}

\end{document}